\documentclass[journal,a4paper]{IEEEtran}
\pdfoutput=1

\usepackage[utf8]{inputenc}
\usepackage{amsfonts}
\usepackage{mathtools}
\usepackage{xspace}
\usepackage{amsthm}
\usepackage{mathrsfs}
\usepackage{pgfplots}
\usetikzlibrary{plotmarks}
\definecolor{darkgreen}{rgb}{0,0.5,0}
\usepackage[colorlinks=true,citecolor=darkgreen,linkcolor=darkgreen]{hyperref}
\usepackage[ruled,vlined,titlenumbered,linesnumbered]{algorithm2e}
\usepackage{multirow}
\usepackage[scaled=0.9]{helvet} %

\usepackage{varioref}
\usepackage{tabularx}
\usepackage{enumitem}
\newcolumntype{L}[1]{>{\raggedright\arraybackslash}p{#1}} %
\newcolumntype{C}[1]{>{\centering\arraybackslash}p{#1}} %
\newcolumntype{R}[1]{>{\raggedleft\arraybackslash}p{#1}} %

\newtheorem{theorem}{Theorem}
\newtheorem{definition}[theorem]{Definition}
\newtheorem{lemma}[theorem]{Lemma}

\newtheorem{problem}[theorem]{Problem}
\newtheorem{remark}[theorem]{Remark}
\newtheorem{proposition}[theorem]{Proposition} 
\newtheorem{example}[theorem]{Example}

\newcommand{\Z}[1]{\ensuremath{\mathbb{Z}_{#1}}} %
\newcommand{\Fqm}{\ensuremath{\mathbb F_{q^m}}}

\newcommand{\Fq}{\ensuremath{\mathbb F_{q}}}
\newcommand{\K}{\ensuremath{\mathbb K}}

\newcommand{\module}[1]{\ensuremath{\mathfrak{M}(#1)}}

\newcommand{\spannedBy}[1]{\ensuremath{\left\langle #1\right\rangle_q}}

\newcommand{\Pade}{Pad\'e\xspace}

\newcommand{\degConstraint}{\ensuremath{D}}

\newcommand{\SkewPolys}{\Fqm[x;\sigma]}

\renewcommand{\bar}{\overline}

\def\vec#1{{\mathchoice{\mbox{\boldmath$\displaystyle #1$}}%
		{\mbox{\boldmath$\textstyle #1$}}%
		{\mbox{\boldmath$\scriptstyle #1$}}%
		{\mbox{\boldmath$\scriptscriptstyle #1$}}}}
\newcommand{\mat}[1]{\vec{#1}}
\newcommand{\Mat}[1]{\vec{#1}}

\renewcommand{\a}{\vec{a}}
\renewcommand{\b}{\vec{b}}
\renewcommand{\c}{\vec{c}}
\newcommand{\e}{\vec{e}}
\newcommand{\f}{\vec{f}}
\newcommand{\g}{\vec{g}}
\newcommand{\h}{\vec{h}}
\renewcommand{\r}{\vec{r}}
\newcommand{\s}{\vec{s}}
\newcommand{\p}{\vec{p}}
\newcommand{\q}{\vec{q}}
\renewcommand{\t}{\vec{t}}
\renewcommand{\u}{\vec{u}}

\renewcommand{\v}{\vec{v}} %
\newcommand{\w}{\vec{w}}
\newcommand{\x}{\vec{x}}
\newcommand{\y}{\vec{y}}
\newcommand{\z}{\vec{z}}

\newcommand{\A}{\Mat{A}}
\newcommand{\B}{\Mat{B}}
\newcommand{\C}{\Mat{C}}
\newcommand{\D}{\Mat{D}}
\newcommand{\E}{\Mat{E}}

\renewcommand{\P}{\Mat{P}}
\newcommand{\I}{\Mat{I}}

\newcommand{\Q}{\Mat{Q}}
\newcommand{\M}{\Mat{M}}
\newcommand{\R}{\Mat{R}}

\newcommand{\U}{\Mat{U}}

\renewcommand{\Z}{\Mat{Z}}
\newcommand{\0}{\Mat{0}}

\newcommand{\Subspacedist}[1]{d_s(#1)}

\newcommand{\myspace}[1]{\mathcal{#1}}

\newcommand{\oh}[1]{\bnd{O}{#1}}

\newcommand{\bnd}[2]{\ensuremath{#1\mathopen{}\left(#2\right)\mathclose{}}}

\newcommand{\nTransmit}{\ensuremath{{n_{t}}}}
\newcommand{\nReceive}{\ensuremath{n_{r}}}
\newcommand{\insertions}{\ensuremath{\gamma}}
\newcommand{\deletions}{\ensuremath{\delta}}

\newcommand{\intOrder}{\ensuremath{\ell}}

\newcommand{\modl}{\; \mathrm{mod}_\mathrm{l} \;}
\newcommand{\modr}{\; \mathrm{mod}_\mathrm{r} \;}
\newcommand{\reml}{\; \mathrm{rem}_\mathrm{l} \;}
\newcommand{\remr}{\; \mathrm{rem}_\mathrm{r} \;}
\newcommand{\rem}{\; \mathrm{rem} \;}

\newcommand{\myref}[2]{\hyperref[#2]{#1~\ref{#2}}}

\newcommand{\ZZ}{\mathbb{Z}}
\newcommand{\khat}{\hat{k}}

\newcommand{\CIGab}{\mathcal{IC}_\mathrm{Gab}}
\newcommand{\LCIGab}{\mathcal{LIC}_\mathrm{Gab}}
\renewcommand{\l}{\vec{l}}

\makeatletter
\newcommand{\removelatexerror}{\let\@latex@error\@gobble}
\makeatother

\definecolor{darkred}{rgb}{0.5,0,0}

\DeclareMathOperator{\cdeg}{cdeg}

\newcommand{\OMul}[1]{\mathcal{M}_{q,m}(#1)}
\newcommand{\Module}{\mathcal{R}}

\newcommand{\IntParam}{\ell'}

\newcommand{\lambdaVec}{\vec{\lambda}}

\newcommand{\vhat}{\hat{\v}}
\newcommand{\lambdaVechat}{\hat{\lambdaVec}}
\newcommand{\smallsum}{{\textstyle\sum}}
\newcommand{\remev}[2]{{#1}\!\left[#2\right]}
\newcommand{\Bset}{\mathcal{B}}

\newcommand{\Prk}{\mathrm{Prk}}
\newcommand{\IPop}[1]{\mathcal{I}_{#1}^{\mathrm{op}}}
\newcommand{\IPrem}[1]{\mathcal{I}_{#1}^{\mathrm{rem}}}
\newcommand{\wtB}{\mathrm{wt}_{\Bset}}
\newcommand{\dB}{\mathrm{d}_{\Bset}}
\newcommand{\MSPop}{\mathcal{M}^{\mathrm{op}}}
\newcommand{\MSPrem}{\mathcal{M}^{\mathrm{rem}}}
\newcommand{\rdeg}{\mathrm{rdeg}}
\newcommand{\llcm}{\mathrm{llcm}}
\newcommand{\OEF}{\mathcal{F}}
\newcommand{\Gal}{\mathrm{Gal}}
\renewcommand{\k}{\vec{k}}
\newcommand{\Qspace}{\mathcal{Q}}
\newcommand{\wtR}{\mathrm{wt}_\mathrm{R}}
\newcommand{\dR}{\mathrm{d}_\mathrm{R}}
\newcommand{\wtSR}{\mathrm{wt}_{\mathrm{SR},\vec{n}}}
\newcommand{\dSR}{\mathrm{d}_{\mathrm{SR},\vec{n}}}
\newcommand{\Jset}{\mathcal{J}}
\newcommand{\Iset}{\mathcal{I}}

\newcommand{\softO}{\tilde{O}} %

\newcommand{\MABin}{\A}
\newcommand{\MABinhat}{\hat{\A}}
\newcommand{\MABout}{\B}
\newcommand{\MABouthat}{\hat{\B}}
\newcommand{\MABoutentry}{B}

\renewcommand{\L}{\vec{L}}

\newcommand{\MABnameFull}[3]{$#2$-ordered weak-Popov approximant basis of $#1$ of order $#3$}
\newcommand{\MABnameFullStandard}{\MABnameFull{\MABin}{\s}{d}}
\newcommand{\RMABnameShort}[3]{\mathsf{owPopovApprox}_{\mathsf{R}}(#1,#2,#3)}
\newcommand{\RMABnameShortStandard}{\RMABnameShort{\MABin}{\s}{d}}
\newcommand{\LMABnameShort}[3]{\mathsf{owPopovApprox}_{\mathsf{L}}(#1,#2,#3)}
\newcommand{\LMABnameShortStandard}{\LMABnameShort{\MABin}{\s}{d}}

\SetKwComment{Comment}{}{}
\newcommand{\myAlgoComment}[1]{\Comment{\normalfont // #1}}

\newcommand{\Frob}[1]{\phi_{#1}}
\newcommand{\Frobq}{\Frob{q}}

\begin{document} 
\title{Fast Decoding of Codes in the Rank, Subspace, and Sum-Rank Metric}
\author{%
  Hannes Bartz,~\IEEEmembership{Member,~IEEE,}
  Thomas Jerkovits,~\IEEEmembership{Member,~IEEE,}
  Sven Puchinger,~\IEEEmembership{Member,~IEEE,}
  Johan Rosenkilde
  \thanks{Parts of this paper have been presented at the \emph{2019 IEEE Information Theory Workshop (ITW)} \cite{bartz2019fast}.}
  \thanks{H.~Bartz and T.~Jerkovits are with the Institute of Communications and Navigation, German Aerospace Center (DLR), Germany (e-mail: \{hannes.bartz, thomas.jerkovits\}@dlr.de). T.~Jerkovits is also with the Institute for Communications Engineering, Technical University of Munich (TUM), Germany.}
  \thanks{S.~Puchinger is with the Department of Electrical and Computer Engineering, Technical University of Munich, 80333 Munich, Germany (e-mail: sven.puchinger@tum.de). This work was partly done while he was with the Department of Applied Mathematics and Computer Science, Technical University of Denmark (DTU), 2800 Kongens Lyngby, Denmark.}
  \thanks{J.~Rosenkilde is with GitHub Denmark Aps, 2100 Copenhagen (email: jsrn@jsrn.dk). This work was
done while he was with the Department of Applied Mathematics and
Computer Science, Technical University of Denmark (DTU), 2800 Kongens Lyngby, Denmark.}
  \thanks{S.~Puchinger has received funding from the European Union's Horizon 2020 research and innovation program under the Marie Sklodowska-Curie grant agreement no.~713683, and from the German Israeli Project Cooperation (DIP) grant no.~KR3517/9-1.}
  \thanks{©2021 IEEE. Personal use of this material is permitted. Permission from IEEE must be obtained for all other uses, in any current or future media, including reprinting/republishing this material for advertising or promotional purposes, creating new collective works, for resale or redistribution to servers or lists, or reuse of any copyrighted component of this work in other works.}
}

\maketitle

\begin{abstract}
We speed up existing decoding algorithms for three code classes in different metrics: interleaved Gabidulin codes in the rank metric, lifted interleaved Gabidulin codes in the subspace metric, and linearized Reed--Solomon codes in the sum-rank metric.
The speed-ups are achieved by new algorithms that reduce the cores of the underlying computational problems of the decoders to one common tool: computing left and right approximant bases of matrices over skew polynomial rings.
To accomplish this, we describe a skew-analogue of the existing \textsf{PM-Basis} algorithm for matrices over ordinary polynomials.
This captures the bulk of the work in multiplication of skew polynomials, and the complexity benefit comes from existing algorithms performing this faster than in classical quadratic complexity.
The new algorithms for the various decoding-related computational problems are interesting in their own and have further applications, in particular parts of decoders of several other codes and foundational problems related to the remainder-evaluation of skew polynomials.

\end{abstract}


\IEEEpeerreviewmaketitle

\section{Introduction}

\noindent
We consider algorithms for decoding certain codes in three different metrics -- rank, subspace and sum-rank metric -- all of which arise as evaluation-like codes of skew polynomials.
Skew polynomials are non-commutative polynomials, where the right multiplication of a scalar $\alpha \in \Fqm$ and the indeterminate $x$ is given as $x \alpha = \sigma(\alpha)x$, where $\sigma$ is an automorphism of $\Fqm$.
The ring of these polynomials is denoted $\SkewPolys$; see Section~\ref{ssec:skew_polys} for the formal definition.

We consider existing decoding principles for the codes and show for each how to speed it up by reducing the core computation to an \emph{approximant basis} computation of matrices over the relevant skew polynomial ring.
A reduction to a similar problem for matrices over ordinary polynomial rings has proved beneficial in speeding up decoding of a number of evaluation codes in the Hamming metric and its soft relaxations \cite{jeannerod_computing_2017}.
Given a matrix $\MABin \in \SkewPolys^{a \times b}$ and an ``order'' $d \in \ZZ_{\geq 0}$, a left approximant basis is a matrix $\MABout \in \SkewPolys^{a \times a}$ such that $\MABout \MABin \equiv 0 \modr x^d$ (congruence right-modulo $x^d$, cf.~Section~\ref{ssec:skew_polys}), and such that $\MABout$ is in a certain normal form while satisfying that any vector $\vec b \in \SkewPolys^{1 \times a}$ such that $\vec b \MABin \equiv 0 \modr x^d$ is in the left $\SkewPolys$-row space of $\MABout$, see Section~\ref{sec:order_bases}.
An analogous definition is given for right approximant bases.
Approximant bases for skew polynomials (more generally, for Ore polynomials) were introduced in \cite{beckermann_fraction-free_2002} (under the name ``order basis'').

\subsection{Main Results}

\begin{table*}
  \caption{%
    Overview of new decoding speeds. Parameters: code length $n$, interleaving parameter $\ell$ (usually $\ell \ll n$).
  For subspace codes, $\nTransmit$ resp.~$\nReceive$ is the dimension of the transmitted resp.~received subspace.
  $\OMul{n}$ is the cost (in operations in $\Fqm$ or $\Fq$) of multiplying two skew-polynomials of degree at most $n$ and $\omega$ is the matrix multiplication exponent, see Sections~\ref{ssec:cost_model} and \ref{ssec:cost_skew_poly_operations}.
}\label{tab:overview_decoders}
\begin{minipage}{\textwidth}
\begin{center}
\newcommand{\specialcell}[2][c]{{\def\arraystretch{1.2}\begin{tabular}[#1]{@{}l@{}}#2\end{tabular}}}
\newcommand{\specialcellsmalldistleftalign}[2][c]{{\def\arraystretch{1}\begin{tabular}[#1]{@{}l@{}}#2\end{tabular}}}
\newcommand{\specialcellcenter}[2][c]{{\def\arraystretch{1}\begin{tabular}[#1]{@{}c@{}}#2\end{tabular}}}
\def\arraystretch{2.0}
\begin{tabular}{l|l|l|l|l|l}
Metric & Code Class & \specialcellsmalldistleftalign{Previously Fastest \\ Decoder (over $\Fqm$)} & \specialcell{Considered Decoder \& \\ Complexity (over $\Fqm$)} & \specialcellsmalldistleftalign{Our Complexity (over the base field\\ of the cost bound $\OMul{n}$)} & Reference \\
\hline \hline
Rank & \specialcellsmalldistleftalign{Interleaved \\ Gabidulin} & $\softO(\ell^\omega \OMul{n})$ \hfill \cite{sidorenko2014fast} & $O(\ell^2 n^2)$ \cite{wachter2014list} & $\softO(\ell^\omega \OMul{n})$ & \specialcellsmalldistleftalign{Theorem~\ref{thm:rank_complexity_summary} \\ Section~\ref{sec:rank_and_subspace}} \\ 
\hline
Subspace & \specialcellsmalldistleftalign{Lifted Interleaved \\ Gabidulin} & $O(\ell^2 \max\{\nTransmit,\nReceive\}^2)$ \hfill \cite{BartzWachterZeh_ISubAMC} & see previously fastest 
& \specialcellsmalldistleftalign{$\softO\!\left(\intOrder^\omega\OMul{\max\{\nTransmit,\nReceive\}} \right)$ \\ plus $O(\intOrder m \nReceive^{\omega-1})$ operations in $\Fq$}
& \specialcellsmalldistleftalign{Theorem~\ref{thm:subspace_complexity_summary} \\ Section~\ref{sec:rank_and_subspace}} \\
\hline
\specialcellsmalldistleftalign{Sum-Rank/ \\ Skew} & \specialcellsmalldistleftalign{Linearized/Skew \\ Reed--Solomon} & $O(n^2)$ \hfill \cite{martinez2019reliable} & see previously fastest &  $\softO(\OMul{n})$ & \specialcellsmalldistleftalign{Theorem~\ref{thm:sum-rank_summary} \\ Section~\ref{sec:sum-rank}} %
\end{tabular}
\end{center}
\end{minipage}
\end{table*}

\begin{table*}
  \caption{%
    Overview of computational tools used to achieve faster decoding algorithms with the complexity of existing algorithms and the proposed ones.
    We indicate the metric which the computational problem is a priori relevant for (R$=$rank, S$=$subspace, and Sr$=$sum-rank metric), and indicate other potential applications discussed in Section~\ref{ssec:further_applications}.
    For $\OMul{n}$ and $\omega$, see Table~\ref{tab:overview_decoders} above.
}\label{tab:overview_tools}
\begin{minipage}{\textwidth}
\begin{center}
\def\arraystretch{1.5}
\setlength{\tabcolsep}{3pt}
\newcommand{\specialcell}[2][c]{{\def\arraystretch{1.2}\begin{tabular}[#1]{@{}l@{}}#2\end{tabular}}}
\newcommand{\specialcellsmalldistleftalign}[2][c]{{\def\arraystretch{1}\begin{tabular}[#1]{@{}l@{}}#2\end{tabular}}}
\newcommand{\specialcellcenter}[2][c]{{\def\arraystretch{1}\begin{tabular}[#1]{@{}c@{}}#2\end{tabular}}}
\begin{tabular}{L{5.6cm}|L{3.4cm}|L{3.7cm}|L{0.2cm}|L{0.2cm}|L{0.2cm}|L{3cm}}
Computational Problem
	& Previous Complexity (over $\Fqm$)
	& Our Complexity (over the base field of $\OMul{n}$)
	& R
	& S
	& Sr
	& Further Applications \\
	\hline \hline
Computation of a right/left $\s$-ordered weak-Popov approximant basis of order $d$ of an $a \times b$ skew-polynomial matrix (Definition~\ref{def:minimal_approximant_basis})
	& $O(a^3 b^2 d^2)$ \cite{beckermann_fraction-free_2002} (left case only)
	& Left/right case, respectively: $\softO\big(a^{\omega-1}\max\{a,b\} \OMul{d}\big)$, $\softO\big(\max\{a,b\}b^{\omega-1} \OMul{d}\big)$
(Theorem~\ref{thm:correctness_DaCApp} in Section~\ref{ssec:fast_order_bases_computation})
	& \textsf{X}
	& \textsf{X}
	& \textsf{X}
	&  \\
	\hline
Vector Operator Interpolation (Problem~\ref{prob:general_interpolation_problem}) with $n$ interpolation points (vectors in $\Fqm^{\ell+1}$) and degree constraint $D$ \hspace{2cm} (complexities given for $D \in \Theta(n)$).
	& $O(\ell^2 n^2)$ \cite{xie_linearized_2013}, $\softO(\ell^3 \OMul{\ell n})$ on special input \cite{alekhnovich_linear_2005} %
	& $\softO(\ell^\omega \OMul{n})$ \hspace{1cm} plus, under some conditions, $O(\intOrder m n^{\omega-1})$ operations in $\Fq$ (Theorem~\ref{thm:fast_interpolation_correctness_complexity} in Section~\ref{ssec:interpolation_speed-up}) %
	& \textsf{X}
	& \textsf{X}
	& 
	& Interpolation step of decoding Mahdavifar-- Vardy and (lifted) folded Gabidulin.\\
	\hline
Vector Root Finding (Problem~\ref{prob:general_root-finding_problem}) for a set of $\ell' \leq \ell+1$ skew polynomial vectors of dimension $\ell+1$, degree at most $n$, with degree constraints $k^{(1)},\dots,k^{(\ell)}$ \hspace{2cm} (complexities given for $\max_i k^{(i)} \in \Theta(n)$)
	& $O(\ell^3 n^2)$ %
	\cite{wachter2014list}, \hspace{1cm} $O(\ell^2 n^2)$ %
	on special input \cite{BartzWachterZeh_ISubAMC}
	& $\softO\!\left( \ell^\omega \OMul{n} \right)$ \hspace{0.5cm} %
	(Theorem~\ref{thm:fast_root_finding_correctness_complexity} in Section~\ref{ssec:root_finding_speed-up})
	& \textsf{X}
	& \textsf{X}
	& 
	&  \\
	\hline
Remainder-Evaluation Operations (Problem~\ref{prob:remainder_arithmetic}):
annihilator polynomial computation, multi-point evaluation, and interpolation (number of points and polynomial degrees $\leq n$) of skew polynomials w.r.t.\ the remainder evaluation.
	& $O(n^2)$ \cite{martinez2019reliable}
	& $\softO(\OMul{n})$ (Theorems~\ref{thm:fast_remainder_MPE}--\ref{thm:fast_remainder_interpolation}, Section~\ref{ssec:fast_remainder_ev_arith})
	& 
	& 
	& \textsf{X}
	& Encoding linearized/skew Reed--Solomon codes. Repair in the locally repairable / PMDS codes in \cite{martinez2019universal}. \\
	\hline
$2$D Vector Remainder Interpolation (Problem~\ref{prob:bivariate_remainder_interpolation}) with $n$ interpolation points (vectors in $\Fqm^{2}$).
	& $O(n^2)$ \cite{martinez2019reliable}
	& $\softO(\OMul{n})$ \hspace{0.5cm} (Theorem~\ref{thm:bivariate_remainder_interpolation_summary} in Section~\ref{ssec:bivariate_remainder_interpolation})
	& 
	& 
	& \textsf{X}
	&  \\

\end{tabular}
\end{center}
\end{minipage}
\end{table*}

\noindent
Our central computational result (Theorem~\ref{thm:correctness_DaCApp}) is an algorithm for computing a right or left minimal approximant basis of an $a \times b$ matrix of order $d$, whose complexity's dependency on the order $d$ is only $\OMul{d}$ (see Table~\ref{tab:overview_tools} for more details), where $\OMul{d}$ is the cost of multiplying two skew polynomials of degree at most $d$ (see Section~\ref{ssec:cost_model}), given in operations in some field (e.g., $\Fqm$ or $\Fq$).
The algorithm is a right (resp.~left) adaptation of the \textsf{PM-Basis} algorithm for computing minimal approximants over ordinary polynomial rings \cite{giorgi_complexity_2003}.

In Sections~\ref{sec:rank_and_subspace} and \ref{sec:sum-rank}, we provide new speed records for decoding certain codes in the rank, subspace, and sum-rank metric; see Table~\vref{tab:overview_decoders} for a summary.

Each of these speed records are achieved by replacing the bottleneck computations in an existing decoding principle with a left or right minimal approximant basis.
To enable these results, we give fast algorithms for a number of decoding-related computational problems which we believe may be interesting in their own right.
Most of these new algorithms rely on fast computation of approximant bases.
See Table~\vref{tab:overview_tools} for an overview of these problems.

During the revision of this paper, we became aware of the preprint \cite{caruso2019residues} by Caruso. Using different techniques, he obtains a decoding algorithm for linearized Reed--Solomon codes and algorithms for remainder-evaluation operations with the same complexity as ours.

\subsection{The Studied Codes and Their History}\label{ssec:intro_codes}
\noindent
\emph{Rank-metric codes} are sets of matrices whose distance is measured by the rank of their difference.
These codes and their most famous subclass, Gabidulin codes, were independently introduced in \cite{Delsarte_1978,Gabidulin_TheoryOfCodes_1985,Roth_RankCodes_1991}.
By now, applications of rank-metric codes abound and include criss-cross error correction in memory chips, space-time codes for MIMO systems, code-based cryptography, network coding, distributed data storage, and digital watermarking.

Interleaved Gabidulin codes are direct sums of $\ell$ Gabidulin codes of the same length over an extension field $\Fqm$: codewords can be represented as an $\Fq^{\ell m \times n}$ matrix by stacking Gabidulin codewords as $\Fq^{m \times n}$ matrices.
If such a matrix is subjected to a random error with a low $\Fq$-rank, we can correct that error with high probability even if the rank exceeds half the minimum distance of the constituent Gabidulin code.
The downside is the rectangular shape of the codewords (since $n \leq m$).
Besides being suitable for any application of rank-metric codes with such a rectangular codeword shape, interleaved Gabidulin codes have been explicitly used in works on network coding \cite{silva2008rank,sidorenko2010decoding} and code-based cryptography \cite{faure2006new,overbeck2007public}.

There are several known polynomial-time decoding algorithms for $\ell$-interleaved Gabidulin codes of length $n$.
All of these algorithms correct up to roughly $\tfrac{\ell}{\ell+1}(n-\bar{k}+1)$ errors, where $\bar{k} := \tfrac{1}{\ell}\sum_i k_i$ is the mean of the dimensions $k_i$ of the constituent Gabidulin codes.
The first-known decoder is due to Loidreau and Overbeck \cite{loidreau2006decoding}.
It is a \emph{partial unique decoder}, which means that for error weights beyond half the minimum distance, it either returns a unique decoding result or fails.
The algorithm is based on solving a linear system of equations and has complexity $O(\ell n^\omega)$. %
Loidreau and Overbeck also derived an upper bound on the relative number of errors of rank $t$ for which the decoder fails. For $t\geq \ell$, it decays exponentially in $m(t-\ell)$.
Sidorenko and Bossert \cite{sidorenko2010decoding} proposed a partial unique decoder for interleaved Gabidulin codes that solves a syndrome key equation.
The algorithm can be implemented in $O(\ell n^2)$ operations in $\Fqm$ using a Berlekamp--Massey-like algorithm \cite{sidorenko2011skew} or the demand-driven row reduction algorithm in \cite{puchinger2017row}.
There is also a divide-\&-conquer approach \cite{sidorenko2014fast} that solves the key equation in $\softO(\ell^\omega \OMul{n})$ operations in the base field of the cost bound $\OMul{n}$.

In this paper, we consider the interpolation-based decoder by Wachter-Zeh and Zeh \cite{wachter2014list},
which returns a list of all codewords within a decoding radius less than $\tfrac{\ell}{\ell+1}(n-\bar{k}+1)$.
It can also be seen as a partial unique decoder by declaring a decoding failure if this list is greater than $1$.
Such a failure event occurs at most in those cases in which the Loidreau--Overbeck decoder fails (see \cite[Lemma~8]{wachter2014list}).
The algorithm consists of an \emph{interpolation step} and a \emph{root-finding step} and has complexity $O(\ell^3n^2)$ operations in $\Fqm$. 
If there is a unique solution to the decoding problem, then the complexity can be reduced to $O(\ell^2n^2)$ \cite{BartzWachterZeh_ISubAMC}.

\emph{Subspace codes} are sets of subspaces of a given vector space that have distance properties w.r.t.\ the \emph{subspace metric}~\cite{koetter2008coding}. Beside the initial application of subspace codes as linear authentication codes~\cite{wang2003linear}, subspace codes were proposed by Kötter and Kschischang for error correction in network coding~\cite{koetter2008coding}.
In (random) linear network coding, errors in the network may propagate through the network due to the linear combination of the incoming packets at intermediate nodes.
In particular, a single corrupted packet would in turn corrupt all later linear combinations which include this packet.
The main idea for subspace codes comes from the observation that the row space of transmitted packets is preserved by the linear operations at the intermediate nodes of the network, and few errors in the network result in a small subspace distance between transmitted and received subspace.
Besides the initial constructions of subspace codes based on Gabidulin codes, so called \emph{lifted Gabidulin} codes, in~\cite{koetter2008coding,silva2008rank,silva2009error}, variants with improved error-correction capabilities, including \emph{interleaved} lifted Gabidulin codes~\cite{BartzWachterZeh_ISubAMC}, were proposed.
The currently fasted decoding algorithms for lifted interleaved Gabidulin codes that attain the best decoding region are the syndrome-based approach from~\cite{bartz2017improved} which requires $\oh{\ell^3\nTransmit^3}$ operations in $\Fqm$ and the interpolation-based decoder from~\cite{BartzWachterZeh_ISubAMC} which requires $\oh{\ell^2\max\{\nTransmit,\nReceive\}^2}$ operations in $\Fqm$, where $\ell$ is the interleaving order and $\nTransmit$ and $\nReceive$ are the dimension of the received and transmitted space, respectively.
We improve the cost of the latter algorithm.

\emph{The sum-rank metric} is a family of metrics interpolating the Hamming and rank metric which was first introduced in \cite{nobrega2010multishot} as being suitable for multi-shot network coding.
There are several known codes designed for this metric: partial unit memory codes constructed from rank-metric codes \cite{wachter2011partial,wachter2012rank,wachter2015convolutional}, convolutional codes \cite{napp2017mrd,napp2018faster}, as well as linearized Reed--Solomon codes \cite{martinez2018skew}.
The latter codes can be seen as a combination of Reed--Solomon and Gabidulin codes, attain the Singleton bound in the sum-rank metric with equality, and are closely related to skew Reed--Solomon codes in the skew metric \cite{boucher2014linear,martinez2018skew}.

Linearized Reed--Solomon codes have recently shown to provide reliable and secure coding schemes for multi-shot network coding \cite{martinez2019reliable}.
Furthermore, there is a construction \cite{martinez2019universal} of locally repairable codes with maximal recoverability (also known as partial MDS codes) based on linearized Reed--Solomon codes, which attains the smallest known field size among all existing code constructions for a wide range of code parameters.

We are aware of two decoding algorithms for linearized and skew Reed--Solomon codes in the literature, both of which are variants of the Welch--Berlekamp decoder for Gabidulin codes \cite{loidreau2006welch}.
One is due to Boucher \cite{boucher2019algorithm} and has cubic complexity $O(n^3)$ over $\Fqm$ in the code length $n$. The other one is quadratic $O(n^2)$ over $\Fqm$ and was presented by Martínez-Peñas and Kschischang \cite{martinez2019reliable}.
Our work is based on the latter.

\subsection{History of Computational Tools}\label{ssec:intro_computational_techniques}

The history of approximant bases starts with matrices over ordinary polynomials $\K[x]$, for a field $\K$.
They are also known as ``minimal approximant bases'', ``order bases'', and ``$\sigma$-bases'', and arose as matrix generalizations of simultaneous and Hermite \Pade approximations through a range of papers in the 1990's, especially \cite{beckermann_uniform_1992,barel_general_1992,beckermann_uniform_1994}; the latter paper presents fairly efficient algorithms for computing approximant bases.
``Shifted'' approximant bases were also introduced in these papers.
An immediate application of an approximant basis of a matrix $\MABin$ is that a subset of its rows form a generating set for all small-degree vectors in the left (resp.~right) kernel of $\MABin$.
Several other computations on polynomial matrices can be reduced to approximant bases, e.g.~row reduced forms \cite{giorgi_complexity_2003,gupta_triangular_2012}; determinants \cite{giorgi_complexity_2003}; Popov and Hermite form \cite{neiger_fast_2016}; even more general approximations \cite{jeannerod_computing_2017}; full-rank bases  and unimodular completion \cite{zhou_unimodular_2014}; and kernel bases \cite{zhou_computing_2012}.
Computing a (left) approximant basis of $\MABin \in \K[x]^{a \times b}$ with $a \leq b$ in roughly the time it takes to multiply two $a \times a$ polynomial matrices together was given as the \textsf{PM-Basis} algorithm in \cite{giorgi_complexity_2003}.
For $a \gg b$, this cost can be improved using ``partial linearization'', see \cite{zhou_efficient_2012} for unshifted or slightly shifted matrices, and \cite{jeannerod_fast_2016,jeannerod_fast_2019} for the general case which requires many more tools.

The notion of approximant is based on ``row reducedness'', see e.g.~\cite{kailath_linear_1980}, which is a matrix over $\K[x]$ whose rows have smallest degree among matrices whose rows span the same $\K[x]$-module.
The Popov form is a row reduced form that is normalised to be canonical \cite{popov_properties_1970}, and the weak Popov form is stronger than a row reduced form, but weaker than the Popov form \cite{mulders2003lattice}.
It seems computationally somewhat more challenging to efficiently compute a reduced form of a matrix than to compute an approximant basis, and the fastest techniques we currently know in the commutative case effectively reduce the former to the latter \cite{giorgi_complexity_2003,neiger_fast_2016}.
Many problems in coding theory which can be solved by approximant bases can instead be solved by row reduction, see e.g.~\cite{nielsen2013list}.

Turning to the non-commutative case, then approximant bases for matrices over skew polynomials, or more generally Ore polynomials (see Section~\ref{ssec:skew_polys}), were introduced in \cite{beckermann_fraction-free_2002}.
That paper, as well as much other literature on computations on Ore polynomials, is concerned with the case where $\K$ is infinite so coefficient growth quickly becomes the computational bottleneck.
To address this, the algorithm of \cite{beckermann_fraction-free_2002} generalises ``fraction-free'' techniques from the commutative case \cite{beckermann_fraction-free_2000}.
When $\K$ is finite, this is however slower than the algorithms of \cite{beckermann_uniform_1994,giorgi_complexity_2003}.
In this paper we consider $\K = \Fqm$ and in particular generalize the algorithm of \cite{giorgi_complexity_2003}.
This turns out to be conceptually straightforward but rather technical.
We will also introduce both a left and a right version of the algorithm; the two cases are of course very similar but subtly different.

Row reducedness and Popov forms were introduced for skew polynomial matrices in \cite{beckermann_fraction-free_2002} using a fraction-free approach.
In \cite{puchinger2017row,puchinger2017alekhnovich} some of us were involved in generalizing the methods of \cite{mulders2003lattice,alekhnovich_linear_2005} which are more efficient when $\K = \Fqm$, and applied this to some of the same decoding problems that we address in the present paper; the algorithms of the present paper are all asymptotically more efficient, see Table~\ref{tab:overview_decoders}.

Besides approximant bases, we study several computational problems that are related to the considered decoding algorithms (see Table~\ref{tab:overview_tools}).

The interpolation and root-finding steps of the interpolation-based decoders in \cite{wachter2014list,BartzWachterZeh_ISubAMC} are instances of the following two computational problems (see Section~\ref{ssec:rank_subspace_computational_problems}):
1) the \emph{vector interpolation problem} (Problem~\ref{prob:general_interpolation_problem}) was first considered in \cite{xie_linearized_2013} to decode Gabidulin, lifted Gabidulin, and Mahdavifar--Vardy codes.
The relation to decoding interleaved Gabidulin codes was given in \cite{wachter2014list} and lifted interleaved Gabidulin codes in \cite{BartzWachterZeh_ISubAMC}.
The problem is also called \emph{bivariate interpolation} since its solutions can be seen as formal bivariate polynomials of bounded $y$-degree with skew-polynomial coefficients, and the problem statement requires these polynomials to satisfy an evaluation condition and degree bound.
Hence, it can be seen as the skew-polynomial analog of the Sudan decoder interpolation step.
2) the \emph{vector root-finding problem} (Problem~\ref{prob:general_root-finding_problem}) was first considered in \cite{wachter2014list} for decoding interleaved Gabidulin codes, where also the currently fastest algorithm was given. The problem was also studied in \cite{BartzWachterZeh_ISubAMC} for decoding lifted interleaved Gabidulin codes. The authors of \cite{BartzWachterZeh_ISubAMC} also present an algorithm that is faster if the solution space has cardinality~$1$.

The two core computational problems of algorithm in \cite{martinez2019reliable} for decoding linearized (or skew) Reed--Solomon codes in the sum-rank (or skew) metric are: 1) fast operations with skew polynomials w.r.t.\ to the remainder evaluation.
This type of evaluation was first studied in \cite{lam1985general,lam1988vandermonde}, it was first used to construct block codes in \cite{boucher2014linear}, and
the currently fastest algorithms to compute the relevant operations were given in \cite{martinez2019reliable}.
2) a $2$-dimensional vector remainder interpolation, which can be seen as the analog of the Welch--Berlekamp reconstruction problem for skew polynomials w.r.t.\ the remainder evaluation. This problem was first studied in \cite{liu2015construction}, and later in \cite{boucher2019algorithm,martinez2019reliable}. The currently fastest algorithm to solve this problem was proposed in \cite{martinez2019reliable}.

\subsection{Reader's Guide}

We set notation, define our cost model, and recall known results on skew polynomials in Section~\ref{sec:preliminaries}.
In Section~\ref{sec:order_bases}, we analyze left and right approximant bases over skew polynomial rings and propose new, faster, algorithms to compute them.
These results lay the foundation for the remainder of the paper, which discusses computational problems related to decoding rank-metric and subspace codes (Section~\ref{sec:rank_and_subspace}) as well as sum-rank-metric codes (Section~\ref{sec:sum-rank}).
These two sections are independent of each other.
Both of them start by a subsection that formally states the relevant computational problems (cf.~Table~\ref{tab:overview_tools}) and recalls their relation to the considered decoders.
The respective remaining subsections propose new algorithms to solve these computational problems.
We conclude the paper in Section~\ref{sec:conclusion}, including several remarks on generality, further applications of the results, and some open problems.
The appendix includes some extended results out of the main scope of the paper, as well as examples.

\section{Preliminaries}\label{sec:preliminaries}

Let $q$ be a prime power, $m$ be a positive integer, and denote by $\Fq$ and $\Fqm$ the finite field of size $q$ and $q^m$, respectively.
The field $\Fqm$ is an extension field of $\Fq$ of extension degree $m$ and hence also a vector space over $\Fq$ of dimension $m$.
The Galois group of the extension is cyclic and consists of the powers of the Frobenius automorphism $\Frobq : \Fqm \to \Fqm, \, \alpha \mapsto \alpha^q$, i.e., $\Gal(\Fqm/\Fq) = \{\Frobq^i \, : \, i=0,\dots,m-1\}$. The generators of the Galois group are the $\Frobq^i$ with $\gcd(i,m)=1$.

\subsection{Cost Model}\label{ssec:cost_model}

We use the big-O notation family to state asymptotic costs of algorithms, and $\softO{(\cdot)}$ which neglects logarithmic factors in the input parameter.
Furthermore, we express the cost of algorithms either in arithmetic operations in the field $\Fqm$ or over $\Fq$: here we include not only $+, -, \cdot$ and $/$, but also applications of a (specific) automorphism $\sigma \in \Gal(\Fqm/\Fq)$.
This is uncommon in the literature on computation at large, but has become standard for work on Gabidulin codes and related codes.
The basic reasoning is that if the extension $\Fqm : \Fq$ is built using a normal basis (see, e.g., \cite{gao1993normal}) %
and $\sigma = \Frobq^i$, then $\sigma(a)$ is simply the cyclic shift of $i$ positions of the vector description of $a$ over $\Fq$ in that basis.
However, multiplication is not a priori as efficient in normal bases as it is in power bases, and the complications arise when attempting requiring that all operations are fast simultaneously.
We let $\OEF(m)$ denote an upper bound on the cost of all of these operations in $\Fqm$ in terms of operations in $\Fq$.
Couveignes and Lercier \cite{couveignes2009elliptic} showed that it is possible to choose a basis such that $\OEF(m) \in \softO(m)$, and we will mostly assume such a basis.
In practice and for small $m$ it might well be faster to use either a power basis with $\OEF(m) \in \softO(m^2)$ (bottleneck being applications of $\sigma$) or a normal basis with $\OEF(m) \in O(m^2)$ (bottleneck being multiplication and division).

In cost bounds, we denote by $\omega$ the matrix multiplication exponent, i.e.~the infimum of values $\omega_0 \in [2; 3]$ such that there is an algorithm for multiplying $n \times n$ matrices over $\Fqm$ in $O(n^{\omega_0})$ operations in $\Fqm$.
The currently known best bound is $\omega < 2.37286$ \cite{le_gall_powers_2014}.

\subsection{Skew Polynomials}
\label{ssec:skew_polys}

In this paper, all codes and algorithms are defined over skew polynomials which are non-commutative polynomials and were introduced by Ore in \cite{ore1933theory}; for this reason they are also known as Ore polynomial rings.
The general construction over any field $\mathbb{K}$ uses an endomorphism $\sigma$ and a ``$\sigma$-derivation'' $\delta : \mathbb{K} \rightarrow \mathbb{K}$, and can be used for unifying theoretical and computational questions on linear differential equations, time-dependent systems and recursively defined sequences of numbers, see e.g. \cite{bronstein_introduction_1996}, sometimes in the specialisation of $D$-finiteness, see e.g. \cite{kauers_holonomic_2013}.

We will only use the specialisation where $\mathbb{K} = \Fqm$, $\sigma = \Frobq^i$ with $\gcd(i,m) = 1$ (i.e.~$\Gal(\Fqm/\Fq) = \langle\sigma\rangle$), and $\delta = 0$.
When $i = 1$, these rings are isomorphic to linearized polynomials, which were also introduced by Ore \cite{ore_special_1933}, and for $i > 1$ behave in much the same way.
Besides their applications in coding theory, these are studied in cryptography~\cite{faure2006new}, dynamical systems \cite{cohen2000dynamics}, and are of theoretical interest \cite{ore_special_1933,evans1992linearized,wu2013linearized}.
In the remainder of the paper, when we say ``skew polynomials'', we mean this restricted setting.
They are sometimes also called \emph{twisted polynomials} or \emph{$\sigma$-polynomials}.

A \emph{skew polynomial} (in our restricted setting) is then a formal polynomial sum $f = \sum_{i \geq 0} f_i x^i$, indexed by powers of an indeterminant $x$, and with only a finite number of $f_i \in \Fqm$ being non-zero.
We add two polynomials monomial-wise as for ordinary polynomials.
Multiplication of skew polynomials is defined by the rule
\begin{equation}
x \cdot a = \sigma(a) \cdot x \label{eq:skew_polynomials_multiplication_rule}
\end{equation}
for any $a \in \Fqm$.
By associativity and distributivity, we have
\begin{equation}
f \cdot g = \sum_{i \geq 0} \Big( \sum_{j \geq 0} f_j \sigma^j (g_{i-j}) \Big) x^i. \label{eq:skew_polynomials_multiplication}
\end{equation}
for any two skew polynomials $f = \sum_i f_i x^i$ and $g = \sum_j g_j x^j$, where we define $f_i = g_i = 0$ for $i < 0$.
The set of skew polynomials with this addition and multiplication rule is a non-commutative integral domain and denoted by $\SkewPolys$.

The \emph{degree} of a skew polynomial is defined by
\begin{equation*}
\deg f := \begin{cases}
\max\{i \, : \, f_i \neq 0\}, &\text{if } f \neq 0, \\
-\infty, &\text{otherwise.}
\end{cases}
\end{equation*}
As for ordinary polynomials, we have $\deg(f \cdot g) = \deg f + \deg g$, and $\deg(f + g) \leq \max\{\deg f, \, \deg g\}$, where equality holds in the latter iff $\deg f \neq \deg g$ or $\deg f = \deg g$ and the leading coefficients of $f$ and $g$ do not sum to zero.

There is both a left and right division algorithm, hence the ring is left and right Euclidean.
Let $f,g,h \in \SkewPolys$ such that $h \neq 0$.
We denote the remainder of the left division of $f$ by $h$ as $f \reml h$, i.e., $f \reml h$ is the unique skew polynomial of degree $< \deg h$ for which $f \reml h = f-h \chi$ for some $\chi \in \SkewPolys$.
Analogously, the remainder w.r.t.\ the right division is denoted by $f \remr h$ (in this case we have $f \remr h = f- \chi h$ for some $\chi \in \SkewPolys$).
We say that $f$ and $g$ are congruent left-modulo $h$, written $f \equiv g \modl h$, if $f-g$ is divisible by $h$ from the left (i.e., $(f-g) \reml h = 0$).
Likewise, %
$f \equiv g \modr h$ if $(f-g) \remr h = 0$.

Since $\SkewPolys$ is left and right Euclidean, it is also a left and right principal ideal domain.
This implies that left and right modules over $\SkewPolys$ share many important properties with modules over $\Fqm[x]$.
For instance, any left or right submodule of $\SkewPolys^a$ is free and any two basis of such a submodule have the same number of elements. 
Hence, the rank of a module is well-defined.
Furthermore, two $a \times b$ matrices $\MABout_1,\MABout_2$ over $\SkewPolys$ generate the same left row (or right column) space if and only if there is an invertible $a \times a$ ($b \times b$) matrix $\U$ with $\MABout_1 = \U \MABout_2$ ($\MABout_1 = \MABout_2 \U$, resp.).
See, e.g., \cite{clark2012non} for more details.

\subsection{Evaluations of Skew Polynomials}\label{ssec:evaluation_maps}

It turns out that skew polynomials give rise to multiple notions of mappings \cite{boucher2014linear} which behave similarly to evaluation of ordinary polynomials, and these can each be used to build ``evaluation codes'' from skew polynomials.
In this paper, we consider two such ``evaluations'':
\begin{itemize}
\item operator evaluation (used in Section~\ref{sec:rank_and_subspace}) and
\item remainder evaluation (used in Section~\ref{sec:sum-rank}).
\end{itemize}
We will distinguish the two evaluation types notationally by their brackets (soft for operator and square for remainder evaluation), see below.

The \emph{operator evaluation map} of a skew polynomial $f = \sum_{i} f_i x^i \in \SkewPolys$ is defined as
\begin{equation*}
f(\cdot) \, : \, \Fqm \to \Fqm, \, \alpha \mapsto \textstyle\sum_{i} f_i \sigma^i(\alpha).
\end{equation*}
For any $f,g \in \SkewPolys$ and $\alpha \in \Fqm$, we have the following sum and product rule:
\begin{align*}
(f+g)(\alpha) &= f(\alpha) + g(\alpha) \\
(f \cdot g)(\alpha) &= f(g(\alpha)).
\end{align*}
Since $\sigma$ is an $\Fq$-linear map, also $f(\cdot)$ is an $\Fq$-linear map and the (operator) root space $\ker f(\cdot) := \{ \alpha \in \Fqm \mid f(\alpha) = 0 \}$ is an $\Fq$-vector space.
Furthermore, we have $\dim\ker f(\cdot) \leq \deg f$ for any non-zero $f \in \SkewPolys$.

For codes we will consider evaluating a skew polynomial $f$ at multiple values $\alpha_1,\ldots,\alpha_n \in \Fqm$ which are linearly indpendent over $\Fq$, for which the following constructions of skew polynomials are crucial:
\begin{itemize}
  \item
  Let $\mathcal{U} \subseteq \Fqm$ be the $\Fq$-subspace spanned by $\alpha_1,\ldots,\alpha_n$.
  Then there is a unique monic skew polynomial $\MSPop_{\mathcal{U}}$, called \emph{(operator) annihilator polynomial of $\mathcal U$} \cite{lidl1997finite,augot2018generalized} (also called \emph{minimal subspace polynomial}) with $\ker \MSPop_{\mathcal{U}}(\cdot) = \mathcal{U}$ and $\deg \MSPop_{\mathcal{U}} = \dim_{\Fq}(\mathcal U) = n$.
\item
  If $\SkewPolys_{<n}$ denotes all skew polynomials of degree less than $n$, then $\SkewPolys_{<n}$ is in bijection with $\Fqm^n$ through operator evaluation at $\alpha_1,\ldots,\alpha_n$.
  In other words, for any $r_1,\dots,r_n \in \Fqm$, there is a unique skew polynomial $\IPop{\{(\alpha_i,r_i)\}_{i=1}^{n}}$, called  the \emph{(operator) interpolation polynomial}, of degree $<n$ such that $\IPop{\{(\alpha_i,r_i)\}_{i=1}^{n}}(\alpha_i) = r_i$ for all $i=1,\dots,n$ \cite[Lemma~3.51]{lidl1997finite}, \cite{silva2007rank,augot2018generalized}.
\end{itemize}

The \emph{remainder evaluation map} of a skew polynomial $f \in \SkewPolys$ is defined by
\begin{align*}
\remev{f}{\cdot} \, : \, \Fqm \mapsto \Fqm, \quad \alpha \mapsto f \remr (x-\alpha).
\end{align*}
For any $f,g \in \SkewPolys$ and $\alpha \in \Fqm$, we have \cite{lam1988vandermonde}
\begin{align*}
\remev{(f+g)}{\alpha} &= \remev{f}{\alpha}+\remev{g}{\alpha} \\
\remev{(f \cdot g)}{\alpha} &=
\begin{cases}
0, &\text{if } c = 0, \\
\remev{f}{\tfrac{\sigma(c) \alpha}{c}}c, &\text{if } c \neq 0,
\end{cases}
\end{align*}
where $c := \remev{g}{\alpha}$.
There are analogs of annihilator and interpolation polynomials for the remainder evaluation.
However, since their definition requires further notation and is only relevant in Section~\ref{sec:sum-rank}, we will discuss these notions at the start of that section.

For more details on the evaluation maps and their differences, we refer to \cite{boucher2014linear}.
Throughout the paper, whenever it is clear from the context which evaluation map we mean, we omit the prefixes "operator" and "remainder".

\subsection{Cost of Operations with Skew Polynomials}\label{ssec:cost_skew_poly_operations}

We denote by $\OMul{n}$ the cost of multiplying two skew polynomials over $\Fqm$ of degree $n$. As there are cost bounds for skew polynomial multiplication that count operations in either $\Fqm$ or $\Fq$, we deliberately let it open over which base field $\OMul{n}$ is given. This means that we state complexities involving $\OMul{n}$ in operations in the base field of the cost bound.
The best-known cost bounds on $\OMul{n}$ are
\begin{equation*}
\OMul{n} \in \softO\!\left( \min\!\left\{n^{\omega-2} m^2, \, n m^{\omega-1} \right\} \right)
\end{equation*}
operations in $\Fq$ using the algorithms in \cite{caruso2017fast,caruso2017new} and
\begin{equation*}
\OMul{n} \in O\!\left( n^{\min\!\left\{\frac{\omega+1}{2},1.635\right\}} \right)
\end{equation*}
operations in $\Fqm$ %
using the algorithm in \cite{puchinger2017fast}.
Using a basis with $\OEF(m) \in \softO(m)$, and assuming $(\omega+1)/2 > 1.635$, i.e.~$\omega > 2.27$, the algorithms in \cite{caruso2017fast,caruso2017new} provide the best cost bounds whenever $n \in \Omega\big(m^{\frac{2}{5-\omega}}\big)$, while \cite{puchinger2017fast} provides the best cost bound when $n \in O\big(m^{\frac{2}{5-\omega}}\big)$.

All algorithms are faster than classical multiplication, which has quadratic complexity $\Theta(n^2)$ operations in $\Fqm$. This is obvious for the multiplication algorithm in \cite{puchinger2017fast} (exponent is reduced from $2$ to $\leq 1.635$), and holds for the one in \cite{caruso2017fast} due to
\begin{equation*}
\softO\!\left(\min\!\left\{n^{\omega-2} m^2, \, n m^{\omega-1} \right\}\right) \subseteq o\!\left(n^2 \OEF(m)\right).
\end{equation*}

By combining the results in \cite{caruso2017new,caruso2017fast,puchinger2017fast,puchinger2018construction}, the following skew polynomial operations can be performed in $\softO(\OMul{n})$ operations in the base field of the cost bound $\OMul{n}$:
\begin{itemize}
\item Left and right division of two skew polynomials of degree at most $n$. %
\item Operator evaluation of a skew polynomial of degree $\leq n$ at $n$ field elements (\emph{multi-point (operator) evaluation}).
\item Compute the operator annihilator polynomial $\MSPop_\mathcal{U}$ of an $n$-dimensional subspace $\mathcal{U}$.
\item Compute an operator interpolation polynomial at $n$ field elements.
\end{itemize}
In Section~\ref{sec:sum-rank}, we will discuss the remainder-evaluation analogs of the latter three operations.
We did not find the analog of the above computational cost bounds in the literature, so we show in Section~\ref{ssec:fast_remainder_ev_arith} that they can also be performed in $\softO(\OMul{n})$ operations in the base field of the cost bound $\OMul{n}$.

\section{Approximant Bases Over $\SkewPolys$}\label{sec:order_bases}

In this section, we study the central computational object that will enable us to speed up decoding algorithms and computational tools discussed in later sections: approximant bases over skew polynomial rings.
Here, we use the notation and adapt the algorithms of \cite{neiger2016bases}, which studied these bases over ordinary polynomial rings.
For skew polynomials over finite fields, the resulting algorithms have smaller complexity than the previously fastest method in \cite{beckermann_fraction-free_2002}.

\subsection{Modules and Matrices over Skew Polynomial Rings}\label{ssec:skew_modules_matrices}

For a matrix $\MABout \in \SkewPolys^{a \times b}$ and $\s \in \ZZ^a$, we define the $\s$-shifted column degree of $\MABout$ to be the tuple
\begin{equation*}
\cdeg_\s(\MABout) = [d_1,\dots,d_b] \in \left(\ZZ \cup \{-\infty\}\right)^b
\end{equation*}
where\ $d_j$ is the maximal shifted degree in the $j$-th column, i.e.,
$d_j := \textstyle\max_{i=1,\dots,a}\{\deg \MABoutentry_{ij} + s_i \}$.
We write $\cdeg(\MABout) := \cdeg_\0(\MABout)$, where $\0 := [0,\dots,0]$.
Analogously, for $\s \in \ZZ^b$, we define the ($\s$-shifted) row degree of $\MABout$ to be
\begin{equation*}
\rdeg_\s \MABout := \cdeg_\s\!\left(\MABout{}^\top\right) \quad \text{ and } \quad \rdeg \MABout := \cdeg\!\left(\MABout{}^\top\right).
\end{equation*}
The degree of the matrix, i.e.~the maximal degree among its entries, is denoted:
\begin{equation*}
\deg \MABout := \max_{i,j}\{\deg \MABoutentry_{ij}\}.
\end{equation*}
If $\v \in \SkewPolys^{1 \times a} \setminus \{\0\}$ is a row vector and $\s = [s_1,\dots,s_a] \in \ZZ^a$ a shift, we define the $\s$-pivot index of $\v$ to be the largest index $i$ with $1 \leq i \leq a$ such that
$\deg v_i + s_i = \rdeg_\s(\v)$, and analogously for column vectors.
If $a \geq b$ (or $a \leq b$, respectively), then we say that $\MABout$ is in column (row) $\s$-ordered weak Popov form if the $\s$-pivot indices of its columns (rows) are strictly increasing in the column (row) index.

The next two lemmas present key properties of matrices in row or column weak Popov form that we will use later in this section.
The first one is a variant of the ``predictable degree property'', see \cite{kailath_linear_1980}, which is central to row- or column-reduced matrices such as those in ordered row or column weak Popov form.
An analogous result holds for singular rank or non-square matrices, but we will need it only for square ones.

{
\def\vlam{\vec\lambda}
\begin{lemma}
	\label{lem:minimality_column_wpf}
  Let $\MABout \in \SkewPolys^{b \times b}$ be full rank and $\s \in \ZZ^b$.
	\begin{itemize}
	\item ``Column case'': Assume $\MABout$ is in $\s$-ordered column weak Popov form, $\t := \cdeg_\s \MABout$, and $\p = \MABout \vlam$ for non-zero column vectors $\p,\vlam \in \SkewPolys^{b \times 1}$.
  Then
	\begin{itemize}
		\item $\cdeg_\s \p = \cdeg_\t \vlam$ and
		\item the $\s$-pivot index of $\p$ equals the $\t$-pivot index of $\lambdaVec$.
	\end{itemize}
	\item ``Row case'': Assume $\MABout$ is in $\s$-ordered row weak Popov form, $\t := \rdeg_\s \MABout$, and $\p = \vlam \MABout$ for non-zero row vectors $\p,\vlam \in \SkewPolys^{1 \times b}$.
  Then
	\begin{itemize}
		\item $\rdeg_\s \p = \rdeg_\t \vlam$ and
		\item the $\s$-pivot index of $\p$ equals the $\t$-pivot index of $\vlam$.
	\end{itemize}
	\end{itemize}

\end{lemma}

\begin{proof}
	We first prove the column case.
	Let $\mu := \cdeg_\t \vlam$ and $h$ be the $\t$-pivot intex of $\vlam$.
	Since $\p = \MABout \vlam$, then $\deg p_i \leq \max_{j=1,\ldots,b}\{ \deg B_{ij} + \deg \lambda_j \} \leq \max_{j=1,\ldots,b} \{t_j - s_i + \deg \lambda_j\}$, and so $\cdeg_\s \p \leq \mu$.
	Let $\u \in \Fqm^{b \times 1}$ be the vector whose $i$-th entry is the $x^{\mu-s_i}$-coefficient of $p_i$ (the coefficient is zero if $\deg p_i < \mu-s_i$).
	Hence, $\cdeg_\s \p = \mu$ iff $\u \neq \0$. Further, if $\u \neq \0$, then the $\s$-pivot index of $\p$ is the greatest non-zero index of $\u$.

	Since $\deg B_{ij} \leq t_j-s_i$ and $\deg \lambda_j \leq \mu-t_j$, the entries of $\u$ only depend on some of the leading coefficients in the matrix $\MABout$ and vector $\vlam$.
	Let $\textrm{lm}_\s(\MABout)$ be the $\s$-leading matrix of $\MABout$ whose $(i,j)$-th entry is the $x^{t_j-s_i}$-coefficient of $B_{ij}$, defined as $0$ if $\deg B_{ij}<t_j-s_i$.
	Similarly, define $l_j$ to be the $x^{\mu-t_j}$-coefficient of $\lambda_j$.
	Then, by the definition of linearized polynomial multiplication, $u_i$ is the inner product of the $i$-th row of $\textrm{lm}_\s(\MABout)$ and the vector $\l_i := [\sigma^{t_1-s_i}(l_1),\dots,\sigma^{t_b-s_i}(l_b)]^\top$.
	
	Since $\MABout$ is full-rank and in $\s$-ordered column weak Popov form, the $\s$-pivot index of its $j$-th column is $j$ and $\textrm{lm}_\s(\MABout)$ is in upper triangular form with only non-zero entries on its diagonal.
	Also, $\l_i \neq \0$ since at least one $\lambda_j$ fulfills $\deg \lambda_j + t_j = \mu$, and $h$ as defined above is the greatest non-zero index of $\l_i$ (independent of $i$).
	Thus, $u_{h}$ is non-zero and $h$ is also the greatest non-zero index of $\u$, which proves the claim.

	The row case follows analogously, the only differences being that $\textrm{lm}_\s(\MABout)$ is defined as the matrix containing the $x^{t_i-s_j}$-coefficient of $B_{ij}$ (which is in lower triangular form), and that $u_i$ is the inner product of the vector $\l_i := [l_1,\dots,l_b]$ (no automorphisms applied) and the $i$-th column of a slight modification of the matrix $\textrm{lm}_\s(\MABout)$, where we apply certain automorphisms to the matrix entries. This does not change the argument above since automorphisms do not map non-zero entries to zero.
\end{proof}

\begin{remark}
The predictable degree property (Lemma~\ref{lem:minimality_column_wpf}) was studied for \emph{row-reduced} matrices over skew polynomials in \cite[Lemma~A.1]{beckermann2006fraction}.
More precisely, the property $\rdeg_\s \p = \rdeg_\t \lambdaVec$ (``row case'') was shown for the shift $\s=\0$.
Since ``row reduced'' is weaker than ``ordered weak Popov'', pivots of $\p$ and $\lambdaVec$ are not necessarily the same in this case.
\end{remark}

The following lemma is the skew analog of \cite[Theorem~1.28, case~(iii)]{neiger2016bases}. We state the theorem for column weak Popov form and write the row case in parentheses.

\begin{lemma}\label{lem:s_ordered_property_multiplication}
Let $\MABout_{1} \in \SkewPolys^{b \times b}$ be in $\s$-ordered column (row) weak Popov form and $\MABout_{2} \in \SkewPolys^{b \times b}$ be in $\t$-ordered column (row) weak Popov form, where $\t := \cdeg_\s(\MABout_{1})$ ($\t := \rdeg_\s(\MABout_{1})$). Then, $\MABout_{1} \MABout_{2}$ ($\MABout_{2} \MABout_{1}$) is in $\s$-ordered column (row) weak Popov form.
\end{lemma}
\begin{proof}
	We prove the column case, the row case follows analogously.
	Let $\u = [u_1,\ldots,u_b] = \cdeg_\t(\MABout_{2})$.
	Let $\h_i$ be the $i$-th column of $\MABout_{1} \MABout_{2}$.
	Denote by $B_{2,ij}$ the $(i,j)$-th entry of $\MABout_{2}$.
	By Lemma~\ref{lem:minimality_column_wpf} then $\cdeg_\s \h_j = \max_{i=1,\ldots,b}\{\deg B_{2,ij} + t_i\} = u_j$, and further the $\s$-pivot index of $\h_j$ is $\max\{i \, : \, \deg B_{2,ij} + t_i = u_j\}$ which is exactly the $\t$-pivot index of the $j$-th column of $\MABout_{2}$.
	Since these are all in strictly increasing order, so must the $\s$-pivots of $\h_1,\ldots,\h_b$.
	Hence $\MABout_{2} \MABout_{1}$ is in ordered weak Popov form.
\end{proof}

}

\subsection{Approximant Bases over $\SkewPolys$}\label{ssec:skew_order_bases}

Let $\MABin \in \SkewPolys^{a \times b}$ and $d \in \ZZ_{\geq 0}$.
A right approximant of $\MABin$ of order $d$ is a vector $\b \in \SkewPolys^{b \times 1}$ such that
\begin{equation*}
\MABin \b \equiv \0 \modl x^d.
\end{equation*}
A left approximant of $\MABin$ of order $d$ is $\b \in \SkewPolys^{1 \times a}$ with
\begin{equation*}
\b \MABin  \equiv \0 \modr x^d.
\end{equation*}

\begin{lemma}\label{lem:approximants_modules}
The set of right (left) approximants of $\MABin$ of order $d$ is a free right (left) $\SkewPolys$-module of rank $b$ (rank $a$).
\end{lemma}

\begin{proof}
The set is a subset of $\SkewPolys^{b \times 1}$ ($\SkewPolys^{1 \times a}$, respectively) and obviously closed under addition and right (left) multiplication by elements of $\SkewPolys$, hence a free right (left) module.
Further, the vector $[ 0, \ldots, 0, x^{d}, 0, \ldots, 0 ]$ of suitable length is clearly a right (left) approximant of $\MABin$ of order $d$, so the module of all right (left) approximants must contain a module of rank $b$ (rank $a$), hence must themselves be of rank $b$ (rank $a$) since it cannot be greater.
\end{proof}

Lemma~\ref{lem:approximants_modules} shows that the following definition is well-posed.

\begin{definition}[left/right approximant bases]\label{def:minimal_approximant_basis}
Let $\MABin \in \SkewPolys^{a \times b}$ and $d \in \ZZ_{\geq 0}$.
\begin{itemize}
\item For $\s \in \ZZ^b$, a right {\MABnameFullStandard} is a full-rank matrix $\MABout \in \SkewPolys^{b \times b}$ s.t.\
	\begin{enumerate}
		\item $\MABout$ is in $\s$-ordered column weak Popov form.
		\item The columns of $\MABout$ are a basis of all right approximants of $\MABin$ of order $d$.
	\end{enumerate}
\item For $\s \in \ZZ^a$, a left {\MABnameFullStandard} is a full-rank matrix $\MABout \in \SkewPolys^{a \times a}$ s.t.\
	\begin{enumerate}
		\item $\MABout$ is in $\s$-ordered row weak Popov form.
		\item The rows of $\MABout$ are a basis of all right approximants of $\MABin$ of order $d$.
	\end{enumerate}
\end{itemize}
We denote by $\RMABnameShortStandard$ (right case) and $\LMABnameShortStandard$ (left case) the sets of all such bases, respectively.
If the input is not relevant, we simply write (left or right) approximant basis.
\end{definition}

\begin{remark}
The most common definition in the literature requires approximant bases only to be row-reduced (denoted by ``($\s$-)minimal approximant basis'').
Here, we use a stronger normal form, ordered weak Popov form.
The motivation comes from \cite{neiger2016bases}, where (over ordinary polynomials) it was shown that the fastest algorithms for computing approximant bases can be adapted to output ordered weak Popov forms at no extra (asymptotic) cost.
\end{remark}

\begin{algorithm*}
	\caption{$\textsf{RightBaseCase}$ \cite{GiorgiPolyMatrix,neiger2016bases}}\label{alg:rightOrdinaryLinAppBas}
	\SetKwInOut{Input}{Input}\SetKwInOut{Output}{Output}
	\Input{matrix $\MABinhat \in \Fqm[x]^{a\times b}$ with $\deg(\MABinhat)<1$, shifts $\s\in\mathbb{Z}^{b}$}
	
	\Output{$\MABouthat \in \Fqm[x]^{b\times b}$, a right {\MABnameFull{\MABinhat}{\s}{1}} over $\Fqm[x]$}
	
	$\pi_\s\gets b\times b$ permutation matrix s.t. $[(s_1,1),\dots,(s_b,b)]\pi_\s$ is lexicographically increasing \label{alg:start_ordinary_PMbasis}  \\
	$[i_1,\dots,i_\rho],[j_1,\dots,j_\rho]\gets$ row and column rank profiles of $\MABinhat \pi_\s$ (i.e., the column/row indices of leading ones in a row/column echelon form of $\MABinhat \pi_\s$) \hfill \myAlgoComment{compute as in \cite[Theorem~2.10]{storjohann2000algorithms}}
	$[k_1,\dots,k_{b-\rho}]\gets\{1,\dots,b\} \setminus \{j_1,\dots,j_\rho\}$ sorted increasingly \\
	$\MABinhat_1\gets$ submatrix of $\MABinhat\pi_\s$ with indices in $\{i_1,\dots,i_\rho\}\times\{j_1,\dots,j_\rho\}$ \\
	$\MABinhat_2\gets$ submatrix of $\MABinhat\pi_\s$ with indices in $\{i_1,\dots,i_\rho\}\times\{k_1,\dots,k_{b-\rho}\}$\\
	$\pi\gets$ permutation s.t. $[j_1 \dots j_\rho k_1 \dots k_{b-\rho}]\pi=[1\dots b]$\\
	\Return{$\pi_\s\pi^{-1}
		\begin{bmatrix}
		x\mat{I}_\rho & -\MABinhat_1^{-1}\MABinhat_2\mat{I}_{b-\rho} 
		\\
		\mat{0}& \mat{I}_{b-\rho}
		\end{bmatrix}
		\pi\pi_\s^{-1} \in \Fqm[x]^{b \times b}$ \label{alg:stop_ordinary_PMbasis}}
\end{algorithm*} 

For approximant bases over ordinary polynomial rings, the ``row/left'' versus the ``column/right'' view becomes one of notational convenience, since we can trivially obtain one from the other by transposition.
In the non-commutative case of approximant bases over skew polynomial rings, this is no longer true (see Example~\ref{ex:left_right_bases_different} in Appendix~\ref{app:examples} for a counterexample) and the row and column cases are simply slightly different: we need theorems and algorithms tailored to each case, even if most of the statements and proofs are very similar for the two cases.

The currently fastest algorithm to compute a left approximant basis over $\SkewPolys$ (in the weaker ``row-reduced form'' instead of ordered weak Popov form) is $O(a^3 b^2 d^2)$ operations in $\Fqm$ \cite{beckermann_fraction-free_2002}.
Note that the algorithm in \cite{beckermann_fraction-free_2002} is designed to handle coefficient growth in certain infinite fields, and also the complexity analysis is only done for this case.
Our own analysis of the algorithm gives the stated complexity over $\Fqm$.

\subsection{A New Algorithm to Compute Approximant Bases}\label{ssec:fast_order_bases_computation}

In this section, we adapt the recursive (left) \textsf{PM-Basis} algorithm~\cite{GiorgiPolyMatrix,neiger2016bases} over ordinary polynomial rings to compute a left and right approximant basis over skew polynomials.
For the base case (Section~\ref{ssec:right_PM-basis_base_case}), we prove that the algorithm over $\Fqm[x]$ can be used with only small modifications.
Also the recursion step (Section~\ref{ssec:right_PM-basis_recursive_step}) is very similar to the original algorithm, but we need to be careful about the non-commutativity of the skew polynomial ring.

\subsubsection{Base Case: Right and Left Approximant Bases of Degree~$1$}
\label{ssec:right_PM-basis_base_case}

In the following, we show how to obtain right and left approximant basis order $1$ of a degree $0$ matrix.
For both sides, we reduce the problem to computing an approximant basis over the ordinary polynomials \cite{GiorgiPolyMatrix,neiger2016bases} (cf.~Algorithm~\ref{alg:rightOrdinaryLinAppBas}) using suitable bijective mappings between $\Fqm[x]$ and $\SkewPolys$.

For the right case, we use the following mapping $\varphi$ and its inverse, which we extend to matrices entry-wise.
\begin{align}
\varphi \, : \, \Fqm[x] &\to \SkewPolys, \notag \\
\textstyle\sum_{i}f_i x^i &\mapsto \textstyle\sum_{i} x^i f_i. \label{eq:phi_mapping}
\end{align}
Note that by the non-commutative multiplication rule \eqref{eq:skew_polynomials_multiplication_rule} for skew polynomials, we have $\sum_{i} x^i f_i = \sum_{i} \sigma^i(f_i) x^i$.
We use two important properties of the mapping:
\begin{align}
&\varphi(f h) = \varphi(f) \varphi(h) \quad \forall \, f \in \Fqm[x] \text{ and } h \in \Fqm[x]_{<1}, \label{eq:phi_key_property_degree_0_poly} \\
&\varphi(f g \rem x) = \varphi(f) \varphi(g) \reml x \quad \forall \, f,g \in \Fqm[x], \label{eq:phi_key_property_0_coeff}
\end{align}
where \eqref{eq:phi_key_property_degree_0_poly} is obvious from the definition
and \eqref{eq:phi_key_property_0_coeff} is a direct consequence of the first property (write $g = g_0 x^0 + (\smallsum_{i>0} g_i x^i)$ and use the additivity of $\varphi$).

The resulting algorithm for computing right skew approximant bases of order $1$ is outlined in Algorithm~\ref{alg:rightSkewLinAppBas}. We prove its correctness using the reduction shown in Figure~\ref{fig:base_case_reduction}.

\begin{figure}[ht]
	\begin{center}
		\vspace{-0.3cm}
		\begin{tikzpicture}
		\def\xdist{2}
		\def\ydist{2}
		\def\labelydist{1}
		\def\labelxdist{0.5}
		
		\draw[draw=none, fill=black!10] (0.5*\xdist,1.5*\labelydist) rectangle (5.5,-1.35*\ydist);
		
		\node (skewpolys) at (0,\labelydist) {$\SkewPolys$};
		\node (ordinarypolys) at (\xdist,\labelydist) {$\Fqm[x]$};
		\node[left] (QmatricesLabel) at (-\labelxdist,0) {matrix of degree $0$};
		\node[left, align=right] (FmatricesLabel) at (-\labelxdist,-\ydist) {approximant \\ basis of order $1$};
		
		\node (Q)    at (0,0) {$\MABin$};
		\node (Qhat) at (\xdist,0) {$\MABinhat$};
		\node (Fhat) at (\xdist,-\ydist) {$\MABouthat$};
		\node (F) at (0,-\ydist) {$\MABout$};
		
		\draw[->, >=latex] (Q) to node[above] {$\varphi^{-1}(\cdot)$} (Qhat);
		\draw[->, >=latex] (Qhat) to node[right, align=left] {\textsf{PM-Basis} algorithm \\ over $\Fqm[x]$ \cite{GiorgiPolyMatrix,neiger2016bases} \\
		(Algorithm~\ref{alg:rightOrdinaryLinAppBas}) %
		} (Fhat);
		\draw[->, >=latex] (Fhat) to node[below] {$\varphi(\cdot)$} (F);
		\draw[->, >=latex, dashed] (Q) to (F);
		
		\end{tikzpicture}
	\end{center}
	\vspace{-0.3cm}
	\caption{Illustration of the reduction used in the correctness proof of Algorithm~\ref{alg:rightSkewLinAppBas} (Theorem~\ref{thm:correctness_rightLinApp}). Variables are defined as in Algorithm~\ref{alg:rightSkewLinAppBas}.}
	\label{fig:base_case_reduction}
\end{figure}
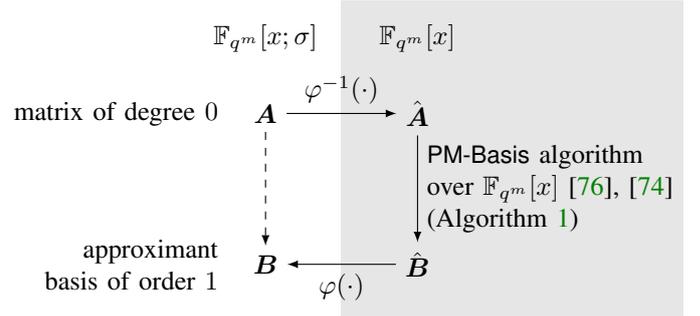

\begin{algorithm}[ht]
	\caption{$\textsf{RightSkewBaseCase}$}\label{alg:rightSkewLinAppBas}
	\SetKwInOut{Input}{Input}\SetKwInOut{Output}{Output}
	\Input{$\MABin\in\SkewPolys^{a\times b}$ with $\deg(\MABin)<1$, $\s\in\mathbb{Z}^{b}$}
	\Output{$\MABout \in \RMABnameShort{\MABin}{\s}{1}$}
	$\MABinhat \in \Fqm[x]_{<1}^{a \times b}\gets \varphi^{-1}(\MABin)$ \hfill \myAlgoComment{mapping $\varphi$ as in \eqref{eq:phi_mapping}}
	$\MABouthat \gets \textsf{RightBaseCase}\big(\MABinhat,\s\big)$ \\	
	\Return{$\varphi\big(\MABouthat\big)$ \hfill \myAlgoComment{mapping $\varphi$ as in \eqref{eq:phi_mapping}}}
\end{algorithm} 

\begin{theorem}\label{thm:correctness_rightLinApp}
Algorithm~\ref{alg:rightSkewLinAppBas} is correct and has complexity
\begin{equation*}
O(\rho^{\omega-2}ab)
\end{equation*}
operations in $\Fqm$ where $\rho\leq \min\{a,b\}$ is the rank of $\MABin$.
\end{theorem}

\begin{proof}
Algorithm~\ref{alg:rightSkewLinAppBas} consists of three parts, which are also illustrated in Figure~\ref{fig:base_case_reduction}: The first line maps the input matrix $\MABin$ to $\Fqm[x]$; note that this is actually the identity mapping since $\deg \MABin < 1$.
Then Lines~\ref{alg:start_ordinary_PMbasis} to \ref{alg:stop_ordinary_PMbasis} apply the well-known \textsf{PM-Basis} algorithm \cite{GiorgiPolyMatrix,neiger2016bases} over $\Fqm[x]$, and finally, the resulting matrix $\MABouthat$, which is an {\MABnameFull{\MABinhat}{\s}{1}} over $\Fqm[x]$, is mapped back to the skew polynomial ring.
We show that $\MABout \in \RMABnameShort{\MABinhat}{\s}{1}$ using properties of $\MABouthat$ and~$\varphi$.

Note that the mapping $\varphi$ does not change the degree of a polynomial. As $\MABouthat$ is in $\s$-ordered weak Popov form, so is $\MABout$.

Denote by $\b_i$ and $\hat{\b}_i$ the $i$-th column of $\MABout$ and $\MABouthat$, respectively.
Since $\hat{\b}_i$ is a right approximant of $\MABinhat$ and due to Property~\eqref{eq:phi_key_property_0_coeff}, we have 
\begin{align}
\MABin \b_i \reml x &= \varphi(\MABinhat) \varphi(\hat{\b}_i) \reml x \notag \\
&= \varphi(\MABinhat \hat{\b}_i \rem x) = \varphi(\0) = \0, \label{eq:phi_right_approximant}
\end{align}
so the columns of $\MABout$ are right approximants of $\MABin$ of order $1$.

It is left to show that the (right) column space of $\MABout$ contains all right approximants of $\MABin$ of order $1$ and that its columns are right $\SkewPolys$-linearly independent.
For this, we identify two key properties of $\MABout$ and $\MABouthat$, respectively.
\begin{enumerate}
	\item \label{itm:F_property_x} The (right) column space of $x \I_{b} \in \SkewPolys^{b \times b}$ is contained in the column space of $\MABout$.
	\item \label{itm:F_property_degree} If $\vhat = \MABouthat \lambdaVechat$ for two vectors $\vhat,\lambdaVechat \in \Fqm[x]^{b}$ and $\deg \vhat = 0$, then $\deg \lambdaVec = 0$.
\end{enumerate}
The first property follows from the shape of $\MABout$, which, by Algorithm~\ref{alg:rightOrdinaryLinAppBas}, is of the form
\begin{equation}
\MABout = \P^{-1}
\begin{bmatrix}
x\I_{\rho} & \D \\
\0 & \I_{b-\rho}
\end{bmatrix} \P \in \SkewPolys^{b \times b}, \label{eq:shape_F}
\end{equation}
where $\deg \D \leq 0$ and $\P \in \Fqm^{b \times b}$ is an invertible permutation matrix (it has exactly one $1$ in each row and column, and $0$ otherwise). Since the column space of $x \I_{b}$ is invariant under permutations of coordinates, and since for any $i=1,\dots,b$, we can easily find a vector $\lambdaVec' \in \SkewPolys^{b \times 1}$ with $\begin{bmatrix}
x\I_{\rho} & \D \\
\0 & \I_{b-\rho}
\end{bmatrix} \lambdaVec' = x \e_i$ (where $\e_i$ is the $i$-th unit vector), the first property follows.
For the second property, first observe that $\deg \lambdaVechat = \cdeg_{\0} \lambdaVechat \leq \cdeg_{\cdeg_{\0} \MABout} \lambdaVechat$ since $\cdeg_{\0} \MABout \geq 0$.
Since \eqref{eq:shape_F}, seen over $\Fqm[x]$, is in unshifted ($\s=\0$) ordered weak Popov form, the predictable degree property implies
\begin{equation*}
 \cdeg_{\cdeg_{\0} \MABout} \lambdaVechat = \cdeg_{\0} \vhat = \deg \vhat = 0.
\end{equation*}

Let now $\v$ be a right approximant of $\MABin$ of order $1$ and we should show that it is in the column space of $\MABout$.
Write $\vec v = \vec v_0 + x \vec v_1$, where $\deg \vec v_0 \leq 0$.
By Property~\ref{itm:F_property_x}), then $x \vec v_1$ is in the column space of $\MABout$, so we are done if the same holds for $\vec v_0$.
By the same argument as in \eqref{eq:phi_right_approximant}, the vector $\vhat_0 := \varphi^{-1}(\v_0) \in \Fqm[x]^b$ is a right approximant of $\MABinhat$ and there is a vector $\lambdaVechat \in \Fqm[x]^b$ such that $\vhat_0 = \MABouthat \lambdaVechat$.
Due to Property~\ref{itm:F_property_degree}), we have $\deg \lambdaVechat = 0$, which by \eqref{eq:phi_key_property_degree_0_poly} implies
\begin{equation*}
\v_0 = \varphi(\vhat_0) = \varphi(\MABouthat \lambdaVechat) = \varphi(\MABouthat) \varphi(\lambdaVechat) = \MABout \varphi(\lambdaVechat).
\end{equation*}
The columns of $\MABout$ are right $\SkewPolys$-linearly independent since $\MABout$ is up to row and column permutations in upper-triangular form (see~\eqref{eq:shape_F}) with non-zero entries on the diagonal.
Correctness of the algorithm follows.

The main computational task is to compute the row and column rank profile of the matrix $\MABinhat\in\Fqm^{a\times b}$ of rank $\rho$ which requires $O(\rho^{\omega-2}ab)$ operations in $\Fqm$~\cite[Thm.~2.10]{storjohann2000algorithms}.
\end{proof}

\begin{remark}\label{rem:PM_base_case_proof_strategy}
We chose to rely on the \textsf{PM-Basis} algorithm over $\Fqm[x]$ in the proof of Theorem~\ref{thm:correctness_rightLinApp} since it stresses the similarities and differences of the skew and ordinary polynomial case for approximant bases of order $1$ of matrices of degree $0$.
For a self-contained proof of Theorem~\ref{thm:correctness_rightLinApp}, which directly adapts the key ideas of the \textsf{PM-Basis} correctness proof in \cite{neiger2016bases}, we refer to the conference version of this paper \cite{bartz2019fast}.

The same reduction is not possible with the mapping $\varphi$ for higher degrees and orders as we can see in Example~\ref{ex:counterexample_reduction_high_orders} (Appendix~\ref{app:examples}).
\end{remark}

For the left case, we use the following bijective mapping.
\begin{align}
\psi \, : \, \Fqm[x] &\to \SkewPolys, \notag \\
\textstyle\sum_{i}f_i x^i &\mapsto \textstyle\sum_{i} f_i x^i. \label{eq:psi_mapping}
\end{align}
The resulting algorithm is presented in Algorithm~\ref{alg:leftSkewLinAppBas} and we prove its correctness in Theorem~\ref{thm:correctness_leftLinApp}.

\begin{algorithm}[ht]
	\caption{$\textsf{LeftSkewBaseCase}$}\label{alg:leftSkewLinAppBas}
	\SetKwInOut{Input}{Input}\SetKwInOut{Output}{Output}
	\Input{$\MABin\in\SkewPolys^{a\times b}$ with $\deg(\MABin)<1$, $\s\in\mathbb{Z}^{a}$}
	\Output{$\MABout \in \LMABnameShort{\MABin}{\s}{1}$} %
	$\MABinhat \in \Fqm[x]_{<1}^{a\times b}\gets \psi^{-1}(\MABin)$ \hfill \myAlgoComment{mapping $\psi$ as in \eqref{eq:psi_mapping}}
	$\MABouthat \gets \textsf{RightBaseCase}\big(\MABinhat{}^\top,\s\big)$ \\	
	\Return{$\psi\big(\MABouthat{}^\top\big)$ \hfill \myAlgoComment{mapping $\psi$ as in \eqref{eq:psi_mapping}}}
\end{algorithm}

\begin{theorem}\label{thm:correctness_leftLinApp}
Algorithm~\ref{alg:leftSkewLinAppBas} is correct and has complexity
\begin{equation*}
O(\rho^{\omega-2}ab)
\end{equation*}
operations in $\Fqm$, where $\rho\leq \min\{a,b\}$ is the rank of $\MABin$.
\end{theorem}

\begin{proof}
The proof is the same as the one of Theorem~\ref{thm:correctness_rightLinApp}, using the analogous properties of \eqref{eq:phi_key_property_degree_0_poly} and \eqref{eq:phi_key_property_0_coeff} for $\psi$ in the left side,
\begin{align}
&\psi(f h) = \psi(f) \psi(h) \quad \forall \, f \in \Fqm[x]_{<1} \text{ and } h \in \Fqm[x], \label{eq:pii_key_property_degree_0_poly} \\
&\psi(f g \rem x) = \psi(f) \psi(g) \remr x \quad \forall \, f,g \in \Fqm[x], \label{eq:pii_key_property_0_coeff}
\end{align}
as well as the following ``transposed'' analogs of the properties of $\MABout$ and $\MABouthat$ in the right case:
\begin{enumerate}
	\item \label{itm:F_property_x_left} The (left) row space of $x \I_{a} \in \SkewPolys^{a \times a}$ is contained in the row space of $\MABout$.
	\item \label{itm:F_property_degree_left} If $\vhat = \lambdaVechat \MABouthat$ for two vectors $\vhat,\lambdaVechat \in \Fqm[x]^{a}$ and $\deg \vhat = 0$, then $\deg \lambdaVec = 0$.
\end{enumerate}
Recall that over $\Fqm[x]$, the transpose of a right approximant basis of $\MABinhat{}^\top$ is a left approximant basis of $\MABinhat$.
\end{proof}

\subsubsection{Recursive Algorithm: Right and Left PM-Basis}
\label{ssec:right_PM-basis_recursive_step}

This section presents a skew-polynomial variant of the \textsf{PM-Basis} algorithm, which computes approximant bases of higher order $d>1$ in a recursive fashion using Algorithm~\ref{alg:rightSkewLinAppBas} (right side) and Algorithm~\ref{alg:leftSkewLinAppBas} (left side) as its base case, respectively.
The recursion step is based on the following lemmas, which would remain true if stated over ordinary polynomial rings. However the ordering of the involved polynomial products and the choice of left/right modulo is central for the statements and proofs to hold over the non-commutative skew polynomial ring.

\begin{lemma}\label{lem:appBasisProduct} %
Let $d\in\mathbb{Z}_{>0}$, $\MABin\in\SkewPolys^{a\times b}$ of degree less than $d$, and $d_1,d_2 \in \mathbb{Z}_{>0}$ be such that $d_1+d_2 = d$.

Let $\s\in\mathbb{Z}^{b}$,
$\MABout_1 \in \RMABnameShort{\MABin \reml x^{d_1}}{\s}{d_1}$, and
$\MABout_2 \in \RMABnameShort{x^{-d_1} \MABin \MABout_1 \reml x^{d-d_1}}{\t}{d_2}$, where $\vec{t} := \cdeg_\s(\MABout_1)$. 
Then, $\MABout_1 \MABout_2 \in \RMABnameShortStandard$.

Let $\s\in\mathbb{Z}^{a}$, $\MABout_1 \in \LMABnameShort{\MABin \reml x^{d_1}}{\s}{d_1}$, and
$\MABout_2 \in \LMABnameShort{\MABout_1 \MABin x^{-d_1} \remr x^{d-d_1}}{\t}{d_2}$, where $\vec{t} := \rdeg_\s(\MABout_1)$.
Then, $\MABout_2 \MABout_1 \in \LMABnameShortStandard$.

\end{lemma}
	
	\begin{IEEEproof}
		We prove the right case, the left-side case follows analogously.
		First, we show that all approximants of $\MABin$ of order $d$ are right $\SkewPolys$-linear combinations of the columns of $\MABout_1 \MABout_2$.		
		Let $\b$ be an approximant of $\MABin$ of order $d$ and decompose $\MABin$ as $\MABin=\MABin\reml x^{d_1}+x^{d_1} \tilde{\MABin}$. Then,
		\begin{align*}
		(\MABin\reml x^{d_1}+x^{d_1}\tilde{\MABin})\b&\equiv \0 \modl x^{d}
		\\
		\quad\Longrightarrow\quad 
		(\MABin\reml x^{d_1})\b&\equiv \0 \modl x^{d_1}.
		\end{align*}
		Hence, $\b$ is also an approximant of $\MABin\reml x^{d_1}$ of order $d_1$ and we can write $\b = \MABout_1 \vec{\lambda}$ for some $\vec{\lambda} \in \SkewPolys^b$.
		This $\lambdaVec$ again fulfills
		\begin{align*}
		\MABin \MABout_1 \lambdaVec &\equiv \0 \modl x^d, \\
		\Longrightarrow \MABin \MABout_1 \lambdaVec &= x^d \v' \\
		\Longrightarrow x^{-d_1}\MABin \MABout_1 \lambdaVec &= x^{d-d_1} \v' = x^{d_2} \v' \\
		\Longrightarrow x^{-d_1}\MABin \MABout_1 \lambdaVec &\equiv \0 \modl x^{d_2}, 
		\end{align*}
		for some $\v' \in \SkewPolys$. Again, we can decompose
		\begin{equation*}
		x^{-d_1}\MABin \MABout_1 = (x^{-d_1}\MABin \MABout_1 \reml x^{d_2}) + x^{d_2} \tilde{\MABin}
		\end{equation*}
		and have
		$(x^{-d_1}\MABin \MABout_1 \reml x^{d_2}) \lambdaVec \equiv \0 \modl x^{d_2}$.
		Thus, $\lambdaVec$ is an approximant of $x^{-d_1}\MABin \MABout_1 \reml x^{d_2}$ of order $d_2$ and can be written as $\lambdaVec = \MABout_2 \vec{\mu}$. Overall, we get
		\begin{equation*}
		\b = \MABout_1 \MABout_2 \vec{\mu},
		\end{equation*}
		so $\b$ is in the right column span of $\MABout_1 \MABout_2$.

		For the other direction, let $\b = \MABout_1 \MABout_2 \vec{\mu}$ be in the column span of $\MABout_1 \MABout_2$. We show that $\b$ is an approximant of $\MABin$ of order $d$.
		Let $\lambdaVec = \MABout_2 \vec{\mu}$. Thus, $\lambdaVec$ is an approximant of $\x^{-d_1} \MABin \MABout_1 \reml x^{d_2}$ and we have
		\begin{align*}
		(\x^{-d_1} \MABin \MABout_1 \reml x^{d_2}) \lambdaVec &\equiv \0 \modl x^{d_2} \\
		\Longrightarrow (\x^{-d_1} \MABin \MABout_1 \reml x^{d_2}) \lambdaVec &= x^{d_2} \v'
		\end{align*}
		for some $\v' \in \SkewPolys^b$. We can again write $\x^{-d_1} \MABin \MABout_1 \reml x^{d_2} = \x^{-d_1} \MABin \MABout_1 - x^{d_2} \tilde{\MABin}$ and get
		\begin{align*}
\x^{-d_1} \MABin \MABout_1 \lambdaVec &= x^{d_2} \v' + x^{d_2} \tilde{\MABin} \lambdaVec \\
\Longrightarrow \MABin \MABout_1 \lambdaVec &= x^{d_1+d_2} \v' + x^{d_1+ d_2} \tilde{\MABin} \lambdaVec \\
\Longrightarrow \MABin \MABout_1 \lambdaVec &\equiv \0 \modl x^{d}.
\end{align*}
		Hence, $\b = \MABout_1 \MABout_1 \vec{\mu}$ is an approximant of $\MABin$ of order $d$.

		By Lemma~\ref{lem:s_ordered_property_multiplication}, $\MABout_1\MABout_2$ is in $\s$-ordered weak Popov form and the statement follows.
	\end{IEEEproof}

Algorithms~\ref{alg:rightSkewDaCAppBasis} and \ref{alg:leftSkewDaCAppBasis} are fast divide \& conquer algorithms for constructing right and left approximant bases over skew polynomial rings, respectively.
The algorithms use Lemma~\ref{lem:appBasisProduct} with $d_1=\lceil d/2\rceil$ and $d_2=d-d_1$ recursively and are fast skew variants of~\cite[\textsf{PM-Basis}]{GiorgiPolyMatrix}.
	
	\begin{algorithm}[ht!]
		\caption{$\textsf{RightSkewPMBasis}$}\label{alg:rightSkewDaCAppBasis}
		\SetKwInOut{Input}{Input}\SetKwInOut{Output}{Output}
		\Input{\begin{itemize}
				\item positive integer $d\in\mathbb{Z}_{>0}$,
				\item matrix $\MABin\in\SkewPolys^{a\times b}$ of degree $<d$,
				\item shifts $\s\in\mathbb{Z}^{b}$.
		\end{itemize}}
		
		\Output{$\B \in \RMABnameShortStandard$} %
		
		\BlankLine
		
		\If{$d=1$}{
			\Return{{\rm $\textsf{RightSkewBaseCase}(\MABin,\s)$} \hfill \myAlgoComment{Algorithm~\ref{alg:rightSkewLinAppBas}}}
		}
		\Else{
			$d_{1}\gets \lceil d/2\rceil$, $d_{2}\gets d-d_{1}$ \\
			$\MABout_{1}\gets\textsf{RightSkewPMBasis}\left(d_{1},\MABin\reml x^{d_{1}},\s\right)$ \\
			$\mat{G}\gets\left(x^{-d_{1}}\MABin\MABout_1\right)\reml x^{d_{2}}$; $\vec{t}\gets \cdeg_\s\left(\MABout_1\right)$ \label{line:mat_mult_right} \\
			$\MABout_2\gets\textsf{RightSkewPMBasis}\left(d_2,\mat{G},\vec{t}\right)$\\
			\Return{$\MABout_1\MABout_2$ \label{line:mat_mult_right_2}}
		}
	\end{algorithm}

	\begin{algorithm}[ht!]
		\caption{$\textsf{LeftSkewPMBasis}$}\label{alg:leftSkewDaCAppBasis}
		\SetKwInOut{Input}{Input}\SetKwInOut{Output}{Output}
		\Input{\begin{itemize}
				\item positive integer $d\in\mathbb{Z}_{>0}$,
				\item matrix $\MABin\in\SkewPolys^{a\times b}$ of degree $<d$, 
				\item shifts $\s\in\mathbb{Z}^{a}$.
		\end{itemize}}
		
		\Output{$\B \in \LMABnameShortStandard$} %
		
		\BlankLine
		
		\If{$d=1$}{
			\Return{{\rm $\textsf{LeftSkewBaseCase}(\MABin,\s)$} \hfill \myAlgoComment{Algorithm~\ref{alg:leftSkewLinAppBas}}}
		}
		\Else{
			$d_{1}\gets \lceil d/2\rceil$, $d_{2}\gets d-d_{1}$ \\
			$\MABout_{1}\gets\textsf{LeftSkewPMBasis}\left(d_{1},\MABin\remr x^{d_{1}},\s\right)$ \\
			$\mat{G}\gets\left(\MABout_1\MABin x^{-d_{1}}\right)\remr x^{d_{2}}$; $\vec{t}\gets \rdeg_\s\left(\MABout_1\right)$ \label{line:mat_mult_left} \\
			$\MABout_2\gets\textsf{LeftSkewPMBasis}\left(d_2,\mat{G},\vec{t}\right)$\\
			\Return{$\MABout_2\MABout_1$ \label{line:mat_mult_left_2}}
		}
	\end{algorithm}

	\begin{theorem}\label{thm:correctness_DaCApp}
		Algorithm~\ref{alg:rightSkewDaCAppBasis} is correct and has complexity 
		\begin{equation*}
		\softO\big(\max\{a,b\}b^{\omega-1} \OMul{d}\big)
		\end{equation*}
		operations in the base field of the cost bound $\OMul{d}$.
		Algorithm~\ref{alg:leftSkewDaCAppBasis} is correct and has complexity
		\begin{equation*}
		\softO\big(a^{\omega-1}\max\{a,b\} \OMul{d}\big)
		\end{equation*}
		operations in the base field of the cost bound $\OMul{d}$.
	\end{theorem}
	
	\begin{proof}
		Correctness follows from Lemma~\ref{lem:appBasisProduct}, as well as the correctness of the base cases (Theorem~\ref{thm:correctness_rightLinApp} for Algorithm~\ref{alg:rightSkewLinAppBas} and Theorem~\ref{thm:correctness_leftLinApp} for Algorithm~\ref{alg:leftSkewLinAppBas}). %
		
		As for the complexity, the algorithms call themselves twice with input size $\approx d/2$.
		Taking a matrix left or right modulo $x^{d_i}$ corresponds to setting all coefficients of degree at least $d_i$ to zero in each entry. Multiplying $x^{-d_{1}}$ from the left in Line~\ref{line:mat_mult_right} of Algorithm~\ref{alg:rightSkewDaCAppBasis} requires to apply an automorphism to each polynomial coefficient, hence costs $O(ab d)$ operations in $\Fqm$. Note that this is not necessary in Algorithm~\ref{alg:leftSkewDaCAppBasis} since the monomial is multiplied from the right.

		Other operations that have a non-negligible cost are the base cases and the matrix multiplications (Lines~\ref{line:mat_mult_right} and \ref{line:mat_mult_right_2} in Algorithm~\ref{alg:rightSkewDaCAppBasis} and Lines~\ref{line:mat_mult_left} and \ref{line:mat_mult_left_2} in Algorithm~\ref{alg:rightSkewDaCAppBasis}).
		We discuss the right case, the other side follows analogously by replacing $a$ and $b$ in the complexity expression.
		The two multiplications are $\MABin \in \SkewPolys^{a \times b}$ times $\MABout_1 \in \SkewPolys^{b \times b}$ and $\MABout_1 \in \SkewPolys^{b \times b}$  times $\MABout_2 \in \SkewPolys^{b \times b}$, all matrices have degree at most $d$.
		The product $\MABin \MABout_1$ can be computed in $O(\tfrac{a}{b} b^\omega \OMul{d}) = O(a b^{\omega-1}\OMul{d})$ if $a \geq b$ and in $O(b^\omega \OMul{d})$ otherwise. The product $\MABout_1\MABout_2$ costs $O(b^\omega \OMul{d})$.
		In total, the matrix multiplications can be computed with complexity $O(\max\{a,b\} b^{\omega-1} \OMul{d})$.
		The base case, Algorithm~\ref{alg:rightSkewLinAppBas}, costs $O(\min\{a,b\}^{\omega-1}ab)$.

		Hence, we obtain the claimed complexity by the master theorem for divide-and-conquer recurrences.
	\end{proof}

\begin{remark}\label{rem:M-basis}
		Using Lemma~\ref{lem:appBasisProduct} with $d_1=1$ (i.e. the base case) and $d_2=d-1$ in an iterative manner results in right and left skew variants of~\cite[\textsf{M-Basis}]{GiorgiPolyMatrix} where the order of $\MABout$ is increased by one in each iteration.
		The complexities of the resulting algorithms are  $\softO\big(\max\{a,b\}b^{\omega-1} d^2\big)$ (right case) and $\softO\big(a^{\omega-1}\max\{a,b\} d^2\big)$ (left case) operations in $\Fqm$, respectively.
		
		This is asymptotically slower than Algorithms~\ref{alg:rightSkewDaCAppBasis} and \ref{alg:leftSkewDaCAppBasis} using skew polynomial multiplication algorithms of sub-quadratic complexity over $\Fqm$, e.g.\ \cite{caruso2017new,puchinger2017fast}.
		In particular, applying the skew \textsf{M-Basis} algorithm to the decoding problems in the remainder of the paper would not improve the asymptotic costs of the state-of-the-art decoder implementations.

		However, for small orders $d$, the skew \textsf{M-Basis} algorithm might be faster than the skew \textsf{PM-Basis} algorithm due to large hidden constants in the asymptotic expressions of asymptotically fast skew polynomial multiplication algorithms. The two methods can also be combined by calling \textsf{M-Basis} (instead of \textsf{PM-Basis}) inside \textsf{PM-Basis} as soon as $d$ is small enough.

		For completeness, we present the skew \textsf{M-Basis} algorithm and prove its complexity in Appendix~\ref{app:M-basis}.
\end{remark}

\section{Fast Decoding of Rank-Metric and Subspace Codes}\label{sec:rank_and_subspace}

We show how to speed up interpolation-based decoding of interleaved Gabidulin codes in the rank metric (Wachter-Zeh--Zeh decoder \cite{wachter2014list}) and lifted interleaved Gabidulin codes in the subspace metric (Bartz--Wachter-Zeh decoder \cite{BartzWachterZeh_ISubAMC}).

The interpolation and root-finding steps of both considered decoders are special instances of two general computational problems, which we state and relate to the decoders.
Then we present new algorithms to solve the two problems by reducing Problem~\ref{prob:general_interpolation_problem} to computing a \emph{left} approximant basis (Algorithm~\ref{alg:fast_interpolation} in Section~\ref{ssec:interpolation_speed-up}), and show that Problem~\ref{prob:general_root-finding_problem} (i.e., root finding) can be efficiently solved by a \emph{right} approximant basis (Algorithm~\ref{alg:fast_root_finding} in Section~\ref{ssec:root_finding_speed-up}).

In this section, we only use the operator evaluation of skew polynomials (cf.~Section~\ref{ssec:evaluation_maps}). %

\subsection{Computational Problems and their Relation to Decoding}\label{ssec:rank_subspace_computational_problems}

\begin{problem}[Vector (Operator) Interpolation%
]\label{prob:general_interpolation_problem}
Given $\ell,n,D \in \ZZ_{> 0}$, $\w \in \ZZ_{\geq 0}^{\ell+1}$, and $\mat{U} = [U_{i,j}] \in \Fqm^{n \times (\ell+1)}$ whose rows (called ``interpolation points'') are $\Fq$-linearly independent.
Consider the $\Fqm$-vector space $\Qspace$ (left scalar multiplication) of vectors $\Q = \left[Q_0,Q_1,\dots,Q_\ell\right]\in\SkewPolys^{\intOrder+1}$ that satisfy the following two conditions:
\begin{align}
\sum_{j=1}^{\ell+1} Q_{j-1}\!\left(U_{i,j}\right) &= 0, &&\forall \, i=1,\dots,n, \label{eq:interpolation_problem_eval} \\
\rdeg_\w(\Q) &< \degConstraint. \label{eq:interpolation_problem_deg}
\end{align}
Find left $\SkewPolys$-linearly independent $\Q^{(1)},\dots,\Q^{(\IntParam)} \in \Qspace \setminus \{\0\}$ whose left $\SkewPolys$-span contains $\Qspace$. 
\end{problem}

\begin{problem}[Vector Root Finding]\label{prob:general_root-finding_problem}
	Given $\ell,n \in \ZZ_{> 0}$, $\vec{k} \in \ZZ_{> 0}^\ell$, and vectors $\Q^{(1)},\dots,\Q^{(\IntParam)} \in \SkewPolys^{\ell+1} \setminus \{\0\}$ that are left $\SkewPolys$-linearly independent (this implies $\ell'\leq \ell+1$) and fulfill $\deg \Q^{(i)} \leq n$ for all $i$. %
	Find a basis of the $\Fqm$-linear affine space (scalar multiplication from the right)
	\begin{align}
	\Module := \big\{ &[f^{(1)},\dots,f^{(\ell)}] \in \SkewPolys^\ell \, : \, \,  \label{eq:affine_root_space} \\
	&Q_0^{(i)} + \textstyle\sum_{j=1}^{\ell} Q_j^{(i)} f^{(j)} = 0 \, \forall \, i, \, \deg f^{(j)} < k^{(j)} \, \forall \, j\big\}. \notag
	\end{align}
\end{problem}

Complexity-wise, we consider only the cases $D \in \Theta(n)$ and $\max_{i}k^{(i)} \in \Theta(n)$ since they are the most relevant for decoding.
See Section~\ref{ssec:remarks_on_generality} for a discussion on the cases $D,\max_{i}k^{(i)} \ll n$ and $D,\max_{i}k^{(i)} \gg n$.
The fastest algorithm to solve Problem~\ref{prob:general_interpolation_problem} with $D \in \Theta(n)$ is \cite{xie_linearized_2013} with a complexity of $O(\ell^2 n^2)$ over $\Fqm$.
If the first column of $\U$ consists of $\Fq$-linearly independent elements (see, e.g., Wachter-Zeh decoder \cite{wachter2014list} below), Problem~\ref{prob:general_interpolation_problem} can be solved with complexity $O(\ell^3 \OMul{\ell n})$ %
 \cite{alekhnovich_linear_2005,puchinger2018construction}.
For $\max_{i}k^{(i)} \in \Theta(n)$, Problem~\ref{prob:general_root-finding_problem} can be solved in $O(\ell^3 n^2)$ over $\Fqm$ \cite{wachter2014list} or, if $|\Module|=1$, in $O(\ell^2 n^2)$ \cite{BartzWachterZeh_ISubAMC}.

\subsubsection{Interpolation-Based Decoding of Rank-Metric Codes}\label{ssec:interleaved_gabidulin_decoder}

We recall the Wachter-Zeh--Zeh decoder and connect it to Problems~\ref{prob:general_interpolation_problem} and~\ref{prob:general_root-finding_problem}.
Let $n \leq m$ and $\ell$ be positive integers, $\vec{\alpha} = [\alpha_1,\dots,\alpha_n] \in \Fqm^n$ be a vector whose entries are linearly independent over $\Fq$, and $\k = [k^{(1)},\dots,k^{(\ell)}] \in \{1,\dots,n\}^\ell$. The corresponding interleaved Gabidulin code \cite{loidreau2006decoding} is
\begin{align*}
&\CIGab[\ell,\vec{\alpha};n,\vec{k}] := \\
&\left\{ \begin{bmatrix}
f^{(1)}(\alpha_1) & \cdots & f^{(1)}(\alpha_n) \\
\vdots & \ddots & \vdots \\
f^{(\ell)}(\alpha_1) & \cdots & f^{(\ell)}(\alpha_n) \\
\end{bmatrix} : f^{(i)} \in \SkewPolys_{< k^{(i)}} \, \forall \, i\right\}.
\end{align*}
All codewords $\C$, which are $\Fqm^{\ell \times n}$ matrices, have a corresponding message polynomial vector $\f := \big[f^{(1)},\dots,f^{(\ell)}\big]$, whose entries evaluate to the rows of $\C$ at $\vec{\alpha}$.

The codes are designed for the following generalization of the rank metric.
Fix a basis of $\Fqm$ over $\Fq$.
Then any element of $\Fqm$ can be written as a vector in $\Fq^m$ by expanding the element in this basis.
The rank weight $\wtR(\A)$ of a matrix $\A \in \Fqm^{\ell \times n}$ is the $\Fq$ rank of the matrix in $\Fq^{\ell m \times n}$ that we obtain by expanding each entry of $\A$ into a column vector. The rank distance of two matrices $\A, \B \in \Fqm^{\ell \times n}$ is the rank weight of their difference, i.e., $\dR(\A,\B) := \wtR(\A-\B)$.
Due to $\Fq^{\ell m \times n} \simeq \mathbb{F}_{q^{m\ell}}^n$, this is the usual rank metric in $\mathbb{F}_{q^{m\ell}}^n$.
See Section~\ref{ssec:intro_codes} for applications of the codes and the metric.

Let $\CIGab[\ell,\vec{\alpha};n,\vec{k}]$ be an interleaved Gabidulin code and
$\R \in \Fqm^{\ell \times n}$
be a received word. %
The interpolation step of the Wachter-Zeh--Zeh decoder solves Problem~\ref{prob:general_interpolation_problem}
with input $\ell$, $n$,
\begin{align}
\degConstraint &= n-\left\lceil\tfrac{\ell(n+1)-\sum_{i=1}^{\ell} k^{(i)}}{\ell+1}\right\rceil+1, \notag\\
\w &= [0,k^{(1)}-1,\dots,k^{(\intOrder)}-1] \in \ZZ_{\geq 0}^{\ell+1}, \quad \text{and} \label{eq:rank_interpolation_problem_input} \\
\U &= \begin{bmatrix}
\vec{\alpha}^\top & \R^\top
\end{bmatrix}
\in \Fqm^{n \times (\ell+1)}.\notag
\end{align}
This instance of the problem always has a non-trivial solution (i.e., $\Qspace \neq \{0\}$).
If the output\footnote{Our interpolation problem output differs slightly from \cite{wachter2014list}, where either one solution or an $\Fqm$-basis of $\Qspace$ is found. It is easy to see that a set of $\SkewPolys$-linearly independent vectors whose span contains $\Qspace$ does not change the root space $\Module$ compared to a full $\Fqm$-basis. Further, Problem~\ref{prob:general_interpolation_problem} and Algorithm~\ref{alg:fast_interpolation} in Section~\ref{ssec:interpolation_speed-up} can be easily adapted to output one solution.} of this problem is input to Problem~\ref{prob:general_root-finding_problem}
(root-finding step), then the space $\Module$ in Problem~\ref{prob:general_root-finding_problem} contains all message polynomial vectors $\f \in \SkewPolys^{\ell}$ of codewords $\C \in \CIGab[\ell,\vec{\alpha};n,\vec{k}]$ whose rank distance to the received words is smaller than
\begin{equation*}
\dR(\C,\R) < \tfrac{\ell}{\ell+1}\left(n-\tfrac{1}{\ell}\textstyle\sum_i k^{(i)}+1\right).
\end{equation*}
This gives a list decoder with list size at most $|\Module|$ ($\Module$ may contain vectors that do not correspond to codewords lying within the decoding radius).
Wachter-Zeh and Zeh derived an exponential upper bound on $|\Module|$ and a bound (which is close to $1$ for many parameters) on the expected size of $|\Module|$ for a received word $\R$ that is chosen uniformly at random from $\Fqm^{\ell \times n}$.\footnote{The proof of \cite[Lemma~6]{wachter2014list} derives a bound on the expected size of $|\Module|$ for a uniformly chosen received word $\R$. However, the lemma statement does not fit to the proof since it assumes that $\R = \C+\E$ for a codeword $\C$ and error $\E$ of weight at most a given value $\tau$, i.e., depending on the code and $\tau$, $\R$ cannot even attain all values of $\Fqm^{\ell \times n}$.}
The algorithm can be turned into a partial unique decoder by declaring a failure for $|\Module|>1$.

The previous-fastest realization of the decoder has complexity $O(\ell^2 n^2)$.
With the new algorithms to solve Problems~\ref{prob:general_interpolation_problem} and \ref{prob:general_root-finding_problem} in the next subsections, we get the following speed-up.

\begin{theorem}\label{thm:rank_complexity_summary}
Decoding an interleaved Gabidulin code $\CIGab[\ell,\vec{\alpha};n,\vec{k}]$ using the decoder in \cite{wachter2014list}, where
\begin{itemize}
\item the interpolation step is implemented using Algorithm~\ref{alg:fast_interpolation} (Section~\ref{ssec:interpolation_speed-up}) with input $\ell$, $n$, $D$, $\w$, $\U$ as in \eqref{eq:rank_interpolation_problem_input} and
\item the root-finding step is implemented using Algorithm~\ref{alg:fast_root_finding} (Section~\ref{ssec:root_finding_speed-up}) with input $\ell$, $n$, $\k$, and the output $\Q^{(1)},\dots,\Q^{(\IntParam)}$ of the interpolation step,
\end{itemize}
has complexity $\softO\!\left( \ell^\omega \OMul{n} \right)$ operations in the base field of the cost bound $\OMul{n}$.
\end{theorem}

\begin{proof}
Correctness and complexity follow directly from the correctness and complexity of Algorithm~\ref{alg:fast_interpolation} (Theorem~\ref{thm:fast_interpolation_correctness_complexity}) and Algorithm~\ref{alg:fast_root_finding} (Theorem~\ref{thm:fast_root_finding_correctness_complexity}) and the results in \cite{wachter2014list} (see also the brief summary above).
We only need to be careful about two points:
the entries of the first column of $\U$ are $\Fq$-linearly independent by definition of the $\alpha_i$; also, $\Q^{(1)},\dots,\Q^{(\IntParam)}$ is a valid input to Algorithm~\ref{alg:fast_root_finding} since by the choice of $D$ and the degree constraint in Problem~\ref{prob:general_interpolation_problem}, we have $\deg \Q^{(i)} \leq n$.
\end{proof}

\subsubsection{Interpolation-Based Decoding of Subspace Codes}\label{ssec:lifted_interleaved_gabidulin_decoder}
We recall the Bartz--Wachter-Zeh decoder \cite{BartzWachterZeh_ISubAMC}.
Let $\nTransmit \leq m$ and $\ell$ be positive integers, $\vec{\alpha} = [\alpha_1,\dots,\alpha_\nTransmit] \in \Fqm^\nTransmit$ be a vector whose entries are linearly independent over $\Fq$, and $\k = [k^{(1)},\dots,k^{(\ell)}] \in \{1,\dots,\nTransmit\}^\ell$. 
The corresponding lifted interleaved Gabidulin code~\cite{silva2008rank} is defined as
\begin{align*}
&\LCIGab[\ell,\vec{\alpha};\nTransmit,\vec{k}] := \\
&\left\{\spannedBy{\begin{bmatrix}
\vec{\alpha}^\top & \C^\top
\end{bmatrix}} : \C\in\CIGab[\ell,\vec{\alpha};\nTransmit,\vec{k}]\right\}
\end{align*}
where $\spannedBy{\begin{bmatrix}\vec{\alpha}^\top & \C^\top\end{bmatrix}}$ denotes the $\Fq$-linear row space of the matrix from $\Fq^{\nTransmit\times m(\ell+1)}$ obtained by expanding each entry of the matrix $\begin{bmatrix}\vec{\alpha}^\top & \C^\top\end{bmatrix} \in \Fqm^{n \times (\ell+1)}$ into a $1 \times m$ row vector over $\Fq$ using a fixed basis of $\Fqm$.
Hence, codewords are $\nTransmit$-dimensional subspaces of $\Fq^{m(\ell+1)}$.
The \emph{subspace distance} between two subspaces $\myspace{U},\myspace{V}$ of $\Fq^{m(\ell+1)}$ is defined as
\begin{equation}\label{eq:subspaceDistance}
	\Subspacedist{\myspace{U},\myspace{V}} %
	=\dim(\myspace{U})+\dim(\myspace{V})-2\dim(\myspace{U}\cap \myspace{V}).
\end{equation}
This is a natural metric in the \emph{operator channel} \cite{koetter2008coding}, which for an input subspace $\myspace{V}$ of $\dim(\myspace{V})=\nTransmit$ returns a subspace
\begin{equation}\label{eq:operatorChannel}
 \myspace{U}=\mathcal{H}_{\nTransmit-\deletions}(\myspace{V})\oplus \myspace{E},
\end{equation}
where $\mathcal{H}_{\nTransmit-\deletions}(\myspace{V})$ is a $(\nTransmit-\deletions)$-dimensional subspace of $\myspace{V}$, and $\myspace{E}$ denotes an error space of dimension~$\insertions$ with $\myspace{V}\cap\myspace{E}=\{\vec{0}\}$.
We call $\insertions$ the number of \emph{insertions} and $\deletions$ the number of \emph{deletions}.
Hence, the received space $\myspace{U}$ has dimension 
\begin{equation}
 \nReceive:=\dim(\myspace{U})=\nTransmit-\deletions+\insertions.
\end{equation}
We say that a subspace $\myspace{V}$ is \emph{$(\insertions,\deletions)$-reachable} from a subspace $\myspace{U}$ if there exists a realization of the operator channel~\eqref{eq:operatorChannel} with $\insertions$ insertions and $\deletions$ deletions that transforms the input $\myspace{V}$ to the output $\myspace{U}$.
If a space $\myspace{V}$ is $(\insertions,\deletions)$-reachable from a space $\myspace{U}$, then we have that $\Subspacedist{\myspace{U},\myspace{V}}=\insertions+\deletions$.
See Section~\ref{ssec:intro_codes} for applications of the codes and the metric.

Let $\LCIGab[\ell,\vec{\alpha};\nTransmit,\vec{k}]$ be a lifted interleaved Gabidulin code and $\myspace{U} \subseteq \Fq^{m(\ell+1)}$ of dimension $\dim(\myspace{U}) = \nReceive$ be a received subspace, given in form of a basis $\U \in \Fqm^{\nReceive \times (\ell+1)}$ with $\myspace{U} = \spannedBy{\U}$.
The interpolation step of the Bartz--Wachter-Zeh decoder asks for a solution $\Q^{(1)},\ldots,\Q^{(\intOrder)}$ to Problem~\ref{prob:general_interpolation_problem} with input $\ell$, $n = \nReceive$, the basis $\U \in \Fqm^{\nReceive \times (\ell+1)}$,
\begin{align}
\degConstraint &= \left\lceil\frac{\nReceive+\sum_{i=1}^\intOrder k^{(i)}-\intOrder+1}{\intOrder+1}\right\rceil, &&\text{and} \notag \\
\w &= [0,k^{(1)}-1,\dots,k^{(\intOrder)}-1] \in \ZZ_{\geq 0}^{\ell+1}. \label{eq:subspace_interpolation_input}
\end{align}
Due to the choice of $D$, this problem instance always has a solution (i.e., $\Qspace \neq \{0\}$), cf.~\cite{BartzWachterZeh_ISubAMC}.
The root-finding step consists of solving Problem~\ref{prob:general_root-finding_problem} with input $\ell, n= \nTransmit, \vec{k}$, as well as the $\Q^{(1)},\ldots,\Q^{(\intOrder)}$ computed above.
Then, the space $\Module$ contains all message polynomial vectors $\vec{f}\in\SkewPolys^{\intOrder}$ corresponding to codewords $\myspace{V}\in\LCIGab[\ell,\vec{\alpha};\nTransmit,\vec{k}]$ that are $(\insertions,\deletions)$-reachable from the received space $\myspace{U}$ with $\dim(\myspace{U})=\nReceive=\nTransmit-\deletions+\insertions$ for all $\insertions$ and $\deletions$ satisfying\footnote{Due to $\Subspacedist{\myspace{U},\myspace{V}}\leq\insertions+\intOrder\deletions$, all $\f$ of codewords with $\Subspacedist{\myspace{U},\myspace{V}}<\intOrder(\nTransmit-\bar{k}+1)$ are in $\Module$, but this is a weaker condition than \eqref{eq:decConditionLIGab}.}
\begin{equation}\label{eq:decConditionLIGab}
	\insertions+\intOrder\deletions<\intOrder\left(\nTransmit-\bar{k}+1\right).
\end{equation}
This gives a list decoder with list size at most $|\Module|$.

Similar to interleaved Gabidulin codes there exists an upper bound on $|\Module|$ (which is exponential in the code parameters) and a bound on the expected size\footnote{%
  As in the Wachter-Zeh--Zeh decoder, drawing a received word uniformly at random usually does not correspond to choosing a codeword and a low-weight error uniformly at random, and hence this result is not directly applicable to most channels considered in the literature.%
} of $|\Module|$ (which is close to $1$ for many parameters) for a received word $\R$ that is drawn uniformly at random from the set of $\nReceive$-dimensional subspaces of $\Fq^{m(\ell+1)}$, see~\cite{BartzWachterZeh_ISubAMC, BartzDissertation}.
The algorithm can also be interpreted as a probabilistic unique decoder by declaring a decoding failure if $|\Module|>1$, cf.~\cite{BartzWachterZeh_ISubAMC}.

Using the new algorithms to solve Problems~\ref{prob:general_interpolation_problem} and \ref{prob:general_root-finding_problem} in the next subsections, we can reduce the complexity of the decoder from $O(\ell^2 \max\{\nReceive,\nTransmit\}^2)$ \cite{BartzWachterZeh_ISubAMC} to the following expression.

\begin{theorem}\label{thm:subspace_complexity_summary}
Decoding a received subspace of dimension $\nReceive$ in a lifted interleaved Gabidulin code $\LCIGab[\ell,\vec{\alpha};\nTransmit,\vec{k}]$ using the decoder in \cite{BartzWachterZeh_ISubAMC}, where
\begin{itemize}
\item the interpolation step is implemented using Algorithm~\ref{alg:fast_interpolation} (Section~\ref{ssec:interpolation_speed-up}) with input $\ell$, $\nReceive$, $D$, $\w$ as in \eqref{eq:subspace_interpolation_input}, and a basis $\U$ of the received space and
\item the root-finding step is implemented using Algorithm~\ref{alg:fast_root_finding} (Section~\ref{ssec:root_finding_speed-up}) with input $\ell$, $\nReceive$, $\k$, and the output $\Q^{(1)},\dots,\Q^{(\IntParam)}$ of the interpolation step,
\end{itemize}
has complexity $\softO\!\left(\intOrder^\omega\OMul{\max\{\nTransmit,\nReceive\}}\right)$ operations in the base field of the cost bound $\OMul{n}$, plus $O\!\left(\intOrder m \nReceive^{\omega-1} \right)$ operations in $\Fq$.
		
\end{theorem}

\begin{proof}
Correctness follows from the correctness of Algorithm~\ref{alg:fast_interpolation} (Theorem~\ref{thm:fast_interpolation_correctness_complexity}) and Algorithm~\ref{alg:fast_root_finding} (Theorem~\ref{thm:fast_root_finding_correctness_complexity}) and the results in \cite{BartzWachterZeh_ISubAMC} (see also the brief summary above).
Note that the vectors $\Q^{(1)},\dots,\Q^{(\IntParam)}$ are a valid input to Algorithm~\ref{alg:fast_root_finding} since $\deg \Q^{(i)} \leq \nReceive$ by Problem~\ref{prob:general_interpolation_problem}.

The complexity is $\softO\!\left( \intOrder^\omega \OMul{D+n}\right)$ operations in the base field of the cost bound $\OMul{n}$ plus $O(\intOrder m n^{\omega-1})$ operations in $\Fq$ for the interpolation step and $\softO\!\left( \ell^\omega \OMul{n + \max_{i}k^{(i)}} \right)$ operations in the base field of the cost bound $\OMul{n}$ for the root-finding step by Theorems~\ref{thm:fast_interpolation_correctness_complexity} and \ref{thm:fast_root_finding_correctness_complexity}, respectively.
The input variables $n,D,\k$ of the two computational problems are connected to the code and channel parameters $\nTransmit,\nReceive$ as follows.
We have $n = \nReceive$, $D \in O\!\left(\max\!\left\{n,\max_{i}k^{(i)}\right\}\right)$, and $\max_i\{k^{(i)}\} \leq \nTransmit$, which implies the dependency on $\max\{\nTransmit,\nReceive\}$.
\end{proof}

\subsection{A New Algorithm for the Interpolation Step}\label{ssec:interpolation_speed-up}

We relate the interpolation step (Problem~\ref{prob:general_interpolation_problem}) to finding a left approximant bases of a matrix $\MABin$ that is constructed from (operator) interpolation and annihilator polynomials depending on the interpolation points (i.e., the input matrix $\U$ of the problem).

To construct the matrix $\MABin$, we first need to transform the interpolation points as in the following lemma.
Note that we apply $\Fq$-linear elementary row operations to $\U$, which due to the $\Fq$-linearity of skew polynomials does not change the interpolation condition, \eqref{eq:interpolation_problem_eval}, of Problem~\ref{prob:general_interpolation_problem}.

\begin{lemma}\label{lem:intProblemSubspace_Zi_matrices}
Consider an instance of Problem~\ref{prob:general_interpolation_problem}.
Using $\Fq$-linear elementary row operations, we can transform $\U$ into a matrix of the form
\begin{align}
\U' = \left[
\def\arraystretch{1.7}
\begin{array}{ccccc}
\multicolumn{1}{c|}{\0_{\nu_1\times a_1}} & & \U^{(1)} & & \\\hline
\multicolumn{2}{c|}{\0_{\nu_2\times a_2}} & & \U^{(2)} &   \\\hline
\multicolumn{3}{c|}{\0_{\nu_3\times a_3}} & \multicolumn{2}{|c}{\U^{(3)}} \\\hline
\multicolumn{5}{c}{\vdots} \\\hline
\multicolumn{4}{c|}{\0_{\nu_\varrho\times a_\varrho}} & \multicolumn{1}{|c}{\U^{(\varrho)}} 
\end{array}
\right], \label{eq:intProblemSubspace_Zi_matrices} %
\end{align}
where $1 \leq \varrho \leq \intOrder+1$ and we have $\U^{(i)} \in \Fqm^{\nu_i\times (\intOrder+1-a_i) }$ for $i=1,\dots,\varrho$, with
\begin{itemize}
\item $0\leq a_1 < a_2 < \cdots < a_\varrho < \intOrder+1$,
\item $1 \leq \nu_i \leq n$ such that $\sum_{i=1}^{\varrho} \nu_i = n$, and
\item the entries of the first column of $\U^{(i)}$ are linearly independent over $\Fq$ for each $i$.
\end{itemize}
The matrix $\U'$ can be obtained with $O\big(\intOrder m n^{\omega-1}\big)$ operations in $\Fq$.
\end{lemma}

\begin{proof}
This can be done by expanding each entry of $\U\in\Fqm^{n\times (\intOrder+1)}$ into a row vector over $\Fq$ of length $m$, by transforming this $n\times m(\intOrder+1)$ matrix into row echelon form, and then mapping the resulting matrix back to an $n \times (\intOrder+1)$ matrix over $\Fqm$.
The structure of $\U'$ then follows immediately from the row echelon form of the expanded matrix (e.g., the width $\nu_i$ of the matrix $\U^{(i)}$ will be the number of pivots in the columns $a_im+1,\dots,(a_i+1)m$ of the expanded matrix).
There will be no zero rows since the rows of $\U$ are $\Fq$-linearly independent.
The complexity follows by \cite[Theorem~2.10]{storjohann2000algorithms}.
\end{proof}

The following lemmas connect Problem~\ref{prob:general_interpolation_problem} to a problem of computing an approximant basis.
Since the first columns of all the matrices $\U^{(i)}$ are $\Fq$-linearly independent, the polynomials $G^{(i)}$ and $R^{(i)}_j$ in the following lemma are well-defined.

\begin{lemma}\label{lem:interpolation_problem_kernel_basis}
Let $\U^{(1)},\dots,\U^{(\varrho)}$ be defined as in Lemma~\ref{lem:intProblemSubspace_Zi_matrices}.
Then, $\Q = [Q_0,\dots,Q_\intOrder] \in \SkewPolys^{\intOrder+1}$ satisfies Condition~\eqref{eq:interpolation_problem_eval} in Problem~\ref{prob:general_interpolation_problem} if and only if there is a vector $\vec{\chi} \in \SkewPolys^{\varrho}$ with
\begin{align}
\begin{bmatrix}
\Q & \vec{\chi}
\end{bmatrix}
\cdot
\MABin
= \0, \label{eq:intProblemSubspace_approximant_bases_reformulation}
\end{align}
where $\MABin \in \SkewPolys^{(\intOrder+1+\varrho) \times \varrho}$ is a matrix whose $i$-th column, for $i=1,\dots,\varrho$, is of the form
\begin{align*}
\left[
\begin{array}{c}
\0_{a_i \times 1} \\
\hline
1 \\
R^{(i)}_{a_i+2} \\
\vdots \\
R^{(i)}_{\intOrder+1} \\
\hline
\0_{(i-1) \times 1} \\
\hline
G^{(i)} \\
\hline
\0_{(\varrho-i) \times 1} \\
\end{array}
\right]
\end{align*}
where, for all $i=1,\dots,\varrho$ and $j=a_i+2,\dots,\intOrder+1$,
\begin{align*}
G^{(i)} &:= \MSPop_{\left\langle U_{1,1}^{(i)},\dots,U_{\nu_i,1}^{(i)}\right\rangle} \\
R^{(i)}_j &:= \IPop{\left\{\left(U_{\kappa,1}^{(i)}, U_{\kappa,j-a_i}^{(i)}\right)\right\}_{\kappa=1}^{\nu_i}}.
\end{align*}
\end{lemma}

\begin{proof}
A vector $\Q = [Q_0,\dots,Q_{\intOrder}] \in \SkewPolys^{\intOrder+1}$ satisfies Condition~\eqref{eq:interpolation_problem_eval} in Problem~\ref{prob:general_interpolation_problem} on all rows of $\U$ if and only if each sub-block $[Q_{a_i},\dots,Q_{\intOrder}]$ satisfies \eqref{eq:interpolation_problem_eval} on the rows of $U^{(i)}$.
Using $G^{(i)}$ and $R^{(i)}_j$ as above, we can rewrite this condition, restricted to $\U^{(i)}$, as
\begin{align}
&\sum_{j=a_i+1}^{\ell+1} Q_{j-1}\!\left(U_{\kappa,j-a_i}^{(i)}\right) = 0 \quad \forall \kappa=1,\dots,\nu_i \label{eq:i-th_interpolation_condition} \\
&\Leftrightarrow \: Q_{a_i}\!\left(U_{\kappa,1}^{(i)}\right) + \sum_{j=a_i+2}^{\ell+1} Q_{j-1}\!\left(R_j^{(i)}\!\left(U_{\kappa,1}^{(i)}\right)\right) = 0 \quad \forall \kappa \notag \\
&\Leftrightarrow \: \left(Q_{a_i} + \sum_{j=a_i+2}^{\ell+1} Q_{j-1} R_j^{(i)} \right)\!\left(U_{\kappa,1}^{(i)}\right) = 0 \quad \forall \kappa \notag \\
&\Leftrightarrow \: Q_{a_i} + \sum_{j=a_i+2}^{\ell+1} Q_{j-1} R_j^{(i)} \equiv 0 \quad \modr \underbrace{\MSPop_{\left\langle U_{1,1}^{(i)},\dots,U_{\nu_i,1}^{(i)}\right\rangle}}_{= \, G^{(i)}} \notag \\
&\Leftrightarrow \: \exists\, \chi_i \in \SkewPolys \, : \notag\\
&\quad \quad \quad \, Q_{a_i}\, + \sum_{j=a_i+2}^{\ell+1} Q_{j-1} R_j^{(i)} + \chi_i G^{(i)} = 0 \notag \\
&\Leftrightarrow \: \exists\, \chi_i \in \SkewPolys \, : \notag \\
&\quad \quad \quad \, \begin{bmatrix}
Q_{a_i} & \cdots & Q_{\intOrder} & \chi_i
\end{bmatrix} \cdot
\begin{bmatrix}
1 \\
R^{(i)}_{a_i+2} \\
\vdots \\
R^{(i)}_{\intOrder+1} \\
G^{(i)}
\end{bmatrix} = 0. \label{eq:left_component_kernels}
\end{align}
This is equivalent to \eqref{eq:intProblemSubspace_approximant_bases_reformulation} since the $\chi_i$'s are independent of each other, but the $Q_j$ are the same for each $i$.
\end{proof}

\begin{example}
We give two examples for the matrix $\MABin$ as in Lemma~\ref{lem:interpolation_problem_kernel_basis}.

For $\nu_1=n$ and $a_1=0$ ($\varrho=1$), the first column of the matrix $\U$ already consists of linearly independent elements. This is an important special case since it is always fulfilled for the Wachter-Zeh--Zeh decoder (interleaved Gabidulin codes, see Section~\ref{ssec:interleaved_gabidulin_decoder}). In this case, $\MABin$ of Lemma~\ref{lem:interpolation_problem_kernel_basis} has the form
\begin{align*}
\MABin = \begin{bmatrix}
1 \\ R_2^{(1)} \\ R_3^{(1)} \\ \vdots \\ R_{\intOrder+1}^{(1)} \\ G^{(1)}
\end{bmatrix}
\in \SkewPolys^{(\intOrder+2)\times 1}.
\end{align*}

For $a_i = i-1$, $i=1,\dots,\varrho$, the matrix $\MABin$ has the form
\begin{align}
\MABin = \begin{bmatrix} 
1 &   &   &   &   \\
R^{(1)}_2 & 1 &   &   &   \\
R^{(1)}_3 & R^{(2)}_3 & 1 &   &   \\
\vdots & \vdots & \ddots & \ddots &   \\
R^{(1)}_{\varrho} & R^{(2)}_{\varrho} & R^{(3)}_{\varrho} & \cdots & 1 \\
R^{(1)}_{\varrho+1} & R^{(2)}_{\varrho+1} & R^{(3)}_{\varrho+1} & \cdots & R^{(\varrho)}_{\varrho+1} \\
\vdots & \vdots & \ddots & \ddots & \vdots \\
R^{(1)}_{\intOrder+1} & R^{(2)}_{\intOrder+1} & R^{(3)}_{\intOrder+1} & \cdots & R^{(\varrho)}_{\intOrder+1} \\
G^{(1)} & & & & \\
 & G^{(2)} & & & \\
 & & G^{(3)} & & \\
 & & & \ddots & \\
 & & & & G^{(\varrho)} \\
\end{bmatrix}. \label{eq:J_example_full}
\end{align}
In general, $\MABin$ has a form as in \eqref{eq:J_example_full}, where we delete the $j$-th column and $(\intOrder+1+j)$-th row (and rename the superscript indices accordingly) if there is no $i$ with $a_i=j-1$.
\end{example}

\begin{remark}
All vectors $\Q = [Q_0,\dots,Q_{\intOrder}] \in \SkewPolys^{\intOrder+1}$ satisfying Condition~\eqref{eq:interpolation_problem_eval} form a left $\SkewPolys$-module (see also \cite{puchinger2017row}).
Lemma~\ref{lem:interpolation_problem_kernel_basis} states that this module is the left kernel of $\MABin$, restricted to the first $\ell+1$ coordinates.
Furthermore, it is the intersection of the left kernels of the columns of the matrix $\MABin$, which for $i=1,\dots,\varrho$ are the modules consisting of all vectors that, when restricted to the first $\ell+1$ coordinates, satisfy \eqref{eq:interpolation_problem_eval} with respect an alternative matrix of interpolation points of the form $\left[ \0_{\nu_i,a_i} \mid \U^{(i)} \right] \in \SkewPolys^{\nu_i \times (\ell+1)}$.
\end{remark}

\begin{lemma}\label{lem:intProblemSubspace_approximant_basis_problem}
Let $\MABin$ be defined as in Lemma~\ref{lem:interpolation_problem_kernel_basis}, $\w \in \ZZ_{\geq 0}$, $D \in \ZZ_{> 0}$.
For $w_\mathrm{min} := \min_{i=1,\dots,\intOrder+1} \{w_i\}$, set $d := D-w_\mathrm{min}+n$ and
\begin{align*}
\s := [w_1,\dots,w_{\intOrder+1},w_\mathrm{min},\dots,w_\mathrm{min}] \in \ZZ_{\geq 0}^{\intOrder+1+\varrho}.
\end{align*}
Then, for $\Q \in \SkewPolys^{\intOrder+1}$ and $\vec{\chi} \in \SkewPolys^\varrho$, we have
\begin{align}
\begin{bmatrix}
\Q & \vec{\chi}
\end{bmatrix} \MABin &= \0 \quad \text{and} \label{eq:intProblemSubspace_kernel_condition} \\
\rdeg_\w \Q &< D \label{eq:intProblemSubspace_degree_restriction}
\end{align}
if and only if
\begin{align}
\begin{bmatrix}
\Q & \vec{\chi}
\end{bmatrix} \MABin \equiv \0 \; &\modr \; x^d \quad \text{and} \label{eq:intProblemSubspace_congruence_relation} \\
\rdeg_\s \begin{bmatrix}
\Q & \vec{\chi}
\end{bmatrix} &< D. \label{eq:intProblemSubspace_degree_restriction_entire_vector}
\end{align}
\end{lemma}

\begin{proof}
Let $\begin{bmatrix}
\Q & \vec{\chi}
\end{bmatrix}$ satisfy \eqref{eq:intProblemSubspace_kernel_condition} and \eqref{eq:intProblemSubspace_degree_restriction}.
Then, obviously \eqref{eq:intProblemSubspace_congruence_relation} holds.
It is left to show the degree constraint.
We have for the entries of $\vec{\chi} = [\chi_1,\dots,\chi_\varrho]$
\begin{equation*}
\deg \chi_i \leq \max_{j=i,\dots,\intOrder+1} \{\deg Q_{j-1}\}-1
\end{equation*}
since we can rewrite \eqref{eq:intProblemSubspace_kernel_condition} into
\begin{align*}
-\chi_i G^{(i)} = Q_{i-1} + \sum_{j=i+1}^{\intOrder+1} Q_{j-1} R^{(i)}_j \quad \forall \, i=1,\dots,\varrho.
\end{align*}
Due to $\deg G^{(i)} = \nu_i$ and $\deg R^{(i)}_j \leq \nu_i-1$, we get the claimed degree bound on the $\chi_i$.
Hence, we have
\begin{align*}
\rdeg_\s \begin{bmatrix}
\Q & \vec{\chi}
\end{bmatrix} = \max\big\{ \rdeg_\w \Q, \, \underbrace{w_\mathrm{min} + \textstyle\max_{i}\{ \deg \chi_i\}}_{\leq \rdeg_\w \Q} \big\} < D.
\end{align*}

For the other direction, the degree bound is obvious.
As for the equality, the $i$-th entry (for $i=1,\dots,\varrho$) of $\begin{bmatrix}
\Q & \vec{\chi}
\end{bmatrix} \MABin$ is
$Q_{i-1} + \sum_{j=i+1}^{\intOrder+1} Q_{j-1} R^{(i)}_j + \chi_i G^{(i)}$,
where
\begin{align*}
\deg Q_{i-1} &\leq D-w_i-1 < D-w_\mathrm{min} \leq d, \\
\deg \big(Q_{j-1} R^{(i)}_j\big) &\leq D-w_j+\nu_i-2 < D-w_\mathrm{min} + n = d, \\
\deg \big(\chi_i G^{(i)}\big) &\leq D-w_\mathrm{min}-1+\nu_i < D-w_\mathrm{min} + n = d,
\end{align*}
thus $\rdeg\left(\begin{bmatrix}
\Q & \vec{\chi}
\end{bmatrix} \MABin\right) < d$. Hence, we have not only $\begin{bmatrix}
\Q & \vec{\chi}
\end{bmatrix} \MABin \equiv \0 \modr x^d$, but also $\begin{bmatrix}
\Q & \vec{\chi}
\end{bmatrix} \MABin =\0$.
\end{proof}

Lemmas~\ref{lem:interpolation_problem_kernel_basis} and \ref{lem:intProblemSubspace_approximant_basis_problem} combined imply a strategy for finding a basis of all solutions of Problem~\ref{prob:general_interpolation_problem}: compute a left approximant basis of $\MABin$ (both as defined in Lemma~\ref{lem:interpolation_problem_kernel_basis}) with respect to the shift vector $\s$ and order $d$ (as defined in Lemma~\ref{lem:intProblemSubspace_approximant_basis_problem}).
This strategy is outlined in Algorithm~\ref{alg:fast_interpolation} and we give its complexity in Theorem~\ref{thm:fast_interpolation_correctness_complexity}.

\begin{algorithm}[ht]
	\caption{Fast Interpolation Algorithm}
	\label{alg:fast_interpolation}
	\SetKwInOut{Input}{Input}\SetKwInOut{Output}{Output}
	\Input{Instance of Problem~\ref{prob:general_interpolation_problem}: 
	$\ell,n,D \in \ZZ_{> 0}$, shift vector $\w \in \ZZ_{\geq 0}^{\ell+1}$, and $\mat{U} = [U_{i,j}] \in \Fqm^{n \times (\ell+1)}$ with $\Fq$-linearly independent rows.
	}
	\Output{If it exists, a solution of Problem~\ref{prob:general_interpolation_problem}. Otherwise, ``no solution''.}
	\If{elements in first column of $\U$ are $\Fq$-lin.\ ind.}{
		$\U^{(1)} \gets \U$, $\varrho \gets 1$, $\nu_1 \gets 1$, $a_1 \gets 0$
super	} \Else{
		$\U^{(i)} \in \Fqm^{\nu_i \times (\ell+1-a_i)}$ for $i=1,\dots,\varrho$ $\gets$ compute as in Lemma~\ref{lem:intProblemSubspace_Zi_matrices} \label{line:intProblemSubspace_compute_Zi} \\
	}
	\For{$i=1,\dots,\varrho$}{
		$G^{(i)} \gets \MSPop_{\langle U_{1,1}^{(i)},\dots,U_{\nu_i,1}^{(i)}\rangle}$ \\
		\For{$j=a_i+2,\dots,\intOrder+1$}{
			$R^{(i)}_j \gets \IPop{\left\{\left(U_{\kappa,1}^{(i)}, U_{\kappa,j-a_i}^{(i)}\right)\right\}_{\kappa=1}^{\nu_i}}$
		}
	}
	$\MABin \gets$ set up matrix from the $G^{(i)}$ and $R^{(i)}_j$ as in Lemma~\ref{lem:interpolation_problem_kernel_basis} \\
	$w_\mathrm{min} \gets \min_{i=1,\dots,\intOrder+1} \{w_i\}$ \\
	$d \gets D-w_\mathrm{min}+n$ \\
	$\s \gets [w_1,\dots,w_{\intOrder+1},w_\mathrm{min},\dots,w_\mathrm{min}] \in \ZZ_{\geq 0}^{\intOrder+1+\varrho}$ \\
	$\B \gets$ left \MABnameFullStandard \hfill \label{line:interpolation_minimal_approx_basis}  \myAlgoComment{Algorithm~\ref{alg:leftSkewDaCAppBasis} in Section~\ref{sec:order_bases}}
	$\{i_1,\dots,i_{\ell'}\} \gets$ indices of rows of $\B$ with $\s$-shifted row degree $<D$ \\
	\If{$\ell'>0$}{
		\For{$j=1,\dots,\ell'$}{
			$\Q^{(j)} \gets \left[B_{i_j,1}, \dots,B_{i_j,\ell+1}\right]$
		}
		\Return{$\Q^{(1)},\dots,\Q^{(\ell')}$}
	} \Else{
		\Return{``no solution''}
	}
\end{algorithm}

\begin{theorem}\label{thm:fast_interpolation_correctness_complexity}
Algorithm~\ref{alg:fast_interpolation} is correct.
For the complexity, assume $D \in \Theta(n)$.
If the first column of the input matrix $\U$ consists of $\Fq$-linearly independent elements, it can be implemented with complexity
\begin{align*}
\softO\!\left( \intOrder^\omega \OMul{n}\right)
\end{align*}
operations in the base field of the cost bound $\OMul{n}$.
Otherwise, it costs
\begin{align*}
\softO\!\left( \intOrder^\omega \OMul{n} \right)
\end{align*}
operations in the base field of the cost bound $\OMul{n}$ plus $O(\intOrder m n^{\omega-1} )$ operations in $\Fq$. %
\end{theorem}

\begin{proof}
Correctness follows by Lemmas~\ref{lem:interpolation_problem_kernel_basis} and \ref{lem:intProblemSubspace_approximant_basis_problem}, and the fact that $\B$ is in $\s$-ordered weak Popov form.
The latter property implies that the left span of the rows of $\B$ indexed by $i_1,\dots,i_{\ell'}$ includes all vectors satisfying both \eqref{eq:intProblemSubspace_congruence_relation} and \eqref{eq:intProblemSubspace_degree_restriction_entire_vector}.
Furthermore, by Lemma~\ref{lem:intProblemSubspace_approximant_basis_problem} these rows are in the left kernel of $\MABin$ (hence, if the row is $\begin{bmatrix}
\Q & \vec{\chi}
\end{bmatrix} \neq \0$ we have $\deg \Q > \deg \vec{\chi}$ due to $\deg R_j^{(i)} < \deg G^{(i)}$ for all $j$), and due to the choice of $\s$, the $\s$-pivots of the rows of $\B$ indexed by $i_1,\dots,i_{\ell'}$ are in the first $\ell+1$ positions. This means that the $\Q^{(i)}$ (the restrictions of these rows to the first $\ell+1$ components) have distinct $\w$-pivots, and are linearly independent.
Hence, $\Q^{(1)},\dots,\Q^{(\ell')}$ are a solution of Problem~\ref{prob:general_interpolation_problem}.

Recall from Section~\ref{ssec:cost_skew_poly_operations} that the annihilator polynomials $G^{(i)}$ and interpolation polynomials $R^{(i)}_j$ can be computed in $\softO(\OMul{\nu_i})$ operations in the base field of the cost bound $\OMul{n}$ each.
Computing all the polynomials $G^{(i)}$ and $R^{(i)}_j$ with $i=1,\dots,\varrho$ and $j=i+1,\dots,\intOrder+1$ hence costs at most
\begin{align*}
\softO\!\left( \intOrder \sum_{i=1}^{\varrho} \OMul{\nu_i} \right) \subseteq \softO\!\left(\intOrder\OMul{n}\right)
\end{align*}
operations in the base field of the cost bound $\OMul{n}$, since $\sum_{i=1}^{\varrho} \nu_i = n$ and $\OMul{\cdot}$ is a convex function.

Checking whether the first column of $\U$ has $\Fq$-rank $n$ can be done by computing the remainder annihilator polynomial $A := \MSPop_{\langle U_{1,1}, \dots, U_{n,1}\rangle}$ of the entries.
The $U_{i,1}$ are linearly independent if and only if $\deg A = n$.
This check can be done in $\softO(\OMul{n})$ operations in the base field of the cost bound $\OMul{n}$ (cf.~Section~\ref{ssec:cost_skew_poly_operations}).
Only if the entries are linearly independent, we need to compute the matrices $\U^{(i)}$ in Line~\ref{line:intProblemSubspace_compute_Zi}.
This costs $O(\intOrder m n^{\omega-1})$ operations in $\Fq$ (cf.~Lemma~\ref{lem:intProblemSubspace_Zi_matrices}).

By definition of $G^{(i)}$ and $R^{(i)}_j$, we have $\deg \MABin \leq n$. Due to $d \leq D+n$, 
Line~\ref{line:interpolation_minimal_approx_basis} costs $\softO(\ell^\omega \OMul{n})$ operations in the base field of the cost bound $\OMul{n}$ by Theorem~\ref{thm:correctness_DaCApp} in Section~\ref{sec:order_bases}.
\end{proof}

Algorithm~\ref{alg:fast_interpolation} can also be phrased in the language of row reduction of an interpolation module basis (cf.~\cite{puchinger2017row,puchinger2017alekhnovich}) instead of approximant bases computation.
We show in Appendix~\ref{app:module_description} how to construct a suitable module basis using the tools developed in this section.

\subsection{A New Algorithm for the Root-Finding Step}\label{ssec:root_finding_speed-up}

The following lemma relates Problem~\ref{prob:general_root-finding_problem} to computing a right approximant basis.

\begin{lemma}\label{lem:rootfinding_by_order_bases}
	Consider an instance of Problem~\ref{prob:general_root-finding_problem}, with $\khat := \max_i\{k^{(i)}\}$, and choose
	\begin{align}
	\MABin &:= \begin{bmatrix}
	Q_0^{(1)} 	& Q_1^{(1)} & \dots 	& Q_\ell^{(1)} \\
	\vdots 		& \vdots	& \ddots 	& \vdots \\
	Q_0^{(\IntParam)} 	& Q_1^{(\IntParam)} & \dots 	& Q_\ell^{(\IntParam)}
	\end{bmatrix} \in \SkewPolys^{\IntParam \times (\ell+1)} \label{eq:root-finding_MABin_matrix}\\
	\s &:= \begin{bmatrix}
	\khat & \khat-k^{(1)}+1 & \dots & \khat-k^{(\ell)}+1
	\end{bmatrix} \in \ZZ_{\geq 0}^{\ell+1} \\
	d &:= \max_{i,j} \left\{\deg Q_j^{(i)}\right\} + \khat.
	\end{align}
	Let $\MABout \in \SkewPolys^{(\ell + 1) \times (\ell+1)}$ be a right \MABnameFullStandard. %
	Then, with $\t = \cdeg_\s (\MABout)$, the root space $\Module$ defined in \eqref{eq:affine_root_space} of Problem~\ref{prob:general_root-finding_problem} satisfies
	\begin{align}
     	\Module &= \big\{ \big[f^{(1)},\dots,f^{(\ell)}\big]^\top \, : \, [f^{(0)},\dots,f^{(\ell)}]^\top = \MABout\v, \label{eq:solSet} \\
&\v \in \SkewPolys^{(\ell+1) \times 1} \textrm{ with } \cdeg_\t \v \leq \khat  \textrm{ and } f^{(0)} = 1 \ \big\} . \notag
	\end{align}
\end{lemma}

\begin{proof}
	By Lemma~\ref{lem:minimality_column_wpf} then for any $\v \in \SkewPolys^{(\ell+1) \times 1}$, we have $\cdeg_\s(\MABout \v) = \max_{i=1,\ldots,\ell+1}\{\deg (v_i) + t_i\} = \cdeg_\t \v$.

	$\subseteq$:
	Note that $\Module$ consists of those vectors of the right-kernel of $\MABin$ having $\s$-degree at most $\khat$ and first element being $1$.
	Any such kernel vector $\f$ of $\MABin$ is in the column space of $\MABout$ by definition of approximant basis, so let $\v$ be such that  $\f = \MABout\v$.
	But then we have $\cdeg_\t \v = \cdeg_\s(\MABout\v) \leq \khat$.

	$\supseteq$:
	Let $\v \in \SkewPolys^{\ell+1}$ with $\cdeg_\t(\v) \leq \khat$.
	Then $\cdeg_\s(\MABout\v) \leq \khat$, i.e.~$\cdeg(\MABout\v) \leq \khat - \min(\s) < \khat$.
	Since $\MABout$ is an approximant basis of $\MABin$, then $\MABin \MABout \v \equiv 0 \modl x^{d}$.
	But $\cdeg(\MABin\MABout\v) \leq \max_{i,j}(\deg Q^{(i)}_j) + \cdeg(\MABout\v) < d$, and hence we can conclude $\MABin \MABout \v = 0$.
	In other words, $\MABout \v$ is a right kernel vector of $\MABin$.
	Since it also has $\s$-degree at most $\khat$, it must be in $\Module$ as long as its first component is $1$.
\end{proof}

Lemma~\ref{lem:rootfinding_by_order_bases} gives an implicit description of the root space $\Module$.
The following lemma shows how to explicitly compute a basis of the affine root space from $\MABout$.

\begin{lemma}\label{lem:rank_root-finding_affine_basis}
Let $\MABout$ and $\t = \cdeg_\s(\MABout)$ be defined as in Lemma~\ref{lem:rootfinding_by_order_bases}.
Denote by $[B_{0,i},\dots,B_{\ell,i}]^\top$ the $i$-th column of $\MABout$, for $i=1,\dots,\ell+1$.
Let $\Jset$ be the set of indices of columns of $\MABout$ which have $\s$-degree at most $\khat$, i.e.~$\forall i \in \Jset$ we have $t_i \leq \khat$, and let $\Iset \subseteq \Jset$ be those indices where the first entry of the corresponding column of $\MABout$ is not zero.

If $\Iset = \emptyset$, then $\Module=\emptyset$.
Otherwise, choose some $i^* \in \Iset$,
denote by $i_1,\dots,i_\iota$ the distinct elements of $\Iset\setminus \{i^*\}$ and by $j_1,\dots,j_\tau$ the distinct elements of $\Jset \setminus \Iset$, respectively.
Define
\begin{align}
\g^* &:= \tfrac{1}{B_{0,i^*}}[B_{1,i^*}, \dots, B_{\ell,i^*}]^\top \label{eq:rank_root_finding_explicit_basis_1} \\
\g^{(r)} &:= [B_{1,i_r}, \dots, B_{\ell,i_r}]^\top - \tfrac{B_{0,i_r}}{B_{0,i^*}}[B_{1,i^*}, \dots, B_{\ell,i^*}]^\top \label{eq:rank_root_finding_explicit_basis_2} 
\end{align}
for $r=1,\dots,\iota$. For $\delta=\iota+\sum_{i =1}^{\tau} (\khat-t_{j_i}+1)$, define the vectors $\g^{(\iota+1)},\dots,\g^{(\delta)} \in \SkewPolys^\ell$ as
\begin{align*}
\g^{\left(\iota+\sum_{i'=1}^{i-1}(\khat-t_{j_i}+1)+j+1\right)} = [B_{1,j_i}, \dots, B_{\ell,j_i}]^\top x^j,
\end{align*}
where $i=1,\dots,\tau$ and $j=0,\dots,\khat-t_{j_i}$.
Then, $\g^{(1)},\dots,\g^{(\delta)}$ are right linearly independent over $\Fqm$ and
\begin{align}
\Module = \g^* + \langle \g^{(1)}, \dots, \g^{(\delta)} \rangle_{\Fqm, \mathrm{right}}, \label{eq:rank_root_finding_explicit_basis_3} 
\end{align}
where $\langle \cdot \rangle_{\Fqm, \mathrm{right}}$ denotes the right $\Fqm$-span.
\end{lemma}

\begin{proof}
According to Lemma~\ref{lem:rootfinding_by_order_bases}, the roots contained in $\Module$ are obtained from linear combinations $[f^{(0)},\dots,f^{(\ell)}] = \MABout\v$ of the columns of $\MABout$ such that $\cdeg_\t \v \leq \khat$ and $f^{(0)}=1$.
The first condition, $\cdeg_\t \v \leq \khat$, implies that
\begin{itemize}
\item $v_i = 0$ for all $i \notin \Jset$ (since $t_i>\khat$ in this case),
\item $v_i \in \Fqm$ for all $i \in \Iset$ (since $t_i=\khat$), and
\item $\deg v_i \leq \khat-t_i$ for all $i \in \Jset \setminus \Iset$ (we write $v_i = \sum_{j=0}^{\khat-t_i} x^j \tilde{v}_{i,j}$ with $\tilde{v}_{i,j} \in \Fqm$ below).
\end{itemize}
If $\Iset=\emptyset$, we cannot have $f^{(0)}\neq0$, hence, $\Module=\emptyset$.
Else, $f^{(0)}=1$ is equivalent to $\sum_{i \in \Iset} B_{0,i} v_i = 1$.
By the elementary operations on the columns indexed by $\Iset$ (see \eqref{eq:rank_root_finding_explicit_basis_1} and \eqref{eq:rank_root_finding_explicit_basis_2}), we obtain the submatrix
\begin{align*}
\begin{bmatrix}
1 & 0 & \dots & 0 \\
\g^* & \g^{(1)} & \dots & \g^{(\iota)}
\end{bmatrix} \in \SkewPolys^{(\ell+1) \times (\iota+1)}.
\end{align*}
By combining all conditions, we have $[f^{(1)}, \dots, f^{(\ell)}]^\top \in \Module$ if and only if
\begin{align}
&\begin{bmatrix}
f^{(1)} \\
\vdots \\
f^{(\ell)}
\end{bmatrix}
= \g^* + \sum_{r=1}^{\iota} \g^{(r)} v_{i_r}' + \sum_{i=1}^{\tau} \begin{bmatrix}
B_{1,j_i} \\
\vdots \\
B_{\ell,j_i}
\end{bmatrix}
\sum_{j=0}^{\khat-t_{j_i}} \tilde{v}_{j_i,j}, \label{eq:rank_root_finding_explicit_basis_4} \\
&= \g^* + \sum_{r=1}^{\iota} \g^{(r)} v_{i_r}' + \sum_{i=1}^{\tau} \sum_{j=0}^{\khat-t_{j_i}}
\g^{\left(\iota+\sum_{i'=1}^{i-1}(\khat-t_{j_i}+1)+j+1\right)} \tilde{v}_{j_i,j}, \notag
\end{align}
with some $v_{i_r}',\tilde{v}_{j_i,j} \in \Fqm$ for all $r,i,j$. This proves \eqref{eq:rank_root_finding_explicit_basis_3}.

Since $\MABout$ is in $\s$-ordered column weak Popov form, for each root $[f^{(1)},\dots,f^{(\ell)}]^\top \in \Module$, there is a unique $\v$ with the given properties and $[f^{(1)},\dots,f^{(\ell)}]^\top = \MABout \v$. We obtain the coefficients $v_{i_r}',\tilde{v}_{j_i,j} \in \Fqm$, for $r,i,j$, of the right $\Fqm$-linear combination in \eqref{eq:rank_root_finding_explicit_basis_4} by a bijective mapping from the vector $\v$.
Hence, the linear combination in \eqref{eq:rank_root_finding_explicit_basis_4} is unique for any root and the right $\Fqm$-linearly independence of the $\g^{(i)}$ follows.
\end{proof}

Lemmas~\ref{lem:rootfinding_by_order_bases} and \ref{lem:rank_root-finding_affine_basis} imply a root-finding algorithm based on computing a right approximant basis.
We outline the procedure in Algorithm~\ref{alg:fast_root_finding} and prove its correctness and complexity in the following theorem.

\begin{algorithm}[ht]
	\caption{Fast Root-Finding Algorithm}\label{alg:fast_root_finding}
	\SetKwInOut{Input}{Input}\SetKwInOut{Output}{Output}
	\Input{Instance of Problem~\ref{prob:general_root-finding_problem}: $\ell,n \in \ZZ_{> 0}$, $\vec{k} \in \ZZ_{> 0}^\ell$, and left $\SkewPolys$-linearly independent vectors $\Q^{(1)},\dots,\Q^{(\IntParam)} \in \SkewPolys^{\ell+1} \setminus \{\0\}$ with $\deg \Q^{(i)} \leq n$ for all $i$.}
	\Output{Solution of Problem~\ref{prob:general_root-finding_problem}: if $\Module \neq \emptyset$, an affine basis $\g^*,\g^{(1)},\dots,\g^{(\delta)}$ of the right $\Fqm$-linear affine space $\Module$ as defined in \eqref{eq:affine_root_space}, i.e.,
	\begin{equation*}
	\Module = \g^* + \langle \g^{(1)}, \dots, \g^{(\delta)} \rangle_{\Fqm,\mathrm{right}}.
	\end{equation*}
	If $\Module = \emptyset$, ``no solution''}
	\BlankLine
	$\khat \gets \max_i\{k^{(i)}\}$ \\
	$\MABin \gets $ as in \eqref{eq:root-finding_MABin_matrix} \\
	$\s \gets [\khat,  \khat-k^{(1)}+1,  \dots,  \khat-k^{(\ell)}+1] $ \\
	$d \gets \max_{i,j} \left\{\deg Q_j^{(i)}\right\} + \khat$ \\
	$\MABout \gets $ right \MABnameFullStandard %
	\hfill \myAlgoComment{Algorithm~\ref{alg:rightSkewDaCAppBasis} in Section~\ref{sec:order_bases}}
	\If{$\MABout$ has a row of $\rdeg_\s \leq \khat$}{
		Compute $\g^*,\g^{(1)},\dots,\g^{(\delta)}$ as in Lemma~\ref{lem:rank_root-finding_affine_basis} \\
		\Return{$\g^*,\g^{(1)},\dots,\g^{(\delta)}$}
	}
	\Else{
		\Return{``no solution''}
	}
\end{algorithm}

\begin{theorem}\label{thm:fast_root_finding_correctness_complexity}
	Algorithm~\ref{alg:fast_root_finding} is correct. For the complexity, assume $\max_i k^{(i)} \in \Theta(n)$. Then, Algorithm~\ref{alg:fast_root_finding} has complexity
	\begin{equation*}
	\softO\!\left( \ell^\omega \OMul{n} \right)
	\end{equation*}
	operations in the base field of the cost bound $\OMul{n}$.
\end{theorem}

\begin{proof}
Correctness follows from Lemmas~\ref{lem:rootfinding_by_order_bases} and \ref{lem:rank_root-finding_affine_basis}.
Complexity-wise the heaviest step is the computation of the right approximant basis, which costs $\softO\!\left( \ell^\omega \OMul{n + \max_{i}k^{(i)}} \right) \subseteq \softO\!\left( \ell^\omega \OMul{n}\right)$ operations in the base field of the cost bound $\OMul{n}$, since by assumption on the degree of $\Q^{(i)}$ in Problem~\ref{prob:general_root-finding_problem}, we have $d \leq n + \max_{i}k^{(i)} \in \Theta(n)$.
Computing the affine basis as in Lemma~\ref{lem:rank_root-finding_affine_basis} costs $O(\ell^2 \max_{i}k^{(i)}) \subseteq O(\ell^2 n)$ operations in $\Fqm$. %
\end{proof}

\section{Fast Decoding of Sum-Rank-Metric Codes}\label{sec:sum-rank}

In this section, we show how to speed up decoding of linearized Reed--Solomon codes in the sum-rank metric.
This is achieved by proposing new, faster, algorithms for the two core computational problems of the Martínez-Peñas--Kschischang decoder \cite{martinez2019reliable}, which in fact decodes a more general class of codes in a more general metric: skew Reed--Solomon codes in the skew metric.
We first state these problems and remind how the decoder works in Section~\ref{ssec:sum-rank_computational_problems}.
We then present our new algorithms for them in Sections~\ref{ssec:fast_remainder_ev_arith} and \ref{ssec:bivariate_remainder_interpolation}.

In this section, we only use the remainder evaluation (cf.~Section~\ref{ssec:evaluation_maps}) of skew polynomials.

\subsection{Computational Problems and their Relation to Decoding}\label{ssec:sum-rank_computational_problems}

To state the two computational problems, we need to first recall some notions related to the remainder evaluation of skew polynomials.
\subsubsection{Preliminaries on Remainder Evaluation}

The following notions were introduced in \cite{lam1985general,lam1988vandermonde,lam1988algebraic}, and we use the notation of \cite{martinez2019reliable}.
Let $A \subseteq \SkewPolys$, $\Omega \subseteq \Fqm$, and $a \in \Fqm$.
The \emph{zero set} of $A$ is defined by $Z(A) := \left\{ \alpha \in \Fqm \, : \, \remev{f}{\alpha} = 0 \, \forall \, f \in A \right\}$, and $I(\Omega) := \left\{ f \in \SkewPolys \, : \, \remev{f}{\alpha} = 0 \, \forall \, \alpha \in \Omega \right\}$ denotes the \emph{associated ideal} of $\Omega$.
The \emph{P-closure} (or polynomial closure) of $\Omega$ is defined by $\bar{\Omega} := Z(I(\Omega))$, and $\Omega$ is called \emph{P-closed} if $\bar{\Omega}=\Omega$.
A P-closure is always P-closed.
The elements of $\Fqm \setminus \bar{\Omega}$ are all said to be \emph{P-independent from $\Omega$}.

A set $\Bset \subseteq \Fqm$ is said to be \emph{P-independent} if any $b \in \Bset$ is P-independent from $\Bset \setminus \{b\}$.
If $\Bset$ is P-independent and $\Omega := \bar{\Bset} \subseteq \Fqm$, we say that $\Bset$ is a \emph{P-basis of $\Omega$}.
$\Omega$ may have many P-bases but they all have the same number of elements, called the \emph{P-rank} of $\Omega$, denoted $\Prk(\Omega) = |\Bset|$.

For any $\Bset \subset \Fqm$ then
$I(\Bset)$ is a left $\SkewPolys$-ideal and hence principal, so there is a unique monic skew polynomial $\MSPrem_\Bset$ of smallest degree that generates it.
We call $\MSPrem_\Bset$ the \emph{remainder annihilator polynomial} of $\Bset$ and we have $\deg \MSPrem_\Bset = \Prk(\bar\Bset)$.
In particular, $\deg \MSPrem_\Bset = |\Bset|$ if and only if $\Bset$ is P-independent.

Let $\Bset = \{ \beta_1,\dots,\beta_n \} \in \Fqm$ be P-independent\footnote{%
  Here and in the sequel, we slightly abuse notation and take this to mean $\Bset$ is an ordered set and that the $\beta_i$ are distinct.%
}.
For any $\vec r = (r_1,\dots,r_n) \in \Fqm$, there is a unique skew polynomial $\IPrem{\Bset, \vec r} \in \SkewPolys$ of degree less than $n$ such that
\begin{align*}
\remev{\IPrem{\Bset, \vec r}}{\beta_i} = r_i \quad \forall \, i=1,\dots,n.
\end{align*}
We call this the \emph{remainder interpolation polynomial} of $\vec r$ on $\Bset$.

\subsubsection{Computational Problems}

The decoder in \cite{martinez2019reliable} is based on the following computational problems.

\begin{problem}[Fast Remainder-Evaluation Operations]\label{prob:remainder_arithmetic}
  Let $\Bset = \{\beta_1,\dots,\beta_n\} \subseteq \Fqm$ be $P$-independent.
\begin{enumerate}[label={\roman*)}]
\item
  \label{itm:rem_annihilator}
  Compute $\MSPrem_{\Bset}$ (\emph{remainder annihilator polynomial}).
\item
  \label{itm:rem_MPE}
  Given $f \in \SkewPolys$ with $\deg f \leq n$, compute $\big[f[\beta_1],\dots,f[\beta_n]\big]$ (\emph{multi-point remainder evaluation}).
\item
  \label{itm:rem_IP}
  Given $\vec r \in \Fqm$, compute $\IPrem{\Bset, \vec r}$ (\emph{remainder interpolation}).
\end{enumerate}
\end{problem}

\begin{problem}[$2$D Vector Remainder Interpolation]\label{prob:bivariate_remainder_interpolation}
  Let $\Bset = \{\beta_1,\dots,\beta_n\} \subseteq \Fqm$ be $P$-independent.
  Given $D \in \ZZ_{> 0}$, $\w \in \ZZ_{\geq 0}^2$, and $\vec r \in \Fqm$, compute a non-zero $[Q_0,Q_1] \in \SkewPolys^2$ such that
\begin{align}
Q_0[\beta_i] + (Q_1 R)[b_i] &= 0 \quad \forall \, i=1,\dots,n, \label{eq:bivariate_remaider_interpolation_eval} \\
\rdeg_{\w} \begin{bmatrix}
Q_0 & Q_1
\end{bmatrix} &< D, \label{eq:bivariate_remaider_interpolation_deg}
\end{align}
where $R := \IPrem{\Bset, \vec r}$.
\end{problem}

As for Problem~\ref{prob:general_interpolation_problem} in Section~\ref{sec:rank_and_subspace}, we assume $D \in \Theta(n)$ for the complexity analysis.
This is the only case relevant for the decoding problem studied in the following. See Section~\ref{ssec:remarks_on_generality} in the conclusion for a discussion on the general case.
The previously fastest algorithms to solve Problems~\ref{prob:remainder_arithmetic} and \ref{prob:bivariate_remainder_interpolation} with $D \in \Theta(n)$ were presented in \cite{martinez2019reliable} both of which use $O(n^2)$ operations in $\Fqm$.

\subsubsection{Decoding of Skew Reed--Solomon Codes}

Let $\Bset = \{\beta_1,\dots,\beta_n\}$ be P-independent.
The \emph{skew Reed--Solomon code (w.r.t.\ $\Bset$)} \cite{boucher2014linear} of dimension $k<n$ is defined as
\begin{align*}
\mathcal{C}_{\mathrm{skew},\Bset} := \left\{ \left[\remev{f}{\beta_1}, \dots, \remev{f}{\beta_n}\right] \, : \, f \in \SkewPolys_{<k}  \right\}.
\end{align*}
The codes are designed for the skew metric, which is defined as follows.
The \emph{skew weight (w.r.t.\ $\Bset$)} \cite{martinez2018skew} is\footnote{An equivalent definition of the skew weight based on the left least common multiple (llcm) is given in~\cite{boucher2020algorithm} as $\wtB(\y)=\deg\llcm_{y_i\neq 0}\left(x-\frac{\sigma(y_i)\beta_i}{y_i}\right)$.}
\begin{align*}
\wtB \, : \, \Fqm^n &\to \ZZ_{\geq 0} \\
\y = [y_1,\dots, y_n] &\mapsto n- \Prk\left(Z\!\left(\IPrem{\Bset, \vec y}\right) \cap \, \bar{\Bset} \right).
\end{align*}
The \emph{skew distance (w.r.t.\ $\Bset$)} is defined by $\dB(\y_1,\y_2) := \wtB(\y_1-\y_2)$ for any $\y_1,\y_2 \in \Fqm^n$.
A skew Reed--Solomon code has minimum distance $d=n-k+1$ w.r.t.\ the skew metric.

The skew metric is related to the sum-rank metric (see Theorem~\ref{thm:sum-rank_relation_skew_and_sum-rank_metric} below), which is defined as follows.
As in Section~\ref{sec:rank_and_subspace}, we define the ($\Fq$-)rank weight of a row vector in $\Fqm^{1 \times n'}$ as the $\Fq$-rank of the $m \times n'$ matrix over $\Fq$ obtained by column-wise expanding each entry of the vector in a basis of $\Fqm$.
For $\vec n = [n_1,\ldots,n_\ell]$ with $n_i \in \ZZ_{>0}$ and $\sum_{i=1}^{\ell} n_i = n$, the \emph{sum-rank weight} (w.r.t.~$\vec n$) on $\Fqm^n$ \cite{nobrega2010multishot} is defined as
\begin{align*}
\wtSR \, : \, \Fqm^n &\to \ZZ_{\geq 0}, \\
\c = \big[\c^{(1)} \mid \dots \mid \c^{(\ell)}\big] &\mapsto \sum_{i=1}^{\ell} \wtR\!\big(\c^{(i)}\big),
\end{align*} 
where we divide $\c$ into subblocks $\c^{(i)} \in \Fqm^{n_i}$.
The \emph{sum-rank distance} of $\a,\b \in \Fqm^n$ is $\dSR(\a,\b) := \wtSR(\a-\b)$.

\begin{theorem}[\cite{lam1985general,lam1988vandermonde,martinez2018skew,martinez2019reliable}]\label{thm:sum-rank_relation_skew_and_sum-rank_metric}
Let $\vec n = [n_1,\ldots,n_\ell]$ with $n_i \in \ZZ_{>0}$ and $\sum_{i=1}^{\ell} n_i = n$, and let $m \in \ZZ_{>0}$ with $m \geq \max_i\{n_i\}$ and $\ell < q$ with $q$ a prime power.
Then there is a P-independent set $\Bset = \{\beta_1,\dots,\beta_n\} \subset \Fqm$ and non-zero field elements $\v = [v_1,\dots,v_n] \in (\Fqm^*)^n$ such that
\begin{align*}
\varphi_{\Bset,\v} \, : \, \left( \Fqm^n, \dB \right) &\to  \left( \Fqm^n, \dSR \right), \\
\c = [c_1,\dots,c_n] &\mapsto [c_1 v_1,\dots,c_n v_n]
\end{align*}
is an isometry (i.e., bijective, distance-preserving mapping).
\end{theorem}

For a pair $\Bset$ and $\v$ as in Theorem~\ref{thm:sum-rank_relation_skew_and_sum-rank_metric}, the linear code $\varphi_{\Bset,\v}\!\left(\mathcal{C}_{\mathrm{skew},\Bset}\right)$ is a \emph{linearized Reed--Solomon code} as introduced in \cite{martinez2018skew}.
Since $\varphi_{\Bset,\v}$ is an isometry, such a code has minimum sum-rank distance $n-k+1$ and is thus maximum distance separable in the sum-rank metric.
Having precomputed $\vec v$, the isometry can be applied or reversed in only $n$ multiplications in $\Fqm$.
Hence, any efficient decoder for skew Reed--Solomon codes in the skew metric is also an efficient decoder for linearized Reed--Solomon codes in the sum-rank metric.
As skew Reed--Solomon codes are more general and can be described in skew polynomial language, we will only treat these codes in the following.

Let $\Bset = \{ \beta_1,\dots,\beta_n \} \in \Fqm$ be P-independent.
Let
\begin{align*}
\r = (r_1,\ldots,r_n) = \c + \e \in \Fqm^n
\end{align*}
such that $\c$ is a codeword of the skew Reed--Solomon code $\mathcal{C}_{\mathrm{skew},\Bset}[n,k]$ and $\e$ is an error of skew weight $\wtB(\e)$.
The Martínez-Peñas--Kschischang decoder \cite{martinez2019reliable} finds a solution $[Q_0,Q_1] \in \SkewPolys^2$ of Problem~\ref{prob:bivariate_remainder_interpolation} with input $D=\lfloor\tfrac{n-k}{2}\rfloor+k-1$, $\w = [0,k-1]$, and $\{ (\beta_i,r_i) \}_{i=1}^n$.
It was shown in \cite[Proposition~4]{martinez2019reliable} that if the skew weight of the error $\e$ is at most $\wtB(\e) \leq \lfloor \tfrac{n-k}{2}\rfloor$, then any such solution satisfies
$-Q_0 = Q_1 f$, where $f \in \SkewPolys_{<k}$ is the unique skew polynomial (i.e., message polynomial) of degree less than $k$ with $\c = \left[\remev{f}{\beta_1}, \dots, \remev{f}{\beta_n}\right]$.
Hence, to finish decoding once $[Q_0, Q_1]$ is obtained, we simply need to divide $-Q_0$ by $Q_1$ from the left and (multi-point) evaluate the resulting polynomial to obtain the original codeword $\c$.

\begin{theorem}\label{thm:sum-rank_summary}
Decoding a skew Reed--Solomon code $\mathcal{C}_{\mathrm{skew},\Bset}$ using the decoder in \cite{martinez2019reliable} has complexity $\softO(\OMul{n})$ 
	operations in the base field of the cost bound $\OMul{n}$, if
\begin{itemize}
  \item the 2D vector remainder interpolation is implemented using Algorithm~\ref{alg:bivariate_interpolation} in Section~\ref{ssec:bivariate_remainder_interpolation} with input $D=\lfloor\tfrac{n-k}{2}\rfloor+k$, $\w = [0,k-1]$, and $\{(\beta_i,r_i)\}_{i=0}^n$;
\item the univariate remainder interpolation and remainder annihilator computation inside Algorithm~\ref{alg:bivariate_interpolation} are implemented using the algorithms implied by Theorems~\ref{thm:fast_remainder_MSP} and \ref{thm:fast_remainder_interpolation} in Section~\ref{ssec:fast_remainder_ev_arith};
\item and, if the output should be the transmitted codeword instead of the message polynomial, the re-encoding is implemented using the fast multi-point evaluation algorithm implied by Theorem~\ref{thm:fast_remainder_MPE} in Section~\ref{ssec:fast_remainder_ev_arith}.
\end{itemize}
Decoding a linearized Reed--Solomon code can be done in the same cost through the isometry $\varphi_{\Bset,\v}$.
\end{theorem}

\begin{proof}
The statement follows from \cite[Proposition~4]{martinez2019reliable} (see summary above) and Theorems~\ref{thm:fast_remainder_MSP}, \ref{thm:fast_remainder_MPE}, \ref{thm:fast_remainder_interpolation}, and \ref{thm:bivariate_remainder_interpolation_summary} (see next subsections).
\end{proof}

\subsection{New Algorithms for Operations with Remainder Evaluation}
\label{ssec:fast_remainder_ev_arith}

We present fast algorithms to solve Problem~\ref{prob:remainder_arithmetic}: computing annihilators, multi-point evaluation, and remainder interpolation.
The methods are similar to corresponding algorithms for the operator evaluation in \cite[Lemma 3.3]{caruso2017fast} (annihilator) and \cite[Sections~3.4 and 3.5]{puchinger2017fast} (multi-point evaluation and interpolation), which are in turn non-commutative adaptations of well-known algorithms over ordinary polynomial rings (see, e.g., \cite{von2013modern}).

\begin{theorem}[Fast remainder annihilator polynomial computation]\label{thm:fast_remainder_MSP}
Let $\Bset = \{\beta_1,\dots,\beta_n\}$ be P-independent.
Then $\MSPrem_{\Bset}$ can be computed in $\softO(\OMul{n})$ operations in the base field of the cost bound $\OMul{n}$.
\end{theorem}

\begin{proof}
Recall that the llcm of two skew polynomials $f,g \in \SkewPolys$ is the unique monic skew polynomial $\llcm(f,g) \in \SkewPolys \setminus \{0\}$ of smallest degree such that there are polynomials $\chi_1,\chi_2 \in \SkewPolys$ with $\chi_1 f = \chi_2 g = \llcm(f,g)$.
Note that we have $\deg \llcm(f,g) \leq \deg f + \deg g$.

Observe that if $\Bset_1, \Bset_2 \subset \Fqm$ are disjoint, then $\llcm\left( \MSPrem_{\Bset_1} , \, \MSPrem_{\Bset_2}  \right)$ is the least-degree monic polynomial in both the left ideal spanned by $\MSPrem_{\Bset_1}$ and by $\MSPrem_{\Bset_2}$, which must therefore be $\MSPrem_{\Bset_1 \cup \Bset_2}$.
Furthermore, it is easy to see that $\MSPrem_{\{\beta\}} = x-\beta$ for any $\beta \in \Fqm$.
Recursively subdividing the initial $\Bset$ in disjoint subsets and structuring this this as a divide-\&-conquer computation, the complexity $C(n)$ of computing $\MSPrem_{\Bset}$ as a function of $n$ obeys $C(1) = O(1)$ and the recursion $C(n) = L(n) + 2C(\lceil n/2 \rceil), n > 1$, where $L(n)$ denotes the cost of computing the $\llcm$ of two skew polynomials of degree at most $n$.
By \cite[Theorem~3.2.7]{caruso2017new} $L(n) \subseteq \softO(\OMul{n})$ operations in the base field of the cost bound $\OMul{n}$, so by the master theorem, $C(n)$ is in the claimed complexity.
\end{proof}

\begin{theorem}[Fast multi-point evaluation]\label{thm:fast_remainder_MPE}
Let $\Bset = \{\beta_1,\dots,\beta_n\}$ be P-independent and $f \in \SkewPolys$ with $\deg f \leq n$.
Then, $[\remev{f}{\beta_1},\dots,\remev{f}{\beta_n}]$ can be computed in $\softO(\OMul{n})$ operations in the base field of the cost bound $\OMul{n}$.
\end{theorem}

\begin{proof}
  Let $\Bset = \Bset_1 \sqcup \Bset_2$ be a partition of $\Bset$, and define
\begin{align*}
f_1 &:= f \remr \MSPrem_{\Bset_1}, \\
f_2 &:= f \remr \MSPrem_{\Bset_2}.
\end{align*}
Then for any $\beta \in \Bset$:
\begin{align*}
\remev{f}{\beta} =
\begin{cases}
\remev{f_1}{\beta}, &\text{if } \beta \in \Bset_1 \\
\remev{f_2}{\beta}, &\text{if }  \beta \in \Bset_2 \ .
\end{cases}
\end{align*}
Indeed for $j=1,2$, the polynomial $f-f_j$ is right-divisible by $\MSPrem_{\Bset_{j}}$ and hence $\remev{(f-f_j)}{\beta}= 0$ for $\beta \in \Bset_j$.

Thus, if we split $\Bset$ in two parts of size $\leq n' := \lceil n/2 \rceil$, we can evaluate at each $\beta \in \Bset$ by computing two remainder annihilator polynomials of degree $n'$, two right divisions of degree $n$, followed by two  recursive multi-point evaluations of polynomials of degree at most $n'$ in as many points.
In the base case, we evaluate a polynomial of degree $\leq 1$ at one point, which costs $O(1)$.
By Theorem~\ref{thm:fast_remainder_MSP} and \cite[Section~3.2.1]{caruso2017new} both the annihilator computations and divisions can be performed in $\softO(\OMul{n})$ operations in the base field of the cost bound $\OMul{n}$, and we obtain the claimed complexity using the master theorem.
\end{proof}

\begin{theorem}\label{thm:fast_remainder_interpolation}
Let $\Bset = \{\beta_1,\dots,\beta_n\}$ be P-independent and $\r \in \Fqm^n$.
Then the interpolation polynomial $\IPrem{\Bset, \vec r} \in \SkewPolys_{<n}$ can be computed in $\softO(\OMul{n})$ operations in the base field of the cost bound $\OMul{n}$.
\end{theorem}

\begin{proof}
  Let $n' = \lceil n/2\rceil$, and $I = \{ 1, \ldots, n' \}$, and $J = \{ n'+1,\ldots,n\}$ and set
  $\Bset_1 := \{\beta_i\}_{i \in I}$ and $\Bset_2 :=\{\beta_i\}_{i \in J}$.
  We claim the identity:
\begin{align*}
\IPrem{\Bset, \r} &= \IPrem{\tilde\Bset_1, \tilde\r_1} \MSPrem_{\Bset_2} + \IPrem{\tilde\Bset_2, \tilde\r_2} \MSPrem_{\Bset_1},
\end{align*}
where
\begin{align*}
  \tilde\Bset_1 &= \left\{ \tfrac{\sigma\!\left(\remev{\MSPrem_{\Bset_2}}{\beta}\right) \beta}
                            {\remev{\MSPrem_{\Bset_2}}{\beta}}
                      \mid \beta \in \Bset_1 \right\} \\
  \tilde\Bset_2 &= \left\{ \tfrac{\sigma\!\left(\remev{\MSPrem_{\Bset_1}}{\beta}\right) \beta}
                            {\remev{\MSPrem_{\Bset_1}}{\beta}}
                      \mid \beta \in \Bset_2 \right\} \\
  \tilde\r_1 &= \big( \tfrac{r_i}{\remev{\MSPrem_{\Bset_2}}{\beta_i}} \big)_{i \in I} \\
  \tilde\r_2 &= \big( \tfrac{r_i}{\remev{\MSPrem_{\Bset_1}}{\beta_i}} \big)_{i \in J} \ .
\end{align*}
Indeed: the right-hand side clearly has degree less than $n$ and remainder-evaluates to $r_i$ at $\beta_i$ for each $i \in \{1,\ldots,n\}$.
Note that the P-independence of $\Bset$ implies $\remev{\MSPrem_{\Bset_j}}{\beta_i} \neq 0$, so the $\tilde{\beta}_i$ and $\tilde{r}_i$ are well-defined.
Furthermore, $\tilde\Bset_1$ is P-independent by the following argument.
It follows from the product rule of remainder evaluation that the monic polynomial
\begin{equation*}
\MSPrem_{\tilde\Bset_1} \cdot \MSPrem_{\Bset_2}
\end{equation*}
vanishes on $\Bset$.
Hence, it must be must be right-divisible by $\MSPrem_\Bset$, which has degree $n$ by the P-independence of $\Bset$.
This implies $\deg \MSPrem_{\tilde\Bset_1} \geq |\tilde\Bset_1|$ which implies the P-independence of $\tilde\Bset_1$.
Mutadis mutandis, $\tilde\Bset_2$ is also P-independent, and the interpolation polynomials $\IPrem{\tilde\Bset_1,\tilde\r_1}$ and $\IPrem{\tilde\Bset_2,\tilde\r_2}$ are therefore well-defined.

Hence, we may compute $\IPrem{\Bset, \r}$ by computing two remainder annihilator polynomials of size $n'$, two multi-point evaluations of polynomials of degree at most $n'$ on $n'$ points, and recursively two interpolations on $n'$ points.
For the base case, we have $\IPrem{\beta, r} = (x-\beta) + r$ for any $\beta \in \Fqm^*$ and $r \in \Fqm$.
By Theorems~\ref{thm:fast_remainder_MSP} and \ref{thm:fast_remainder_MPE} and the master theorem, we obtain the desired complexity.
\end{proof}

\subsection{A New Algorithm for the 2D Vector Interpolation Problem}
\label{ssec:bivariate_remainder_interpolation}

The following statements reduce Problem~\ref{prob:bivariate_remainder_interpolation} (2D vector remainder interpolation) to computing a left approximant basis.
This will lead to a faster algorithm to solve the problem.

\begin{lemma}\label{lem:skew_metric_equvialent_problems}
Consider an instance of Problem~\ref{prob:bivariate_remainder_interpolation} and let $R := \IPrem{\Bset,\vec r}$ and $G := \MSPrem_\Bset$. Then, Condition \eqref{eq:bivariate_remaider_interpolation_eval} in Problem~\ref{prob:bivariate_remainder_interpolation} 
is equivalent to
\begin{align}
Q_0+Q_1 R \equiv 0 \quad \modr G. \label{eq:skew_RS_pade_approx_new_congruence}
\end{align}
\end{lemma}

\begin{proof}
First note that $Q_0+Q_1 R \equiv 0 \modr G$ if and only if
\begin{equation*}
\exists \, \chi \in \SkewPolys \, : \, Q_0+Q_1 R = \chi G.
\end{equation*}
Due to $G[b_i] = 0$, we have for all $i=1,\dots,n$
\begin{equation*}
(Q_0 + Q_1 R)[b_i] = (\chi G)[b_i] = Q[b_i] +\underbrace{(\chi G)[b_i]}_{= \, 0} = 0,
\end{equation*}
so \eqref{eq:skew_RS_pade_approx_new_congruence} implies \eqref{eq:bivariate_remaider_interpolation_eval}.
For the other direction, we note that due to $(Q_0 + Q_1 R)]b_i] = 0$ for all $i$, we have $Q_0 + Q_1 R \in I(\Bset)$. Since $G$ generates the left ideal $I(\Bset)$, there must be a polynomial $\chi \in \SkewPolys$ with $Q_0 + Q_1 R = \chi G$.
\end{proof}

\begin{lemma}\label{lem:skew_RS_order_basis}
  Consider an instance of Problem~\ref{prob:bivariate_remainder_interpolation} and let $R := \IPrem{\Bset, \r} \in \SkewPolys$ and $G := \MSPrem_{\Bset}$.
Let $\s = [s_1,s_2,s_3] := [w_1,w_2,\min\{w_1,w_2\}]$, and $d = D + n - \min\{w_1,w_2\}$, as well as
\begin{align*}
\MABin = \begin{bmatrix}
1 \\
R \\
G
\end{bmatrix}.
\end{align*}
Let $\MABout$ be a left \MABnameFullStandard. %
Then Problem~\ref{prob:bivariate_remainder_interpolation} has a solution if and only if $\MABout$ contains at least one row of $\s$-shifted degree at most $D-1$.
Furthermore, for any such row $\v = [v_1,v_2,v_3]$, then
$[Q_0,Q_1] := [v_1,v_2]$
is a solution of Problem~\ref{prob:bivariate_remainder_interpolation}.
\end{lemma}

\begin{proof}
Due to Lemma~\ref{lem:skew_metric_equvialent_problems}, Condition~\eqref{eq:bivariate_remaider_interpolation_eval} in Problem~\ref{prob:bivariate_remainder_interpolation} is equivalent to \eqref{eq:skew_RS_pade_approx_new_congruence}. It is easy to see that some $Q_0,Q_1 \in \SkewPolys$ fulfill \eqref{eq:bivariate_remaider_interpolation_eval} if and only if there is a polynomial $\chi \in \SkewPolys$ with
\begin{align*}
Q_0 + Q_1 R + \chi G &= 0 \\
\Leftrightarrow \quad [Q_0,Q_1,\chi] \cdot \MABin &= 0.
\end{align*}
Hence, the $Q_0,Q_1$ fulfilling \eqref{eq:bivariate_remaider_interpolation_eval} correspond directly to the vectors $[Q_0,Q_1,\chi]$ in the left kernel of the matrix $\MABin$.
Furthermore, consider the shifted degree of such a $Q_0,Q_1$ which also satisfies the degree constraints of Problem~\ref{prob:bivariate_remainder_interpolation}:
\begin{align*}
\deg Q_0 + s_1 &< D, \\
\deg Q_1 + s_2 &< D, \\
\deg \chi + s_3 &= \deg (Q_0 + Q_1 R)+\min\{w_1,w_2\}-\deg G \\
&\leq \max\{ \deg Q_0,\deg Q_1 + n-1\}\\
&\quad -n+\min\{w_1,w_2\} <D \ .
\end{align*}
In other words, $\rdeg_\s [Q_0,Q_1,\chi] < D$.
Any vector $\v = [v_1,v_2,v_3] \in \SkewPolys^3$ with $\rdeg_\s \v <D$ fulfills
\begin{align*}
\deg\!\left( \v \cdot \MABin \right) < D + n-\min\{w_1,w_2\},
\end{align*}
so by the choice of $d$, any vector of this shifted degree is a left approximant of $\MABin$ of order $d$ if and only if it is in the left kernel of $\MABin$.

Hence, the solutions of Problem~\ref{prob:bivariate_remainder_interpolation} are exactly the first two entries of all non-zero left approximants of $\MABin$ of order $d$ with $\s$-shifted degree at most $D-1$. Since the rows of $\MABout$ are left approximants, any row of sufficiently small shifted degree is a solution of the problem.
Moreover, the problem has a solution if and only if the row space of $\MABout$ contains a row of sufficiently small $\s$-shifted degree. Since $\MABout$ is in $\s$-shifted weak Popov form, one of its rows has minimal $\s$-shifted degree among all vectors of the row space, i.e., at most $D-1$ if and only if the problem has a solution.
\end{proof}

Lemma~\ref{lem:skew_RS_order_basis} implies an algorithm to solve Problem~\ref{prob:bivariate_remainder_interpolation}, which we outline in Algorithm~\ref{alg:bivariate_interpolation}. We summarize its complexity in Theorem~\ref{thm:bivariate_remainder_interpolation_summary} below.

\begin{algorithm}[ht!]
\caption{Fast 2D Vector Remainder Interpolation}\label{alg:bivariate_interpolation}
\SetKwInOut{Input}{Input}\SetKwInOut{Output}{Output}
\Input{Instance of Problem~\ref{prob:bivariate_remainder_interpolation}: $\Bset = \{\beta_1,\ldots,\beta_n\} \subset \Fqm$ and P-independent, $D \in \ZZ_{> 0}$, $\w = [w_1,w_2] \in \ZZ_{\geq 0}^2$, and $\vec r \in \Fqm^n$.}

\Output{Solution $[Q_0,Q_1] \in \SkewPolys^2 \setminus \{\0\}$ if the problem has a solution, ``no solution'' otherwise.} 

\BlankLine

$G \gets \MSPrem_{\Bset}$
$R \gets \IPrem{\{(\beta_i,r_i)\}_{i=1}^{n}}$
$\s \gets [w_1,w_2,\min\{w_1,w_2\}]$ \\
$d \gets D + n - \min\{w_1,w_2\}$ \\
$\MABin \gets \begin{bmatrix}
1 \\
R \\
G
\end{bmatrix}$ \\
$\MABout \gets$ left \MABnameFullStandard \hfill \myAlgoComment{Algorithm~\ref{alg:leftSkewDaCAppBasis} in Section~\ref{sec:order_bases}}%
\If{$\MABout$ has a row $\v = [Q_0,Q_1,\chi]$ of $\rdeg_\s \v <D$}{
\Return{$[Q_0,Q_1]$}
} \Else {
\Return{``no solution''}
}
\end{algorithm}

\begin{theorem}\label{thm:bivariate_remainder_interpolation_summary}
Algorithm~\ref{alg:bivariate_interpolation} is correct. Assuming $D \in \Theta(n)$, it has complexity
\begin{align*}
\softO(\OMul{n})
\end{align*}
operations in the base field of the cost bound $\OMul{n}$.
\end{theorem}

\begin{proof}
Correctness follows directly from Lemma~\ref{lem:skew_RS_order_basis}.

Setting up the matrix $\MABin$ consists of computing a remainder annihilator polynomial of degree $n$ and an interpolation polynomial of degree $<n$.
Both operations can be done in $\softO(\OMul{n})$ operations in the base field of the cost bound $\OMul{n}$ using Theorem~\ref{thm:fast_remainder_MSP} and \ref{thm:fast_remainder_interpolation}, respectively.
The approximant basis can be computed in $\softO(\OMul{\max\{D,n\}})\subseteq \softO(\OMul{n})$ operations in the base field of the cost bound $\OMul{n}$ using Algorithm~\ref{alg:rightSkewDaCAppBasis} in Section~\ref{sec:order_bases}.
\end{proof}

\section{Conclusion}\label{sec:conclusion}

\subsection{Summary}

We have presented new algorithms for the underlying computational problems of three different decoders: interpolation-based decoding of interleaved Gabidulin codes in the rank metric, interpolation-based decoding of lifted interleaved Gabidulin codes in the subspace metric, and decoding of linearized/skew Reed--Solomon codes in the sum-rank/skew metric.
Most of these computational problems were shown to be reducible to computing a left or right approximant basis over skew polynomial rings.

For all considered computational problems, hence also all considered decoders, we obtain an improvement in the dependence of the main parameter of a problem, say $n$, of the (soft-$O$) asymptotic complexity bound from a quadratic (or larger) dependence $n^2$ over $\Fqm$ to the cost $\OMul{n}$ of multiplying two skew polynomials of degree at most $n$. Since the latter, expressed in operations in $\Fqm$, is sub-quadratic in $n$ (at least $\OMul{n} \in O(n^{1.69})$, cf.~Section~\ref{ssec:cost_skew_poly_operations}), we obtain significant speed-ups for all algorithms.
See Tables~\ref{tab:overview_decoders} and \ref{tab:overview_tools} in the introduction for a detailed summary.

On the level of decoders, in the subspace- and sum-rank-metric cases we obtain faster decoding algorithms than previously known, while in the rank-metric case, we match the fastest state-of-the-art \cite{sidorenko2014fast} for decoding interleaved Gabidulin codes with a different decoding method.

\subsection{Further Applications}\label{ssec:further_applications}

Some of the studied computational problems (cf.~Table~\ref{tab:overview_tools} in the introduction) have further applications beyond the scope of this paper, which we briefly summarize in the following. Since we have obtained faster algorithms to solve these problems, this might also influence these applications.

The vector (operator) interpolation (Problem~\ref{prob:general_interpolation_problem}) also corresponds to the interpolation steps in the decoding algorithms for Mahdavifar--Vardy \cite{mahdavifar2018algebraic}, folded Gabidulin~\cite{bartz2017algebraic}, and virtual interleaved Gabidulin~\cite{guruswami_list_2013} codes.
Hence, Algorithm~\ref{alg:fast_interpolation} immediate speeds up the interpolation steps of these decoders.
Note that root finding in these algorithms is not an instance of the vector root-finding problem (Problem~\ref{prob:general_root-finding_problem}), hence further work is necessary to improve the overall complexity of these decoding algorithms.

Encoding in a linearized or skew Reed--Solomon code corresponds to a multi-point evaluation of a message polynomial at the evaluation points. Hence, Theorem~\ref{thm:fast_remainder_MPE} implies a faster encoder.

The maximally recoverable locally repairable (also called partial MDS) codes in \cite{martinez2019universal} are defined via linearized Reed--Solomon codes.
Repairing globally with these codes corresponds to erasure decoding of these codes and can be implemented by a skew polynomial remainder interpolation (part of Problem~\ref{prob:remainder_arithmetic}). Hence, the algorithm implied by Theorem~\ref{thm:fast_remainder_interpolation} immediately speeds up the repair process of these codes.

\subsection{Remarks on Generality}\label{ssec:remarks_on_generality}

All definitions and statements in Section~\ref{sec:order_bases} (approximant bases), except for complexities, remain true when stated for skew polynomials over arbitrary finite Galois extensions $\mathbb{L}/\mathbb{K}$ instead of $\Fqm/\Fq$ and automorphisms $\sigma \in \Gal(\mathbb{L}/\mathbb{K})$ with $\mathbb{K} = \mathbb{L}^\sigma$. The complexities are as stated if we in addition assume that there is a working basis of $\mathbb{L}/\mathbb{K}$ which allows to multiply, add, and apply $\sigma$ to elements of $\mathbb{L}$ in $\softO([\mathbb{L}:\mathbb{\mathbb{K}}])$ operations in $\mathbb{K}$ (this is the same assumption as in \cite{caruso2017fast}).

The output of Algorithm~\ref{alg:fast_interpolation} has slightly more structure than required by Problem~\ref{prob:general_interpolation_problem} (vector operator interpolation problem in Section~\ref{sec:rank_and_subspace}): the found $\SkewPolys$-linearly independent vectors $\Q^{(1)},\dots,\Q^{(\ell')}$ are reduced, i.e.~the vector of $\w$-degrees is lexicographically minimal over all possible bases of $\Qspace$.

In Section~\ref{sec:rank_and_subspace}, we assumed for the complexity analysis that the input parameters $D$ and $n$ of the vector interpolation problem (Problem~\ref{prob:general_interpolation_problem}) satisfy $D \in \Theta(n)$ since this is the only case relevant for the decoding problems considered here.
It can be seen by adapting the proof of Theorem~\ref{thm:fast_interpolation_correctness_complexity} that for general $D$ and $n$, Algorithm~\ref{alg:fast_interpolation} has complexity $\softO(\ell^\omega \OMul{D+n})$ operations in the base field of the cost bound $\OMul{n}$.
Hence, for $D \ll n$ and $D \gg n$, the algorithm---as stated---is not faster than the one in \cite{xie_linearized_2013}, which has complexity $O(\ell^2 D n)$ over $\Fqm$ in general.
The details are out of the scope of this paper, but we briefly outline observations that we believe could lead to an improved cost of Algorithm~\ref{alg:fast_interpolation} for these parameter ranges:
If $n \ll D$, then the left kernel of $\MABin$ contains a basis of $\ell+1$ elements, whose degree can be bounded only in $n$ and $\w$.
Hence, it appears possible to choose the order $d$ of the sought approximant basis much smaller than $D+n$.
The case $n \gg D$ may be improved by separating the interpolation constraints into $\approx n/D$ groups of $D$ constraints each, and then chaining the minimal approximant basis computations while sifting out high-degree rows.

Analogously, we can improve the cost of solving Problem~\ref{prob:bivariate_remainder_interpolation} (2D vector remainder interpolation in Section~\ref{sec:sum-rank}) for $D \notin \Theta(n)$ by the same methods.

In Problem~\ref{prob:general_root-finding_problem} (vector root-finding problem in Section~\ref{sec:rank_and_subspace}), we assumed that $\max_{i}k^{(i)} \in \Theta(n)$.
In general, Algorithm~\ref{alg:fast_root_finding} has complexity $\softO\!\left( \ell^\omega \OMul{n + \max_{i}k^{(i)}} \right)$ operations in the base field of the cost bound $\OMul{n}$. For $n \gg \max_{i}k^{(i)}$, this may be slower than the algorithms in \cite{wachter2014list,BartzWachterZeh_ISubAMC}.
Again we believe Algorithm~\ref{alg:fast_root_finding} could enjoy modifications similar to those outlined above for Algorithm~\ref{alg:fast_interpolation} to handle these extremal parameter cases more efficiently.

\subsection{Open Problems}

The complexity bound of the new algorithm for the vector operator interpolation problem (Problem~\ref{prob:general_interpolation_problem}) has an extra term $O(\ell m n^{\omega-1})$ operations in $\Fq$ if the first components of the interpolation points are not $\Fq$-linearly independent (cf.~Theorem~\ref{thm:fast_interpolation_correctness_complexity}).
This is due to the fact that we first need to bring the interpolation point matrix into a specific form, which is algorithmically done by transforming an $n \times (\ell+1)m$ matrix over $\Fq$ into reduced row echelon form.
Given the currently fastest skew-polynomial multiplication algorithms, the term $O(\ell m n^{\omega-1})$ operations in $\Fq$ is negligible compared to the term $\softO\!\left(\ell^\omega\OMul{n}\right)$ operations in the base field of the cost bound $\OMul{n}$.
At this point, however, it is not known whether skew-polynomial multiplication \emph{could} be sped up so this term is smallest for some parameters.
It is known that square matrix multiplication and skew-polynomial multiplication are softly equivalent (i.e.~$m^\omega \in \softO(\OMul{m})$ over $\Fq$ (if $\OMul{m}$ is expressed in operations in $\Fq$) and $\OMul{m} \in \softO(m^\omega)$ operations in $\Fq$, cf.~\cite{caruso2017fast,puchinger2017fast}), and answering the above question seem to require relating square matrix multiplication with low-degree skew-polynomial multiplication.

Though we are not aware of an application, it is quite natural to generalize the 2D vector remainder interpolation problem (Problem~\ref{prob:bivariate_remainder_interpolation}) to larger dimensions, analog to the vector operator interpolation problem (Problem~\ref{prob:general_interpolation_problem}).
If the first components of the evaluation points are $P$-independent, it appears to be straightforward to adapt the methods developed in Section~\ref{ssec:interpolation_speed-up} (faster vector operator interpolation) to the $(\ell+1)$ dimensional vector remainder evaluation case.
This corresponds to the special case that the first components of the interpolation points in Problem~\ref{prob:general_interpolation_problem} are $\Fq$-linearly independent.
It is not obvious how to solve the problem if the $P$-independence assumption is dropped.

\appendices

\section{Skew M-Basis Algorithm}\label{app:M-basis}

In this section, we present right and left skew analogs of the \textsf{M-Basis} algorithm \cite[\textsf{M-Basis}]{GiorgiPolyMatrix}.
The algorithms are asymptotically slower than the skew \textsf{PM-Basis} algorithms presented in Section~\ref{ssec:right_PM-basis_recursive_step}, but might be faster for small orders $d$ since their hidden constant is smaller as they do not rely on asymptotically fast skew polynomial arithmetic (cf.~Remark~\ref{rem:M-basis}).

\begin{algorithm}[ht!]
	\caption{$\textsf{RightSkewMBasis}$}\label{alg:rightSkewMBasis}
	\SetKwInOut{Input}{Input}\SetKwInOut{Output}{Output}
	\Input{\begin{itemize}
			\item positive integer $d\in\mathbb{Z}_{>0}$,
			\item matrix $\MABin\in\SkewPolys^{a\times b}$ of degree $<d$, 
			\item shifts $\s\in\mathbb{Z}^{b}$.
	\end{itemize}}
	
	\Output{$\B \in \RMABnameShortStandard$} %
	
	\BlankLine
	
	\If{d=1}{
		\Return{\emph{$\textsf{RightSkewBaseCase}$($\MABin,\s$)} \hfill \myAlgoComment{Algorithm~\ref{alg:rightSkewLinAppBas} in Section~\ref{sec:order_bases}}}
	}
	\Else{
		$\MABout_{1}\gets\textsf{RightSkewMBasis}\left(1,\MABin\reml x,\s\right)$ \\
		$\mat{G}\gets\left(x^{-1}\MABin\MABout_1\right)\reml x^{d-1}$; $\vec{t}\gets \cdeg_\s\left(\MABout_1\right)$ \label{line:mat_mult_right_MBasis} \\
		$\MABout_2\gets\textsf{RightSkewMBasis}\left(d-1,\mat{G},\vec{t}\right)$\\
		\Return{$\MABout_1\MABout_2$\label{line:multiply_result_right_MBasis}}
	}
\end{algorithm}

\begin{algorithm}[ht!]
	\caption{$\textsf{LeftSkewMBasis}$}\label{alg:leftSkewMBasis}
	\SetKwInOut{Input}{Input}\SetKwInOut{Output}{Output}
	\Input{\begin{itemize}
			\item positive integer $d\in\mathbb{Z}_{>0}$,
			\item matrix $\MABin\in\SkewPolys^{a\times b}$ of degree $<d$, 
			\item shifts $\s\in\mathbb{Z}^{a}$.
	\end{itemize}}
	
	\Output{$\B \in \LMABnameShortStandard$} %
	
	\BlankLine
	
	\If{d=1}{
		\Return{\emph{$\textsf{LeftSkewBaseCase}$($\MABin,\s$)} \hfill \myAlgoComment{Algorithm~\ref{alg:leftSkewLinAppBas} in Section~\ref{sec:order_bases}}}
	}
	\Else{
		$\MABout_{1}\gets\textsf{LeftSkewMBasis}\left(1,\MABin\remr x^,\s\right)$ \\
		$\mat{G}\gets\left(\MABout_1\MABin x^{-1}\right)\remr x^{d-1}$; $\vec{t}\gets \rdeg_\s\left(\MABout_1\right)$ \label{line:mat_mult_left_MBasis} \\
		$\MABout_2\gets\textsf{RightSkewMBasis}\left(d-1,\mat{G},\vec{t}\right)$\\
		\Return{$\MABout_2 \MABout_1$ \label{line:multiply_result_left_MBasis}}
	}
\end{algorithm}

\begin{theorem}
Algorithms~\ref{alg:rightSkewMBasis} and \ref{alg:leftSkewMBasis} are correct.
Algorithm~\ref{alg:rightSkewMBasis} has complexity
\begin{equation*}
\softO\big(\max\{a,b\}b^{\omega-1} d^2\big)
\end{equation*}
and Algorithm~\ref{alg:leftSkewMBasis} has complexity
\begin{equation*}
\softO\big(a^{\omega-1}\max\{a,b\} d^2\big)
\end{equation*}
operations in $\Fqm$.
\end{theorem}

\begin{proof}
Correctness follows from Lemma~\ref{lem:appBasisProduct}, as well as the correctness of the base cases (Theorem~\ref{thm:correctness_rightLinApp} for Algorithm~\ref{alg:rightSkewLinAppBas} and Theorem~\ref{thm:correctness_leftLinApp} for Algorithm~\ref{alg:leftSkewLinAppBas}).

The base cases, Algorithm~\ref{alg:rightSkewLinAppBas} for the left case and Algorithm~\ref{alg:leftSkewLinAppBas} are called exactly $d$ times.
In the right case,  Lines~\ref{line:mat_mult_right_MBasis} and \ref{line:multiply_result_right_MBasis} are executed exactly $d-1$ times. Since $\Q$ has degree $0$ and $\MABout_1$ has degree $1$ (see proof of Theorem~\ref{thm:correctness_rightLinApp}), the multiplication $x^{-1}\Q \MABout_1$ costs $O(\max\{a,b\}b^{\omega-1})$ operations in $\Fqm$ and the multiplication $\MABout_1 \MABout_2$ can be done in $O(\max\{a,b\}b^{\omega-1} d)$. Overall, this costs $O(\max\{a,b\}b^{\omega-1} d^2)$ over $\Fqm$. The left case follows analogously.
\end{proof}

\section{Examples}\label{app:examples}

Here, we present some examples that are mentioned in the paper.
Example~\ref{ex:left_right_bases_different} shows that we need to treat left and right approximant bases separately over skew polynomials (cf.~Section~\ref{ssec:skew_order_bases}). This is different to the case of commutative polynomial rings.

\begin{example}\label{ex:left_right_bases_different}
Consider the field $\mathbb{F}_{2^2}$ (represented by $\mathbb{F}_{2^2} = \mathbb{F}_2[b]/(b^2+1)$), with $\sigma = \Frob{2}$, and the following $2 \times 2$ matrix containing skew polynomials
\begin{align*}
\MABin = \begin{bmatrix}
(b + 1)x^{3} + bx & x^{3} + bx^{2} + (b + 1)x \\
(b + 1)x^{3} + bx^{2} + x + b & x^{3} + x^{2} + 1
\end{bmatrix}
\end{align*}
For $\s = [0,0]$ and $d=3$, a left and a right $\s$-minimal approximant basis of $\MABin$ of order $d$ are given as 
\begin{align*}
\MABout_\mathrm{left} &= 
\begin{bmatrix}
x^{2} & 0 \\
bx + b & x
\end{bmatrix} \in \SkewPolys^{2 \times 2} \quad \text{ and} \\ %
\MABout_\mathrm{right} &= \begin{bmatrix}
x^{2} + (b + 1)x & 1 \\
x & x + b
\end{bmatrix} \in \SkewPolys^{2 \times 2},
\end{align*}
respectively. However, we have
\begin{align*}
\MABin^\top \MABout_\mathrm{left}^\top \reml x^3 &=
\begin{bmatrix}
0 & (b + 1)x \\
0 & x^{2} + (b + 1)x
\end{bmatrix}, \\
\MABout_\mathrm{right}^\top \MABin^\top \remr x^3 &=
\begin{bmatrix}
0 & x^{2} + (b + 1)x \\
x^{2} + (b + 1)x & 0
\end{bmatrix}.
\end{align*}
Hence, in contrast to the ordinary polynomial ring $\Fqm[x]$, the matrix $\MABout_\mathrm{left}^\top$ is not a right $\s$-minimal approximant basis of $\MABin^\top$ of order $d$ and $\MABout_\mathrm{right}^\top$ is not a left $\s$-minimal approximant basis of $\MABin^\top$ of order $d$.
\end{example}

Example~\ref{ex:counterexample_reduction_high_orders} shows that, in contrast to matrices of degree $0$ and order $1$, right approximant bases over skew polynomials cannot be in general computed from ones over ordinary polynomial rings using the mapping $\varphi$ (cf.~\eqref{eq:phi_mapping}). See Remark~\ref{rem:PM_base_case_proof_strategy} in Section~\ref{ssec:fast_order_bases_computation} for more details.

\begin{example}\label{ex:counterexample_reduction_high_orders}
Consider the field $\mathbb{F}_{2^2}$ (represented by $\mathbb{F}_{2^2} = \mathbb{F}_2[b]/(b^2+1)$), with $\sigma = \Frob{2}$ and the matrix
\begin{align*}
\MABin = \begin{bmatrix}
(b + 1)x^{2} + (b + 1) & bx^{2} + bx + (b + 1) \\
x + b & x^{2} + bx + b
\end{bmatrix} \\
\in \SkewPolys^{2 \times 2} \ .
\end{align*}
We want to compute an approximant basis of $\MABin$ of order $2$ with respect to the shift vector $\s = [0,0]$ (i.e., unshifted).
First, we compute 
\begin{align*}
\MABinhat &= \varphi^{-1}(\MABin) \\
&= \begin{bmatrix}
\left(b + 1\right) x^{2} + b + 1 & b x^{2} + \left(b + 1\right) x + b + 1 \\
x + b & x^{2} + \left(b + 1\right) x + b
\end{bmatrix} \\
& \in \Fqm[x]^{2 \times 2},
\end{align*}
and, using the \textsf{PM-Basis} algorithm over $\Fqm[x]$ \cite{GiorgiPolyMatrix,neiger2016bases}, an $\s$-minimal approximant basis of order $2$ of $\MABinhat$ is,
\begin{align*}
\MABouthat &= \begin{bmatrix}
x + 1 & x \\
1 & x
\end{bmatrix} \in \Fqm[x]^{2 \times 2}.
\end{align*}
However, we have
\begin{align*}
\MABin \cdot \varphi\!\left(\MABouthat\right) &= \begin{bmatrix}
(b + 1)x^{3} + x^{2} + x & x^{3} + bx^{2} \\
x & x^{3} + (b + 1)x^{2}
\end{bmatrix} \\
&\equiv \begin{bmatrix}
x & 0 \\
x & 0
\end{bmatrix} \modl x^2 \ ,
\end{align*}
so the rows of $\varphi(\MABouthat)$ are not approximants of $\MABin$ of order $2$.
\end{example}

\section{Module Description of the Vector Operator Interpolation Problem}
\label{app:module_description}

In this section we show how to find a basis for the left $\SkewPolys$-module described by condition~\eqref{eq:interpolation_problem_eval} in Problem~\ref{prob:general_interpolation_problem}.
For notational convenience, we denote the $i$-th row of the input matrix $\U \in \Fqm^{n \times (\intOrder+1)}$ of the problem as $\u_i = [U_{i,1},\dots,U_{i,\ell+1}]$.
Recall that the $\u_i$ are called interpolation points.
We define the corresponding left $\SkewPolys$ module as 
\begin{align*}
 \module{\{\u_1,\dots,\u_n\}}:=\Big\{&[Q_0,\dots,Q_{\ell}]\in\SkewPolys^{\ell+1} \, : \\
 \, &\sum_{j=1}^{\ell+1} Q_{j-1}\!\left(U_{i,j}\right)=0, \forall i=1,\dots,n \Big\}.
\end{align*} 
A basis for $\module{\{\u_1,\dots,\u_{\nReceive}\}}$ allows to solve Problem~\ref{prob:general_interpolation_problem} using the row reduction methods from~\cite{puchinger2017row}. In this section, we show how to set up such a basis in general
thereby generalizing the special case of Problem~\ref{prob:general_interpolation_problem} discussed in \cite{puchinger2017row} (first column of $\U$ linearly independent).
The following results lays the foundations for constructing a basis for the interpolation module $\module{\{\u_1,\dots,\u_{n}\}}$ recursively.

Consider a matrix $\Z\in\Fqm^{n\times (\ell+1)}$ of the form
\begin{equation}
\Z = \left[
\def\arraystretch{1.7}
\begin{array}{ccccc}
	\multicolumn{5}{c}{\Z^{(1)}} \\\hline
	\multicolumn{2}{c|}{\0} & & \Z^{(*)} &   \\
\end{array}
\right]
\end{equation}
where $\Z^{(1)}\in\Fqm^{\nu\times (\ell+1)}$ with $z_{1,1}^{(1)},\dots,z_{\nu,1}^{(1)}$ being $\Fq$-linearly independent and $\Z^{(*)}\in\Fqm^{(n-\nu)\times \ell}$.
Denote by $\z_i$ and $\z^{(*)}_i$ the $i$-th row of $\Z$ and $\Z^{(*)}$, respectively.

\begin{proposition} 
  \label{prop:recurse_basis}
  If $\L \in \SkewPolys^{\ell \times \ell}$ is a (lower-triangular) basis for $\module{\{\z^{(*)}_{1},\ldots,\z^{(*)}_{n-\nu}\}} \subseteq \SkewPolys^{\ell}$, then the following matrix is a (lower-triangular) basis for $\module{\{\z_{1},\ldots,\z_n\}} \subseteq \SkewPolys^{(\ell+1)\times (\ell+1)}$:
  \[
    \M =
    \left[\begin{array}{c|c}
      G       & \hspace*{3em}                      \\
      \hline
      \begin{array}{c}
        R_1 \\
        \vdots \\
        R_{\ell}
      \end{array} & \L \\
    \end{array}\right]
    \ ,
  \]
  where 
  \begin{equation}
  	G \gets\MSPop_{\langle z_{1,1}^{(1)},\dots,z_{\nu,1}^{(1)}\rangle}
  \end{equation}
  and each $R_j$ is the interpolation skew polynomial given by:
  \[
    R_j(z_{i,1}^{(1)}) = -\L_j(z_{i,2}^{(1)},\ldots,z_{i,\ell+1}^{(1)}) \ , \quad i=1, \ldots, \nu \ ,
  \]
  where $\L_j$ is the $j$'th row of $\L$.
\end{proposition}

\begin{proof}
  We first show that the rows of $\M$ are in $\module{\{\z_{1},\ldots,\z_n\}}$.
  Clearly $G(z_{i,1}) = 0$ for all $i=1,\dots,n$.
  For $1 \leq i \leq \nu$, it is similarly obvious that $(R_j \mid \L_j)(\z_i) = 0$, so remaining is only to show $(R_j \mid \L_j)(\z_i) = 0$ for $i > \nu$.
  We have $(R_j \mid \L_j)(\z_i) = 0  \iff \L_j \in \module{\{\z^{(*)}_{1},\ldots,\z^{(*)}_{n-\nu}\}}$ which is true.

  To show that $\module{\{\z_{1},\ldots,\z_n\}}$ is in the row span of $\M$, take any $\vec Q = [Q_0,\ldots,Q_{\ell}] \in \module{\{\z_{1},\ldots,\z_n\}}$.
  We have that $[Q_2,\ldots,Q_{\ell}] \in \module{\{\z^{(*)}_{1},\ldots,\z^{(*)}_{n-\nu}\}}$, so there is a $\q \in \SkewPolys^{\ell}$ such that $[Q_1,\ldots,Q_{\ell}] = \q\L$.
  Since the rows of $\vec M$ are in $\module{\{\z_{1},\ldots,\z_n\}}$, so is the following vector:
  \[
    \vec Q' = \vec Q - (0 \mid \vec q) \vec M = (T, 0,\ldots,0) \ .
  \]
  Hence $T(z_{i,1}) = 0$ for $i=1,\ldots,\nu$, and so $T$ must be right-divisible by $G$.
\end{proof}

\begin{proposition}\label{prop:special_basis_case}
 Let $\L \in \SkewPolys^{\ell \times \ell}$ be a (lower-triangular) basis for $\module{\{\z^{(*)}_{1},\ldots,\z^{(*)}_{n-\nu}\}} \subseteq \SkewPolys^{\ell}$, then the following matrix is a (lower-triangular) basis for $\module{\{(0\,|\,\z^{(*)}_{1}),\ldots,(0\,|\,\z^{(*)}_{n-\nu})\}} \subseteq \SkewPolys^{(\ell+1)\times (\ell+1)}$:
 \begin{equation}
  \M =
    \left[\begin{array}{c|c}
      1       & \hspace*{3em}                      \\
      \hline
      \begin{array}{c}
        0 \\
        \vdots \\
        0
      \end{array} & \L \\
    \end{array}\right]
 \end{equation}
\end{proposition}

\begin{proof}
 We have that the first entries of the interpolation points are zero and thus not $\Fq$-linearly independent as in Proposition~\ref{prop:recurse_basis}. 
 However, the polynomials $G$ and $R_j$ from Proposition~\ref{prop:recurse_basis} are still well-defined.
 In particular, we have that $G\gets\MSPop_{\langle 0,\dots,0\rangle}=1$ and $R_j=0$ since $R_j(0)=-\L_j(z_{i,1}^{(*)},\dots,z_{i,\ell}^{(*)})=0$ for all $i=1,\dots,n-\nu$ and $j=1,\dots,\ell$. 
 Using similar arguments as in the proof of Proposition~\ref{prop:recurse_basis} we have that the rows of $\M$ vanish on all interpolation points $(0\,|\,\z^{(*)}_{1}),\ldots,(0\,|\,\z^{(*)}_{n-\nu})$ and form a basis for $\module{\{(0\,|\,\z_{1}),\ldots,(0\,|\,\z_n)\}} \subseteq \SkewPolys^{(\ell+1)\times (\ell+1)}$. 
\end{proof}

By applying the result of Proposition~\ref{prop:recurse_basis} and~\ref{prop:special_basis_case} recursively, we obtain Algorithm~\ref{alg:basis}.

\begin{remark}
 Note, that if the entries $u_{1,1}^{(1)},\dots,u_{n,1}^{(1)}$ are $\Fq$-linearly independent, the output of Algorithm~\ref{alg:basis} is a matrix $\M$ as given in~\cite[Lemma~5]{puchinger2017row} for decoding interleaved Gabidulin codes.
 Hence, Algorithm~\ref{alg:basis} handles the general case for constructing a basis for the interpolation module.
\end{remark}

\def\ModuleBasis{\ensuremath{\mathsf{ModuleBasis}}}
\begin{algorithm}[ht]
  \caption{$\ModuleBasis(\intOrder,\U)$}
  \label{alg:basis}
  \SetKwInOut{Input}{Input}\SetKwInOut{Output}{Output}
  \Input{$\intOrder \in \ZZ_{> 0}, \U\in\Fqm^{n\times (\intOrder+1)}$ containing the interpolation points $\u_1,\dots,\u_{n}$ as rows.}
  \Output{$\M \in \SkewPolys^{(\intOrder+1) \times (\intOrder+1)}$, a lower-triangular basis of $\module{\{\u_1,\dots,\u_{n}\}}$.} 
  \BlankLine
  Compute the matrix $\U'$, $\varrho$, $\nu_i$ and $a_i$ for all $i=1,\dots,\varrho$ as in Lemma~\ref{lem:intProblemSubspace_Zi_matrices} \label{line:setup_U_prime}\\
  $\M\gets\I_{(\intOrder-a_\varrho) \times (\intOrder-a_\varrho)}$ \\
  $cnt\gets\varrho$ \\
  \For{$i=1,\dots,a_\varrho+1$}{
  $\L\gets\M$ \\
  	\If{$a_\varrho-i+1\neq a_{cnt}$}{
  		$G \gets 1$ \\
  		$R_j \gets 0$ for all $j=1,\dots,\intOrder-a_\varrho+i-1$
  	}
  	\Else{
	   	$G \gets\MSPop_{\langle u_{1,1}^{(cnt)},\dots,u_{\nu_{cnt},1}^{(cnt)}\rangle}$ \label{line:anni_poly}\\
		$R_j \gets \IPop{\left\{\left(u_{\kappa,1}^{(cnt)}, \L_j(u_{\kappa,2}^{(cnt)},\dots,u_{\kappa,\intOrder+1-a_{cnt}}^{(cnt)})\right)\right\}_{\kappa=1}^{\nu_{cnt}}}$ where $\L_j$ denotes the $j$-th row of $\L$ for $j=1,\dots,\intOrder-a_{cnt}$ \label{line:int_poly}\\
		\begin{equation*}
		\M \gets
	        \left[\begin{array}{c|c}
	          G       & \hspace*{3em}                      \\
	          \hline
	          \begin{array}{c}
	            R_1 \\
	            \vdots \\
	            R_{\intOrder-a_{cnt}}
	          \end{array} & \L \\
	        \end{array}\right]
	        \ .
		\end{equation*}  
		$cnt\gets cnt-1$
	  }
  	}    
  \Return{$\M$}    
\end{algorithm}

\begin{theorem}
  Algorithm~\ref{alg:basis} is correct.
  It has computational complexity $\softO(\intOrder^2\OMul{n})$ 
	operations in the base field of the cost bound $\OMul{n}$ plus $O(\intOrder m n^{\omega-1})$ operations in $\Fq$.
\end{theorem}
\begin{proof}
  The correctness of the algorithm follows by applying Proposition~\ref{prop:recurse_basis} and Proposition~\ref{prop:special_basis_case} recursively.

According to Lemma~\ref{lem:intProblemSubspace_Zi_matrices} the computation $\U'$ in Line~\ref{line:setup_U_prime} requires $O\big(\ell m n^{\omega-1}\big)$ operations in $\Fq$.
In each of the $a_\rho+1\in O(\ell)$ steps we need to construct the annihilator polynomial $G$, which requires $\softO(\OMul{\nu_i})\in\softO(\OMul{n})$ operations in the base field of the cost bound. 
Line~\ref{line:int_poly} corresponds to a multi-point evaluation of a row of $\L$ at at most $n$ points, which requires $\softO(\intOrder\OMul{n})$ operations in the base field of the cost bound, and the construction of the interpolation polynomials which requires $\softO(\intOrder\OMul{n})$ operations in the base field of the cost bound.
Hence, the Algorithm requires at most $\softO(\intOrder^2\OMul{n})$ operations in the base field of the cost bound plus $O(\intOrder m n^{\omega-1})$ operations in $\Fq$.

\end{proof}

\bibliographystyle{IEEEtran}
\bibliography{main}

\begin{IEEEbiographynophoto}{Hannes Bartz}
(S'14-M'16) was born in Trostberg, Germany, in 1985.
He received his Dipl.-Ing. and Dr.-Ing. degree from the Technical University of Munich, Germany, in 2010 and 2017, respectively.
In his dissertation (supervised by Prof. Gerhard Kramer) he developed efficient algebraic decoding schemes for error-correcting codes in subspace and rank metric.
In July 2017 he joined the Information Transmission Group within the Institute of Communications and Navigation at the German Aerospace Center (DLR).
His main research interests are code-based post-quantum cryptography and algebraic coding theory.
In 2018 he has been appointed as a Lecturer at the Institute for Communications Engineering (LNT), Technical University of Munich, Germany.
He received the Prof. Dr. Ralf Kötter memorial award in 2012.
\end{IEEEbiographynophoto}

\begin{IEEEbiographynophoto}{Thomas Jerkovits}
received the B.Sc. degree in electrical engineering from Ulm University (UUlm), Ulm, Germany and the M.Sc. degree in electrical engineering from Technical University of Munich (TUM), Munich, Germany in 2013 and 2015, respectively. He currently is working at the German Aerospace Center (DLR) as a member of the Quantum Resistance Cryptography Group. He is also pursuing the doctoral degree at the Institute for Communications Engineering of TUM.
\end{IEEEbiographynophoto}

\begin{IEEEbiographynophoto}{Sven Puchinger}
(S'14, M'19) is a postdoctoral researcher at the Technical University of Munich (TUM), Germany. He received the B.Sc. degree in electrical engineering and the B.Sc. degree in mathematics from Ulm University, Germany, in 2012 and 2016, respectively. During his studies, he spent two semesters at the University of Toronto, Canada. He received his Ph.D. degree from the Institute of Communications Engineering, Ulm University, Germany, in 2018. He has been a postdoc at the Technical University of Munich (2018--2019 and since 2021) and the Technical University of Denmark (2019--2021), Denmark. His research interests are coding theory, its applications, and related computer-algebra methods.
\end{IEEEbiographynophoto}

\begin{IEEEbiographynophoto}{Johan Rosenkilde}
holds a Master’s degree in computer science (2010) and
PhD in mathematics (2013), both from the Technical University of
Denmark. He was then a post-doc at both Ulm University, Germany and at
Inria Saclay, France. From 2015-2021 he was at the Technical University
of Denmark, first as assistant professor then as associate professor. He
is now a Research Engineer at GitHub. His algebraic research interests
include coding theory and computer algebra.
\end{IEEEbiographynophoto}

\end{document}